\documentclass[12pt,notitlepage]{amsart}%
\usepackage{amssymb}
\usepackage{amsfonts}
\usepackage{graphicx}
\usepackage{amscd}
\usepackage{graphicx}
\usepackage{amsmath}
\usepackage{hyperref}%
\newtheorem{theorem}{Theorem}
\theoremstyle{plain}

\newtheorem{corollary}{Corollary}

\newtheorem{definition}{Definition}
\newtheorem{example}{Example}

\newtheorem{lemma}{Lemma}
\newtheorem{notation}{Notation}

\newtheorem{proposition}{Proposition}
\newtheorem{remark}{Remark}

\numberwithin{equation}{section}
\numberwithin{theorem}{section}
\numberwithin{lemma}{section}
\numberwithin{proposition}{section}
\numberwithin{corollary}{section}

\ifx\pdfoutput\relax\let\pdfoutput=\undefined\fi
\newcount\msipdfoutput
\ifx\pdfoutput\undefined\else
\ifcase\pdfoutput\else
\ifx\paperwidth\undefined\else
\ifdim\paperheight=0pt\relax\else\pdfpageheight\paperheight\fi
\ifdim\paperwidth=0pt\relax\else\pdfpagewidth\paperwidth\fi
\fi\fi\fi
\begin{document}
\title{Non-Archimedean Statistical Field Theory}
\author{W. A. Z\'{u}\~{n}iga-Galindo}
\address{University of Texas Rio Grande Valley\\
School of Mathematical \& Statistical Sciences\\
One West University Blvd\\
Brownsville, TX 78520, United States.}
\email{wilson.zunigagalindo@utrgv.edu.}
\thanks{The author was partially supported by the Debnath Endowed Professorship.}

\begin{abstract}
We construct in a rigorous mathematical way interacting quantum field theories
on a $p$-adic spacetime. The main result is the construction of a measure on a
function space which allows a rigorous definition of the partition function.
The advantage of the approach presented here is that all the perturbation
calculations can be carried out in the standard way using functional
derivatives, but in a mathematically rigorous way.

\end{abstract}
\date{18/04/2022}
\maketitle
\tableofcontents

\section{Introduction}

In this article we construct (in a rigorous mathematical way) interacting
quantum field theories over a $p$-adic spacetime in an arbitrary dimension. We
provide a large family of energy functionals $E(\varphi,J)$ admitting natural
discretizations in finite-dimensional vector spaces such that the partition
function
\begin{equation}
Z^{\text{phys}}(J)=\int D(\varphi)e^{-\frac{1}{K_{B}T}E(\varphi,J)}
\label{Eq_0}%
\end{equation}
can be defined rigorously as the limit of the mentioned discretizations. Our
main result is the construction of a measure on a function space such that
(\ref{Eq_0}) makes mathematical sense, and the calculations \ of the $n$-point
correlation functions can be carried out using perturbation expansions via
functional derivatives, in a rigorous mathematical way. Our results include
$\varphi^{4}$-theories. In this case, we show that the $n$-point correlation
functions admit a convergent expansion in the coupling parameter in a certain
space of distributions. By the Wick theorem all of the distributions appearing
in the mentioned series can be expressed as a sum of product of Green
functions, which have singularities. Consequently a renormalization procedure
is required, we expect \ to study the renormalization of the Feynman integrals
attached to (\ref{Eq_0}) in a forthcoming publication.

From now on $p$ denotes a fixed prime number. A $p$-adic number is a series of
the form%
\begin{equation}
x=x_{-k}p^{-k}+x_{-k+1}p^{-k+1}+\ldots+x_{0}+x_{1}p+\ldots,\text{ with }%
x_{-k}\neq0\text{,} \label{p-adic-number}%
\end{equation}
where the $x_{j}$s \ are $p$-adic digits, i.e. numbers in the set $\left\{
0,1,\ldots,p-1\right\}  $. The set of all possible series of the form
(\ref{p-adic-number}) constitutes the field of $p$-adic numbers $\mathbb{Q}%
_{p}$. There are natural field operations, sum and multiplication, on series
of the form (\ref{p-adic-number}), see e.g. \cite{Koblitz}. There is also a
natural norm in $\mathbb{Q}_{p}$ defined as $\left\vert x\right\vert
_{p}=p^{k}$, for a nonzero $p$-adic number of the form (\ref{p-adic-number}).
The field of $p$-adic numbers with the distance induced by $\left\vert
\cdot\right\vert _{p}$ is a complete ultrametric space. The ultrametric (or
non-Archimedean) property refers to the fact that $\left\vert x-y\right\vert
_{p}\leq\max\left\{  \left\vert x-z\right\vert _{p},\left\vert z-y\right\vert
_{p}\right\}  $ for any $x$, $y$, $z\in\mathbb{Q}_{p}$. We denote by
$\mathbb{Z}_{p}$ the unit ball, which consists of all series with expansions
of the form (\ref{p-adic-number}) with $-k\geq0$. We extend the $p-$adic norm
to $\mathbb{Q}_{p}^{N}$ by taking $||x||_{p}=\max_{1\leq i\leq N}|x_{i}|_{p}$,
for $x=(x_{1},\dots,x_{N})\in\mathbb{Q}_{p}^{N}$.

A fundamental scientific problem is the understanding of the structure of
space-time at the level of the Planck scale, and the construction of
physical-mathematical models of it. This problem occurs naturally when trying
to unify general relativity and quantum mechanics. In the 1930s Bronstein
showed that general relativity and quantum mechanics imply that the
uncertainty $\Delta x$ of any length measurement satisfies $\Delta x\geq
L_{\text{Planck}}:=\sqrt{\frac{\hbar G}{c^{3}}}$, where $L_{\text{Planck}}$ is
the Planck length ($L_{\text{Planck}}\approx10^{-33}$ $cm$). This implies that
space-time is not an infinitely divisible continuum (mathematically speaking,
the spacetime must be a completely disconnected topological space at the level
of the Planck scale). Bronstein's \ inequality has motivated the development
of several different physical theories. At any rate, this inequality implies
the need of using non-Archimedean mathematics in models dealing with the
Planck scale. In the 1980s, Volovich proposed the conjecture that the
space-time at the Planck scale is non-Archimedean, see \cite{Volovich1}. This
conjecture has propelled a wide variety of investigations in cosmology,
quantum mechanics, string theory, QFT, etc., and the influence of this
conjecture is still relevant nowadays, see e.g. \cite{Abdesselam},
\cite{Bocardo-Zuniga}-\cite{GC-Zuniga}, \cite{Gubser et al.}-\cite{Harlow et
al}, \cite{Kochubei et al}-\cite{KKZuniga}, \cite{LM89}-\cite{Mis2},
\cite{V-V-Z}-\cite{Zuniga-LNM-2016}.

The space $\mathbb{Q}_{p}^{N}$ has a very rich mathematical structure. The
axiomatic quantum field \ theory can be extended to $\mathbb{Q}_{p}^{N}$. In
\cite{Mendoza-Zuniga}, we construct a family of quantum scalar fields over a
$p-$adic spacetime which satisfy $p-$adic analogues of the
G\aa rding--Wightman axioms. Since the space of test functions on
$\mathbb{Q}_{p}^{N}$ is nuclear the techniques of white noise calculus are
available in the $p$-adic setting, see e.g. \cite{Ber-Kon},
\cite{Gelfand-Vilenkin}, \cite{Huang-Yang}, \cite{Hida et al}. This implies
that a rigorous functional integral approach is available in the $p$-adic
framework, see e.g. \cite{Glimm-Jaffe}, \cite{Simon-0}, \cite{Simon-1}. In
\cite{Zuniga-JFAA}, see also \cite[Chapter 11]{KKZuniga},
\cite{Albeverio-et-al-3}-\cite{Albeverio-et-al-4}, we introduced a class of
non-Archimedean massive Euclidean fields, in arbitrary dimension, which are
constructed as solutions of certain covariant $p$-adic stochastic
pseudo-differential equations, by using techniques of white noise calculus. In
\cite{Arroyo-Zuniga}, we construct a large class of interacting Euclidean
quantum field theories, over a $p$-adic space time, by using white noise
calculus. These quantum fields fulfill all the Osterwalder-Schrader axioms,
except the reflection positivity. In all these theories the time is a $p$-adic
variable. Since $\mathbb{Q}_{p}$ is not an ordered field, there is no notion
of past and future. In certain theories, it is possible to introduce a
quadratic form. The orthogonal group of this form plays the role of Lorentz
group. Anyway, we do not have a light cone structure, and then this type of
theory is also acausal, see \cite{Mendoza-Zuniga}. The relevant feature is
that the vacuum of\ all these theories performs fluctuations.

In the case of $\varphi^{4}$-theories the energy functional $E(\varphi,0)$
takes the form%
\begin{align}
E(\varphi,0;\delta,\gamma,\alpha_{2},\alpha_{4}) &  =\frac{\gamma}{2}\text{ }%
{\textstyle\int\limits_{\mathbb{Q}_{p}^{N}}}
\varphi\left(  x\right)  \boldsymbol{W}\left(  \partial,\delta\right)
\varphi\left(  x\right)  d^{N}x+\frac{\alpha_{2}}{2}\text{ }%
{\textstyle\int\limits_{\mathbb{Q}_{p}^{N}}}
\varphi^{2}\left(  x\right)  d^{N}x\nonumber\\
&  +\frac{\alpha_{4}}{2}%
{\textstyle\int\limits_{\mathbb{Q}_{p}^{N}}}
\varphi^{4}\left(  x\right)  d^{N}x,\label{Eq_energy}%
\end{align}
where $\varphi:\mathbb{Q}_{p}^{N}\rightarrow\mathbb{R}$ is a test function
($\varphi\in\mathcal{D}_{\mathbb{R}}\left(  \mathbb{Q}_{p}^{N}\right)  $),
$\delta>N$, $\gamma>0$, $\alpha_{2}\geq0$, $\alpha_{4}\geq0$, and
$\boldsymbol{W}\left(  \partial,\delta\right)  \varphi\left(  x\right)
=\mathcal{F}_{\kappa\rightarrow x}^{-1}(A_{w_{\delta}}(\left\Vert
\kappa\right\Vert )\mathcal{F}_{x\rightarrow\kappa}\varphi)$ is
pseudo-differential operator, whose symbol has a singularity at the origin.

An interesting observation is that the one-dimensional Vladimirov operator is
a special case of the operators $\boldsymbol{W}\left(  \partial,\delta\right)
$, in this case the action $E(\varphi,0;\delta,\gamma,0,0)$ appeared in
$p$-adic string theory, see \cite{Spokoiny}, \cite{Zhang}, \cite{Zabrodin},
see also \cite{GC-Zuniga} and the references therein.

In order to make sense of the partition function attached to $E(\varphi
,0;\delta,\gamma,\alpha_{2},\alpha_{4})$, see (\ref{Eq_0}), we discretize the
fields like in classical QFT. As fields we use test functions $\varphi\in$
$\mathcal{D}_{\mathbb{R}}\left(  \mathbb{Q}_{p}^{N}\right)  $, which are
locally constant with compact support. We have $\mathcal{D}_{\mathbb{R}%
}\left(  \mathbb{Q}_{p}^{N}\right)  =\cup_{l=1}^{\infty}\mathcal{D}%
_{\mathbb{R}}^{l}\left(  \mathbb{Q}_{p}^{N}\right)  $, where $\mathcal{D}%
_{\mathbb{R}}^{l}\left(  \mathbb{Q}_{p}^{N}\right)  \simeq\mathbb{R}^{\#G_{l}%
}$ is a real, finite dimensional vector space consisting of test functions
supported in the ball $B_{l}^{N}=\left\{  x\in\mathbb{Q}_{p}^{N};\left\Vert
x\right\Vert _{p}\leq p^{l}\right\}  $ having the form%
\begin{equation}
\varphi\left(  x\right)  =%
{\textstyle\sum\limits_{\boldsymbol{i}\in G_{l}}}
\varphi\left(  \boldsymbol{i}\right)  \Omega\left(  p^{l}\left\Vert
x-\boldsymbol{i}\right\Vert _{p}\right)  \text{, \ }\varphi\left(
\boldsymbol{i}\right)  \in\mathbb{R}\text{,}\label{Eq_0_1}%
\end{equation}
where $G_{l}$ is a finite set of indices and $\Omega\left(  p^{l}\left\Vert
x-\boldsymbol{i}\right\Vert _{p}\right)  $ is the characteristic function of
the ball $B_{-l}^{N}\left(  \boldsymbol{i}\right)  =\left\{  x\in
\mathbb{Q}_{p}^{N};\left\Vert x-\boldsymbol{i}\right\Vert _{p}\leq
p^{-l}\right\}  $. Now a natural discretization of partition function
$\mathcal{Z}^{\left(  l\right)  }$ is obtained by restricting the fields to
$\mathcal{D}_{\mathbb{R}}^{l}\left(  \mathbb{Q}_{p}^{N}\right)  \simeq
\mathbb{R}^{\#G_{l}}$ as follows. By identifying $\varphi$ with the column
vector $\left[  \varphi\left(  \boldsymbol{i}\right)  \right]
_{\boldsymbol{i}\in G_{l}}$, one obtains that
\[
E(\varphi,0;\delta,\gamma,\alpha_{2},0)=%
{\textstyle\sum\limits_{\boldsymbol{i},\boldsymbol{j}\in G_{l}}}
p^{-lN}U_{\boldsymbol{i},\boldsymbol{j}}(l)\varphi\left(  \boldsymbol{i}%
\right)  \varphi\left(  \boldsymbol{j}\right)  ,
\]
is a quadratic form in $\left[  \varphi\left(  \boldsymbol{i}\right)  \right]
_{\boldsymbol{i}\in G_{l}}$, cf. Lemma \ref{Lemma6}, and thus taking
$K_{B}T=1$, it is natural to propose that
\[
\mathcal{Z}^{\left(  l\right)  }=\int D_{l}(\varphi)e^{-E(\varphi
,0;\delta,\gamma,\alpha_{2},0)}\overset{\text{def.}}{=}%
{\textstyle\int\limits_{\mathbb{R}^{\#G_{l}}}}
e^{-%
{\textstyle\sum\limits_{\boldsymbol{i},\boldsymbol{j}\in G_{l}}}
p^{-lN}U_{\boldsymbol{i},\boldsymbol{j}}(l)\varphi\left(  \boldsymbol{i}%
\right)  \varphi\left(  \boldsymbol{j}\right)  }%
{\textstyle\prod\limits_{\boldsymbol{i}\in G_{l}}}
d\varphi\left(  \boldsymbol{i}\right)  ,
\]
where $%
{\textstyle\prod\nolimits_{\boldsymbol{i}\in G_{l}}}
d\varphi\left(  \boldsymbol{i}\right)  $ is the Lebesgue measure on
$\mathbb{R}^{\#G_{l}}$, which is a finite dimensional Gaussian integral. We
denote the corresponding Gaussian measure as $\mathbb{P}_{l}$. The next step
is to show the existence of a probability measure $\mathbb{P}$ such that
$\mathbb{P=}\lim_{l\rightarrow\infty}\mathbb{P}_{l}$ `in some sense'. This
requires passing to the momenta space and using the Lizorkin space
$\mathcal{L}_{\mathbb{R}}\left(  \mathbb{Q}_{p}^{N}\right)  \subset
\mathcal{D}_{\mathbb{R}}\left(  \mathbb{Q}_{p}^{N}\right)  $, resp.
$\mathcal{L}_{\mathbb{R}}^{l}\left(  \mathbb{Q}_{p}^{N}\right)  \subset
\mathcal{D}_{\mathbb{R}}^{l}\left(  \mathbb{Q}_{p}^{N}\right)  $. The key
point is that the operator
\[
\frac{\gamma}{2}\boldsymbol{W}\left(  \partial,\delta\right)  +\frac
{\alpha_{2}}{2}:\mathcal{L}_{\mathbb{R}}\left(  \mathbb{Q}_{p}^{N}\right)
\rightarrow\mathcal{L}_{\mathbb{R}}\left(  \mathbb{Q}_{p}^{N}\right)
\]
has an inverse in $\mathcal{L}_{\mathbb{R}}\left(  \mathbb{Q}_{p}^{N}\right)
$ for any $\alpha_{2}\geq0$. The construction of the measure $\mathbb{P}$ is
made in two steps. In the first step, by using Kolmogorov's consistency
theorem, one shows the existence of a unique probability measure $\mathbb{P}$
in $\mathbb{R}^{\infty}\cup\left\{  \text{point}\right\}  $ such any linear
functional $f\rightarrow\int_{\mathcal{L}_{\mathbb{R}}^{l}\left(
\mathbb{Q}_{p}^{N}\right)  }fd\mathbb{P}_{l}$, where $f$ is a continuous
bounded function in $\mathcal{L}_{\mathbb{R}}^{l}\left(  \mathbb{Q}_{p}%
^{N}\right)  $, has unique extension of the form $\int_{\mathcal{L}%
_{\mathbb{R}}^{l}\left(  \mathbb{Q}_{p}^{N}\right)  }fd\mathbb{P}_{l}%
=\int_{\mathcal{L}_{\mathbb{R}}^{l}\left(  \mathbb{Q}_{p}^{N}\right)
}fd\mathbb{P}$, cf. Lemma \ref{Lemma11}. In the second step by using the
\ Gel'fand triple $\mathcal{L}_{\mathbb{R}}\left(  \mathbb{Q}_{p}^{N}\right)
\hookrightarrow L_{\mathbb{R}}^{2}\left(  \mathbb{Q}_{p}^{N}\right)
\hookrightarrow\mathcal{L}_{\mathbb{R}}^{\prime}\left(  \mathbb{Q}_{p}%
^{N}\right)  $, where $\mathcal{L}_{\mathbb{R}}^{\prime}\left(  \mathbb{Q}%
_{p}^{N}\right)  $ is the topological dual of $\mathcal{L}_{\mathbb{R}}\left(
\mathbb{Q}_{p}^{N}\right)  $, and the Bochner-Minlos theorem, there exists a
probability measure $\mathbb{P}$ on $\left(  \mathcal{L}_{\mathbb{R}}^{\prime
}\left(  \mathbb{Q}_{p}^{N}\right)  ,\mathcal{B}\right)  $, \ that coincides
with the probability measure constructed in the first step, cf. Theorem
\ref{Theorem1}.

For an interaction energy $E_{\text{int}}(\varphi)$ satisfying $\exp\left(
-E_{\text{int}}(\varphi)\right)  \leq1$, it verifies that
\[%
{\displaystyle\int\nolimits_{\mathcal{L}_{\mathbb{R}}^{l}\left(
\mathbb{Q}_{p}^{N}\right)  }}
e^{-E_{\text{int}}\left(  \varphi\right)  }d\mathbb{P}_{l}=%
{\displaystyle\int\nolimits_{\mathcal{L}_{\mathbb{R}}^{l}\left(
\mathbb{Q}_{p}^{N}\right)  }}
e^{-E_{\text{int}}\left(  \varphi\right)  }d\mathbb{P\rightarrow}%
{\displaystyle\int\nolimits_{\mathcal{L}_{\mathbb{R}}\left(  \mathbb{Q}%
_{p}^{N}\right)  }}
e^{-E_{\text{int}}\left(  \varphi\right)  }d\mathbb{P}%
\]
as $l\rightarrow\infty$. Then a $\mathcal{P}\left(  \varphi\right)  $-theory
is given by a cylinder probability measure of the form%
\begin{equation}
\frac{1_{\mathcal{L}_{\mathbb{R}}}\left(  \varphi\right)  e^{-E_{\text{int}%
}\left(  \varphi\right)  }d\mathbb{P}}{%
{\displaystyle\int\nolimits_{\mathcal{L}_{\mathbb{R}}\left(  \mathbb{Q}%
_{p}^{N}\right)  }}
e^{-E_{\text{int}}\left(  \varphi\right)  }d\mathbb{P}} \label{measure_1}%
\end{equation}
in the space of fields $\mathcal{L}_{\mathbb{R}}\left(  \mathbb{Q}_{p}%
^{N}\right)  $. Notice that $\mathbb{P}$\ is a probability measure on
$\mathcal{L}_{\mathbb{R}}^{\prime}\left(  \mathbb{Q}_{p}^{N}\right)  $, but
due to the factor $1_{\mathcal{L}_{\mathbb{R}}}\left(  \varphi\right)  $ our
fields are test functions, and not distributions as in \cite{Arroyo-Zuniga},
\cite{Kochubei et al}, see also \cite{Albeverio-et-al-3}%
-\cite{Albeverio-et-al-4}, \cite{GS1999}, and the references therein. Then,
the Wick operator $:\cdot:$\ (or Wick regularization) is not required in the
definition of $E_{\text{int}}\left(  \varphi\right)  $. This is a very
relevant difference with respect to \cite{Arroyo-Zuniga}, \cite{Kochubei et
al}. Here we consider polynomial interactions. The advantage of the approach
presented here is that all the perturbation calculations can be carried out in
the standard way using functional derivatives, but in a mathematically
rigorous way, see Theorem \ref{Theorem2}. However, a renormalization procedure
is required. In \cite{Arroyo-Zuniga} we construct probability measures for
general, interacting QFTs, but using Hida-Kondratiev spaces, which are more
bigger than the spaces of distributions used here. However, doing explicit
calculations in this very general framework is not easy.

The mathematical framework presented here allows the construction of
complex-valued measures of type%
\[
\frac{1_{\mathcal{L}_{\mathbb{R}}}\left(  \varphi\right)  \exp\sqrt
{-1}\left\{  \frac{\alpha_{4}}{2}%
{\textstyle\int\limits_{\mathbb{Q}_{p}^{N}}}
\varphi^{4}\left(  x\right)  d^{N}x+%
{\textstyle\int\limits_{\mathbb{Q}_{p}^{N}}}
J(x)\varphi\left(  x\right)  d^{N}x\right\}  }{%
{\displaystyle\int\nolimits_{\mathcal{L}_{\mathbb{R}}\left(  \mathbb{Q}%
_{p}^{N}\right)  }}
\exp\sqrt{-1}\left\{  \frac{\alpha_{4}}{2}%
{\textstyle\int\limits_{\mathbb{Q}_{p}^{N}}}
\varphi^{4}\left(  x\right)  d^{N}x\right\}  d\mathbb{P}}d\mathbb{P}\text{.}%
\]
Furthermore all the corresponding perturbation expansions can be carried out
in the standard form. These measures are obtained from measures of type
(\ref{measure_1}) by performing a Wick rotation of type $\varphi
\rightarrow\sqrt{-1}\varphi$, see Section \ref{Section_Wick_rotation}. The
novelty is that this Wick rotation is not performed in spacetime, and thus
\ all these quantum field theories are acausal. More precisely, special
relativity is not valid in the spacetime of these theories. However, the
vacuum of all these theories perform thermal (resp. quantum) fluctuations,
because the Feynman rules are valid, at least formally, in these theories.

The energy functional $E(\varphi,J;\delta,\gamma,\alpha_{2},\alpha_{4})$,
$\varphi\in\mathcal{D}_{\mathbb{R}}^{l}\left(  \mathbb{Q}_{p}^{N}\right)  $,
see (\ref{Eq_energy}), can be interpreted as the Hamiltonian of a continuous
Ising model in the ball $B_{l}^{N}$ with an external magnetic field $J$. The
Landau-Ginzburg energy functional $E(\varphi,0;\delta,\gamma,\alpha_{2}%
,\alpha_{4})$ is non-local, i.e. only long range interactions occur,
furthermore, it has $\boldsymbol{Z}_{2}$ symmetry ($\varphi\rightarrow
-\varphi$). Finally, all the results presented in this article are valid if
$\mathbb{Q}_{p}$ is replaced by any non-Archimedean local field.

\section{\label{Section1} Basic facts on $p$-adic analysis}

In this section we fix the notation and collect some basic results on $p$-adic
analysis that we will use through the article. For a detailed exposition on
$p$-adic analysis the reader may consult \cite{A-K-S}, \cite{Taibleson},
\cite{V-V-Z}.

\subsection{The field of $p$-adic numbers}

Throughout this article $p$ will denote a prime number. Since we have to deal
with quadratic forms, for the sake of simplicity, we assume that $p\geq3$
throughout the article. The field of $p-$adic numbers $\mathbb{Q}_{p}$ is
defined as the completion of the field of rational numbers $\mathbb{Q}$ with
respect to the $p-$adic norm $|\cdot|_{p}$, which is defined as
\[
|x|_{p}=%
\begin{cases}
0 & \text{if }x=0\\
p^{-\gamma} & \text{if }x=p^{\gamma}\dfrac{a}{b},
\end{cases}
\]
where $a$ and $b$ are integers coprime with $p$. The integer $\gamma
=ord_{p}(x):=ord(x)$, with $ord(0):=+\infty$, is called the\textit{\ }%
$p-$\textit{adic order of} $x$. We extend the $p-$adic norm to $\mathbb{Q}%
_{p}^{N}$ by taking%
\[
||x||_{p}:=\max_{1\leq i\leq N}|x_{i}|_{p},\qquad\text{for }x=(x_{1}%
,\dots,x_{N})\in\mathbb{Q}_{p}^{N}.
\]
We define $ord(x)=\min_{1\leq i\leq N}\{ord(x_{i})\}$, then $||x||_{p}%
=p^{-ord(x)}$.\ The metric space $\left(  \mathbb{Q}_{p}^{N},||\cdot
||_{p}\right)  $ is a complete ultrametric space. As a topological space
$\mathbb{Q}_{p}$\ is homeomorphic to a Cantor-like subset of the real line,
see e.g. \cite{A-K-S}, \cite{V-V-Z}.

Any $p-$adic number $x\neq0$ has a unique expansion of the form
\[
x=p^{ord(x)}\sum_{j=0}^{\infty}x_{j}p^{j},
\]
where $x_{j}\in\{0,1,2,\dots,p-1\}$ and $x_{0}\neq0$. By using this expansion,
we define \textit{the fractional part }$\{x\}_{p}$\textit{ of }$x\in
\mathbb{Q}_{p}$ as the rational number
\[
\{x\}_{p}=%
\begin{cases}
0 & \text{if }x=0\text{ or }ord(x)\geq0\\
p^{ord(x)}\sum_{j=0}^{-ord(x)-1}x_{j}p^{j} & \text{if }ord(x)<0.
\end{cases}
\]
In addition, any $x\in\mathbb{Q}_{p}^{N}\smallsetminus\left\{  0\right\}  $
can be represented uniquely as $x=p^{ord(x)}v\left(  x\right)  $ where
$\left\Vert v\left(  x\right)  \right\Vert _{p}=1$.

\subsection{Topology of $\mathbb{Q}_{p}^{N}$}

For $r\in\mathbb{Z}$, denote by $B_{r}^{N}(a)=\{x\in\mathbb{Q}_{p}%
^{N};||x-a||_{p}\leq p^{r}\}$ \textit{the ball of radius }$p^{r}$ \textit{with
center at} $a=(a_{1},\dots,a_{N})\in\mathbb{Q}_{p}^{N}$, and take $B_{r}%
^{N}(0):=B_{r}^{N}$. Note that $B_{r}^{N}(a)=B_{r}(a_{1})\times\cdots\times
B_{r}(a_{N})$, where $B_{r}(a_{i}):=\{x\in\mathbb{Q}_{p};|x_{i}-a_{i}|_{p}\leq
p^{r}\}$ is the one-dimensional ball of radius $p^{r}$ with center at
$a_{i}\in\mathbb{Q}_{p}$. The ball $B_{0}^{N}$ equals the product of $N$
copies of $B_{0}=\mathbb{Z}_{p}$, \textit{the ring of }$p-$\textit{adic
integers}. We also denote by $S_{r}^{N}(a)=\{x\in\mathbb{Q}_{p}^{N}%
;||x-a||_{p}=p^{r}\}$ \textit{the sphere of radius }$p^{r}$ \textit{with
center at} $a=(a_{1},\dots,a_{N})\in\mathbb{Q}_{p}^{N}$, and take $S_{r}%
^{N}(0):=S_{r}^{N}$. We notice that $S_{0}^{1}=\mathbb{Z}_{p}^{\times}$ (the
group of units of $\mathbb{Z}_{p}$), but $\left(  \mathbb{Z}_{p}^{\times
}\right)  ^{N}\subsetneq S_{0}^{N}$. The balls and spheres are both open and
closed subsets in $\mathbb{Q}_{p}^{N}$. In addition, two balls in
$\mathbb{Q}_{p}^{N}$ are either disjoint or one is contained in the other.

As a topological space $\left(  \mathbb{Q}_{p}^{N},||\cdot||_{p}\right)  $ is
totally disconnected, i.e. the only connected \ subsets of $\mathbb{Q}_{p}%
^{N}$ are the empty set and the points. A subset of $\mathbb{Q}_{p}^{N}$ is
compact if and only if it is closed and bounded in $\mathbb{Q}_{p}^{N}$, see
e.g. \cite[Section 1.3]{V-V-Z}, or \cite[Section 1.8]{A-K-S}. The balls and
spheres are compact subsets. Thus $\left(  \mathbb{Q}_{p}^{N},||\cdot
||_{p}\right)  $ is a locally compact topological space.

Since $(\mathbb{Q}_{p}^{N},+)$ is a locally compact topological group, there
exists a Haar measure $d^{N}x$, which is invariant under translations, i.e.
$d^{N}(x+a)=d^{N}x$. If we normalize this measure by the condition
$\int_{\mathbb{Z}_{p}^{N}}dx=1$, then $d^{N}x$ is unique.

\begin{notation}
We will use $\Omega\left(  p^{-r}||x-a||_{p}\right)  $ to denote the
characteristic function of the ball $B_{r}^{N}(a)$. For more general sets, we
will use the notation $1_{A}$ for the characteristic function of a set $A$.
\end{notation}

\subsection{The Bruhat-Schwartz space}

A complex-valued function $\varphi$ defined on $\mathbb{Q}_{p}^{N}$ is
\textit{called locally constant} if for any $x\in\mathbb{Q}_{p}^{N}$ there
exist an integer $l(x)\in\mathbb{Z}$ such that%
\begin{equation}
\varphi(x+x^{\prime})=\varphi(x)\text{ for any }x^{\prime}\in B_{l(x)}^{N}.
\label{local_constancy}%
\end{equation}
A function $\varphi:\mathbb{Q}_{p}^{N}\rightarrow\mathbb{C}$ is called a
\textit{Bruhat-Schwartz function (or a test function)} if it is locally
constant with compact support. Any test function can be represented as a
linear combination, with complex coefficients, of characteristic functions of
balls. The $\mathbb{C}$-vector space of Bruhat-Schwartz functions is denoted
by $\mathcal{D}(\mathbb{Q}_{p}^{N}):=\mathcal{D}$. We denote by $\mathcal{D}%
_{\mathbb{R}}(\mathbb{Q}_{p}^{N}):=\mathcal{D}_{\mathbb{R}}$\ the $\mathbb{R}%
$-vector space of Bruhat-Schwartz functions. For $\varphi\in\mathcal{D}%
(\mathbb{Q}_{p}^{N})$, the largest number $l=l(\varphi)$ satisfying
(\ref{local_constancy}) is called \textit{the exponent of local constancy (or
the parameter of constancy) of} $\varphi$.

We denote by $\mathcal{D}_{m}^{l}(\mathbb{Q}_{p}^{N})$ the finite-dimensional
space of test functions from $\mathcal{D}(\mathbb{Q}_{p}^{N})$ having supports
in the ball $B_{m}^{N}$ and with parameters \ of constancy $\geq l$. We now
define a topology on $\mathcal{D}$ as follows. We say that a sequence
$\left\{  \varphi_{j}\right\}  _{j\in\mathbb{N}}$ of functions in
$\mathcal{D}$ converges to zero, if the two following conditions hold:

(1) there are two fixed integers $k_{0}$ and $m_{0}$ such that \ each
$\varphi_{j}\in$ $\mathcal{D}_{m_{0}}^{k_{0}}$;

(2) $\varphi_{j}\rightarrow0$ uniformly.

$\mathcal{D}$ endowed with the above topology becomes a topological vector space.

\subsection{$L^{\rho}$ spaces}

Given $\rho\in\lbrack1,\infty)$, we denote by $L^{\rho}:=L^{\rho}\left(
\mathbb{Q}
_{p}^{N}\right)  :=L^{\rho}\left(
\mathbb{Q}
_{p}^{N},d^{N}x\right)  ,$ the $\mathbb{C}-$vector space of all the complex
valued functions $g$ satisfying $\int_{%
\mathbb{Q}
_{p}^{N}}\left\vert g\left(  x\right)  \right\vert ^{\rho}d^{N}x<\infty$. The
corresponding $\mathbb{R}$-vector spaces are denoted as $L_{\mathbb{R}}^{\rho
}\allowbreak:=L_{\mathbb{R}}^{\rho}\left(
\mathbb{Q}
_{p}^{N}\right)  =L_{\mathbb{R}}^{\rho}\left(
\mathbb{Q}
_{p}^{N},d^{N}x\right)  $, $1\leq\rho<\infty$.

If $U$ is an open subset of $\mathbb{Q}_{p}^{N}$, $\mathcal{D}(U)$ denotes the
space of test functions with supports contained in $U$, then $\mathcal{D}(U)$
is dense in
\[
L^{\rho}\left(  U\right)  =\left\{  \varphi:U\rightarrow\mathbb{C};\left\Vert
\varphi\right\Vert _{\rho}=\left\{  \int_{U}\left\vert \varphi\left(
x\right)  \right\vert ^{\rho}d^{N}x\right\}  ^{\frac{1}{\rho}}<\infty\right\}
,
\]
where $d^{N}x$ is the normalized Haar measure on $\left(  \mathbb{Q}_{p}%
^{N},+\right)  $, for $1\leq\rho<\infty$, see e.g. \cite[Section 4.3]{A-K-S}.
We denote by $L_{\mathbb{R}}^{\rho}\left(  U\right)  $ the real counterpart of
$L^{\rho}\left(  U\right)  $.

\subsection{The Fourier transform}

Set $\chi_{p}(y)=\exp(2\pi i\{y\}_{p})$ for $y\in\mathbb{Q}_{p}$. The map
$\chi_{p}(\cdot)$ is an additive character on $\mathbb{Q}_{p}$, i.e. a
continuous map from $\left(  \mathbb{Q}_{p},+\right)  $ into $S$ (the unit
circle considered as multiplicative group) satisfying $\chi_{p}(x_{0}%
+x_{1})=\chi_{p}(x_{0})\chi_{p}(x_{1})$, $x_{0},x_{1}\in\mathbb{Q}_{p}$. \ The
additive characters of $\mathbb{Q}_{p}$ form an Abelian group which is
isomorphic to $\left(  \mathbb{Q}_{p},+\right)  $. The isomorphism is given by
$\kappa\rightarrow\chi_{p}(\kappa x)$, see e.g. \cite[Section 2.3]{A-K-S}.

Given $\kappa=(\kappa_{1},\dots,\kappa_{N})$ and $y=(x_{1},\dots
,x_{N})\allowbreak\in\mathbb{Q}_{p}^{N}$, we set $\kappa\cdot x:=\sum
_{j=1}^{N}\kappa_{j}x_{j}$. The Fourier transform of $\varphi\in
\mathcal{D}(\mathbb{Q}_{p}^{N})$ is defined as
\[
(\mathcal{F}\varphi)(\kappa)=\int_{\mathbb{Q}_{p}^{N}}\chi_{p}(\kappa\cdot
x)\varphi(x)d^{N}x\quad\text{for }\kappa\in\mathbb{Q}_{p}^{N},
\]
where $d^{N}x$ is the normalized Haar measure on $\mathbb{Q}_{p}^{N}$. The
Fourier transform is a linear isomorphism from $\mathcal{D}(\mathbb{Q}_{p}%
^{N})$ onto itself satisfying
\begin{equation}
(\mathcal{F}(\mathcal{F}\varphi))(\kappa)=\varphi(-\kappa), \label{Eq_FFT}%
\end{equation}
see e.g. \cite[Section 4.8]{A-K-S}. We will also use the notation
$\mathcal{F}_{x\rightarrow\kappa}\varphi$ and $\widehat{\varphi}$\ for the
Fourier transform of $\varphi$.

The Fourier transform extends to $L^{2}$. If $f\in L^{2},$ its Fourier
transform is defined as
\[
(\mathcal{F}f)(\kappa)=\lim_{k\rightarrow\infty}\int_{||x||_{p}\leq p^{k}}%
\chi_{p}(\kappa\cdot x)f(x)d^{N}x,\quad\text{for }\kappa\in%
\mathbb{Q}
_{p}^{N},
\]
where the limit is taken in $L^{2}$. We recall that the Fourier transform is
unitary on $L^{2},$ i.e. $||f||_{L^{2}}=||\mathcal{F}f||_{L^{2}}$ for $f\in
L^{2}$ and that (\ref{Eq_FFT}) is also valid in $L^{2}$, see e.g.
\cite[Chapter III, Section 2]{Taibleson}.

\subsection{Distributions}

The $\mathbb{C}$-vector space $\mathcal{D}^{\prime}\left(  \mathbb{Q}_{p}%
^{N}\right)  $ $:=\mathcal{D}^{\prime}$ of all continuous linear functionals
on $\mathcal{D}(\mathbb{Q}_{p}^{N})$ is called the \textit{Bruhat-Schwartz
space of distributions}. Every linear functional on $\mathcal{D}$ is
continuous, i.e. $\mathcal{D}^{\prime}$\ agrees with the algebraic dual of
$\mathcal{D}$, see e.g. \cite[Chapter 1, VI.3, Lemma]{V-V-Z}. We denote by
$\mathcal{D}_{\mathbb{R}}^{\prime}\left(  \mathbb{Q}_{p}^{N}\right)  $
$:=\mathcal{D}_{\mathbb{R}}^{\prime}$ the dual space of $\mathcal{D}%
_{\mathbb{R}}$.

We endow $\mathcal{D}^{\prime}$ with the weak topology, i.e. a sequence
$\left\{  T_{j}\right\}  _{j\in\mathbb{N}}$ in $\mathcal{D}^{\prime}$
converges to $T$ if $\lim_{j\rightarrow\infty}T_{j}\left(  \varphi\right)
=T\left(  \varphi\right)  $ for any $\varphi\in\mathcal{D}$. \ The map
\[%
\begin{array}
[c]{lll}%
\mathcal{D}^{\prime}\times\mathcal{D} & \rightarrow & \mathbb{C}\\
\left(  T,\varphi\right)  & \rightarrow & T\left(  \varphi\right)
\end{array}
\]
is a bilinear form which is continuous in $T$ and $\varphi$ separately. We
call this map the pairing between $\mathcal{D}^{\prime}$ and $\mathcal{D}$.
From now on we will use $\left(  T,\varphi\right)  $ instead of $T\left(
\varphi\right)  $.

Every $f$\ in $L_{loc}^{1}$ defines a distribution $f\in\mathcal{D}^{\prime
}\left(  \mathbb{Q}_{p}^{N}\right)  $ by the formula
\[
\left(  f,\varphi\right)  =%
{\textstyle\int\limits_{\mathbb{Q}_{p}^{n}}}
f\left(  x\right)  \varphi\left(  x\right)  d^{N}x.
\]
Such distributions are called \textit{regular distributions}. Notice that for
$f$\ $\in L_{\mathbb{R}}^{2}$, $\left(  f,\varphi\right)  =\left\langle
f,\varphi\right\rangle $, where $\left\langle \cdot,\cdot\right\rangle $
denotes the scalar product in $L_{\mathbb{R}}^{2}$.

\begin{remark}
\label{Nota_Nuclear}Let $B(\psi,\varphi)$ be a bilinear functional, $\psi
\in\mathcal{D}\left(  \mathbb{Q}_{p}^{N}\right)  $, $\varphi\in\mathcal{D}%
\left(  \mathbb{Q}_{p}^{M}\right)  $. Then there exists a unique distribution
$T\in\mathcal{D}^{\prime}\left(  \mathbb{Q}_{p}^{N}\times\mathbb{Q}_{p}%
^{M}\right)  $ such that
\[
\left(  T,\psi\left(  x\right)  \varphi\left(  y\right)  \right)
=B(\psi,\varphi)\text{, for }\psi\in\mathcal{D}\left(  \mathbb{Q}_{p}%
^{N}\right)  ,\varphi\in\mathcal{D}\left(  \mathbb{Q}_{p}^{M}\right)  ,
\]
cf. \cite[Chapter 1, VI.7, Theorem]{V-V-Z}
\end{remark}

\subsection{The Fourier transform of a distribution}

The Fourier transform $\mathcal{F}\left[  T\right]  $ of a distribution
$T\in\mathcal{D}^{\prime}\left(  \mathbb{Q}_{p}^{N}\right)  $ is defined by%
\[
\left(  \mathcal{F}\left[  T\right]  ,\varphi\right)  =\left(  T,\mathcal{F}%
\left[  \varphi\right]  \right)  \text{ for all }\varphi\in\mathcal{D}%
(\mathbb{Q}_{p}^{N})\text{.}%
\]
The Fourier transform $T\rightarrow\mathcal{F}\left[  T\right]  $ is a linear
(and continuous) isomorphism from $\mathcal{D}^{\prime}\left(  \mathbb{Q}%
_{p}^{N}\right)  $\ onto $\mathcal{D}^{\prime}\left(  \mathbb{Q}_{p}%
^{N}\right)  $. Furthermore, $T=\mathcal{F}\left[  \mathcal{F}\left[
T\right]  \left(  -\xi\right)  \right]  $.

\section{$\boldsymbol{W}_{\delta}$ \ operators and their discretizations}

\subsection{The $\boldsymbol{W}_{\delta}$ operators}

Take \ $\mathbb{R}_{+}:=\left\{  x\in\mathbb{R};x\geq0\right\}  $, and fix a
function%
\[
w_{\delta}:\mathbb{Q}_{p}^{N}\rightarrow\mathbb{R}_{+}%
\]
satisfying the following properties:

\noindent(i) $w_{\delta}\left(  y\right)  $ is a radial i.e. $w_{\delta
}(y)=w_{\delta}(\left\Vert y\right\Vert _{p})$;

\noindent(ii) $w_{\delta}(\left\Vert y\right\Vert _{p})$ is a continuous and
increasing function of $\left\Vert y\right\Vert _{p}$;

\noindent(iii) $w_{\delta}\left(  y\right)  =0$ if and only if $y=0$;

\noindent(iv) there exist constants $C_{0},C_{1}>0$ and $\delta>N$ such that
\begin{equation}
C_{0}\left\Vert y\right\Vert _{p}^{\delta}\leq w_{\delta}(\left\Vert
y\right\Vert _{p})\leq C_{1}\left\Vert y\right\Vert _{p}^{\delta}\text{, for
}y\in\mathbb{Q}_{p}^{N}\text{.}\label{Eq_1}%
\end{equation}
We now define the operator
\begin{equation}
\boldsymbol{W}_{\delta}\varphi(x)={\int\limits_{\mathbb{Q}_{p}^{N}}}%
\frac{\varphi\left(  x-y\right)  -\varphi\left(  x\right)  }{w_{\delta}\left(
\Vert y\Vert_{p}\right)  }d^{N}y\text{, for }\varphi\in\mathcal{D}\left(
\mathbb{Q}_{p}^{N}\right)  \text{.}\label{EQ_oper_W_def}%
\end{equation}
The operator $\boldsymbol{W}_{\delta}$ is pseudo-differential; more precisely,
if%
\begin{equation}
A_{w_{\delta}}\left(  \kappa\right)  :={\int\limits_{\mathbb{Q}_{p}^{N}}}%
\frac{1-\chi_{p}\left(  y\cdot\kappa\right)  }{w_{\delta}\left(  \Vert
y\Vert_{p}\right)  }d^{N}y,\label{Eq_Kernel}%
\end{equation}
then
\begin{equation}
\boldsymbol{W}_{\delta}\varphi\left(  x\right)  =-\mathcal{F}_{\kappa
\rightarrow x}^{-1}\left[  A_{w_{\delta}}\left(  \kappa\right)  \mathcal{F}%
_{x\rightarrow\kappa}\varphi\right]  =:-\boldsymbol{W}\left(  \partial
,\delta\right)  \varphi\left(  x\right)  \text{, for }\varphi\in
\mathcal{D}\left(  \mathbb{Q}_{p}^{N}\right)  \text{.}\label{EQ_oper_W_pseudo}%
\end{equation}
The function $A_{w_{\delta}}\left(  \kappa\right)  $ is radial (so we use the
notation $A_{w_{\delta}}\left(  \kappa\right)  =A_{w_{\delta}}\left(
\Vert\kappa\Vert_{p}\right)  $), continuous, non-negative, $A_{w_{\delta}%
}\left(  0\right)  =0$, and it satisfies
\[
C_{0}^{\prime}\left\Vert \kappa\right\Vert _{p}^{\delta-N}\leq A_{w_{\delta}%
}(\left\Vert \kappa\right\Vert _{p})\leq C_{1}^{\prime}\left\Vert
\kappa\right\Vert _{p}^{\delta-N}\text{, for }\kappa\in\mathbb{Q}_{p}%
^{N}\text{,}%
\]
cf. \cite[Lemmas 4, 5, 8 ]{Zuniga-LNM-2016}. The operator $\boldsymbol{W}%
\left(  \partial,\delta\right)  $ extends to an unbounded and densely defined
operator in $L^{2}\left(  \mathbb{Q}_{p}^{N}\right)  $ with domain
\begin{equation}
Dom(\boldsymbol{W}\left(  \partial,\delta\right)  )=\left\{  \varphi\in
L^{2};A_{w_{\delta}}(\left\Vert \kappa\right\Vert _{p})\mathcal{F}\varphi\in
L^{2}\right\}  .\label{Dom_W}%
\end{equation}
In addition:

\noindent(i) $\left(  \boldsymbol{W}\left(  \partial,\delta\right)
,Dom(\boldsymbol{W}\left(  \partial,\delta\right)  )\right)  $ is self-adjoint
and positive operator;

\noindent(ii) $-\boldsymbol{W}\left(  \partial,\delta\right)  $ is the
infinitesimal generator of a contraction $C_{0}-$semigroup, cf.
\cite[Proposition 7]{Zuniga-LNM-2016}.

The evolution equation%
\[
\frac{\partial u\left(  x,t\right)  }{\partial t}+\boldsymbol{W}\left(
\partial,\delta\right)  u(x,t)=0\text{, \ \ \ }x\in\mathbb{Q}_{p}^{N}\text{,
}t\geq0\text{,}%
\]
is a $p$-adic heat equation, which means that the corresponding semigroup is
attached to a Markov stochastic process, see \cite[Theorem 16]%
{Zuniga-LNM-2016}.

\begin{example}
An important example of a $\boldsymbol{W}\left(  \partial,\delta\right)  $
operator is the Taibleson-Vladimirov operator, which is defined as%
\[
\boldsymbol{D}^{\beta}\phi\left(  x\right)  =\frac{1-p^{\beta}}{1-p^{-\beta
-N}}\int\limits_{\mathbb{Q}_{p}^{N}}\frac{\phi\left(  x-y\right)  -\phi\left(
x\right)  }{\left\Vert y\right\Vert _{p}^{\beta+N}}d^{N}y=\mathcal{F}%
_{\kappa\rightarrow x}^{-1}\left(  \left\Vert \kappa\right\Vert _{p}^{\beta
}\mathcal{F}_{x\rightarrow\kappa}\phi\right)  \text{,}%
\]
where $\beta>0$ and $\phi\in\mathcal{D}\left(  \mathbb{Q}_{p}^{N}\right)  $,
see \cite[Section 2.2.7]{Zuniga-LNM-2016}.
\end{example}

The $\boldsymbol{W}_{\delta}$ operators were introduced by
Chac\'{o}n-Cort\'{e}s and Z\'{u}\~{n}iga-Galindo, see \cite{Zuniga-LNM-2016}%
\ and the references therein. They are a generalization of the Vladimirov and
Taibleson operators.

\subsection{Discretization of $\boldsymbol{W}_{\delta}$ operators}

For $l\geq1$, we \ set $G_{l}:=p^{-l}\mathbb{Z}_{p}^{N}/p^{l}\mathbb{Z}%
_{p}^{N}$ and denote by $\mathcal{D}_{\mathbb{R}}^{l}(\mathbb{Q}_{p}%
^{N}):=\mathcal{D}_{\mathbb{R}}^{l}$ the $\mathbb{R}$-vector space of all test
functions of the form%
\begin{equation}
\varphi\left(  x\right)  =%
{\textstyle\sum\limits_{\boldsymbol{i}\in G_{l}}}
\varphi\left(  \boldsymbol{i}\right)  \Omega\left(  p^{l}\left\Vert
x-\boldsymbol{i}\right\Vert _{p}\right)  \text{, \ }\varphi\left(
\boldsymbol{i}\right)  \in\mathbb{R}\text{,} \label{Eq_repre}%
\end{equation}
where $\boldsymbol{i}$ runs through a fixed system of representatives of
$G_{l}$, and $\Omega\left(  p^{l}\left\Vert x-\boldsymbol{i}\right\Vert
_{p}\right)  $ is the characteristic function of the ball $\boldsymbol{i}%
+p^{l}\mathbb{Z}_{p}^{N}$. Notice that $\varphi$ is supported on
$p^{-l}\mathbb{Z}_{p}^{N}$ and that $\mathcal{D}_{\mathbb{R}}^{l}$ is a finite
dimensional vector space spanned by the basis
\begin{equation}
\left\{  \Omega\left(  p^{l}\left\Vert x-\boldsymbol{i}\right\Vert
_{p}\right)  \right\}  _{\boldsymbol{i}\in G_{l}}. \label{Basis}%
\end{equation}
We will identify $\varphi\in\mathcal{D}_{\mathbb{R}}^{l}$ with the column
vector $\left[  \varphi\left(  \boldsymbol{i}\right)  \right]
_{\boldsymbol{i}\in G_{l}}$. Furthermore, $\mathcal{D}_{\mathbb{R}}^{l}$
$\hookrightarrow\mathcal{D}_{\mathbb{R}}^{l+1}$ (continuous embedding), and
$\mathcal{D}_{\mathbb{R}}=\underrightarrow{\lim}\mathcal{D}_{\mathbb{R}}%
^{l}=\cup_{l=1}^{\infty}\mathcal{D}_{\mathbb{R}}^{l}$.

\begin{remark}
We set
\[
d\left(  l,w_{\delta}\right)  :={\int\limits_{\mathbb{Q}_{p}^{N}\setminus
B_{-l}^{N}}}\frac{d^{N}y}{w_{\delta}\left(  \Vert y\Vert_{p}\right)  }.
\]
By (\ref{Eq_1}), $d\left(  l,w_{\delta}\right)  <\infty$. Furthermore, we
\ have%
\begin{equation}
\frac{p^{\left(  \delta-N\right)  l}}{C_{1}}{\int\limits_{\mathbb{Q}_{p}%
^{N}\setminus\mathbb{Z}_{p}^{N}}}\frac{d^{N}z}{\Vert z\Vert_{p}^{\delta}}\leq
d\left(  l,w_{\delta}\right)  \leq\frac{p^{\left(  \delta-N\right)  l}}{C_{0}%
}{\int\limits_{\mathbb{Q}_{p}^{N}\setminus\mathbb{Z}_{p}^{N}}}\frac{d^{N}%
z}{\Vert z\Vert_{p}^{\delta}}, \label{Eq_1A}%
\end{equation}
which implies that $d\left(  l,w_{\delta}\right)  \geq Cp^{\left(
\delta-N\right)  l}$ for some positive constant $C$. In particular, $d\left(
l,w_{\delta}\right)  \rightarrow\infty$ as $l\rightarrow\infty$.
\end{remark}

We denote by $\boldsymbol{W}_{\delta}^{\left(  l\right)  }$ the restriction
$\boldsymbol{W}_{\delta}:\mathcal{D}_{\mathbb{R}}\left(  B_{l}^{N}\right)
\rightarrow\mathcal{D}_{\mathbb{R}}\left(  B_{l}^{N}\right)  $. Take
$\varphi\in\mathcal{D}_{\mathbb{R}}\left(  B_{l}^{N}\right)  $ is then%
\begin{gather}
\boldsymbol{W}_{\delta}^{\left(  l\right)  }\varphi(x)={\int
\limits_{\mathbb{Q}_{p}^{N}}}\frac{\varphi\left(  x-y\right)  -\varphi\left(
x\right)  }{w_{\delta}\left(  \Vert y\Vert_{p}\right)  }d^{N}y={\int
\limits_{B_{l}^{N}}}\frac{\varphi\left(  x-y\right)  -\varphi\left(  x\right)
}{w_{\delta}\left(  \Vert y\Vert_{p}\right)  }d^{N}y+\label{Eq_D_l}\\
{\int\limits_{\mathbb{Q}_{p}^{N}\smallsetminus B_{l}^{N}}}\frac{\varphi\left(
x-y\right)  -\varphi\left(  x\right)  }{w_{\delta}\left(  \Vert y\Vert
_{p}\right)  }d^{N}y={\int\limits_{B_{l}^{N}}}\frac{\varphi\left(  x-y\right)
-\varphi\left(  x\right)  }{w_{\delta}\left(  \Vert y\Vert_{p}\right)  }%
d^{N}y-\left(  \text{ }{\int\limits_{\mathbb{Q}_{p}^{N}\setminus B_{l}^{N}}%
}\frac{d^{N}y}{w_{\delta}\left(  \Vert y\Vert_{p}\right)  }\right)
\varphi\left(  x\right)  .\nonumber
\end{gather}

\begin{notation}
The cardinality of a finite set $B$ is denoted as $\#B$.
\end{notation}

We set
\begin{equation}
A_{\boldsymbol{i},\boldsymbol{j}}\left(  l\right)  :=\left\{
\begin{array}
[c]{lll}%
\frac{p^{-lN}}{w_{\delta}\left(  \left\Vert \boldsymbol{i}-\boldsymbol{j}%
\right\Vert _{p}\right)  } & \text{if} & \boldsymbol{i}\neq\boldsymbol{j}\\
&  & \\
0 & \text{if} & \boldsymbol{i}=\boldsymbol{j}\text{,}%
\end{array}
\right.  \label{Eq_Matrix_A}%
\end{equation}
and $A:=\left[  A_{\boldsymbol{i},\boldsymbol{j}}\left(  l\right)  \right]
_{\boldsymbol{i},\boldsymbol{j}\in G_{l}}$. We denote by $\mathbb{I}$ the
identity matrix of size $\#G_{l}\times\#G_{l}$.

\begin{lemma}
\label{Lemma4}The restriction $\boldsymbol{W}_{\delta}^{\left(  l\right)
}:\mathcal{D}_{\mathbb{R}}^{l}\rightarrow\mathcal{D}_{\mathbb{R}}^{l}$ is a
well-defined linear operator. Furthermore, the following formula holds true:
\[
\boldsymbol{W}_{\delta}^{\left(  l\right)  }\varphi(x)=%
{\textstyle\sum\limits_{\boldsymbol{i}\in G_{l}}}
\left\{
{\textstyle\sum\limits_{\boldsymbol{j}\in G_{l}}}
A_{\boldsymbol{i},\boldsymbol{j}}\left(  l\right)  \varphi\left(
\boldsymbol{j}\right)  -\varphi\left(  \boldsymbol{i}\right)  d\left(
l,w_{\delta}\right)  \right\}  \Omega\left(  p^{l}\left\Vert x-\boldsymbol{i}%
\right\Vert _{p}\right)  \text{, }%
\]
which implies that $A-d\left(  l,w_{\delta}\right)  \mathbb{I}$ is the matrix
of the operator $\boldsymbol{W}_{\delta}^{\left(  l\right)  }$ in the basis
(\ref{Basis}).
\end{lemma}

\begin{proof}
For $x\in\boldsymbol{i}+p^{l}\mathbb{Z}_{p}^{N}$ and for $\varphi\left(
x\right)  $ of the from (\ref{Eq_repre}), we have
\begin{gather*}
\boldsymbol{W}_{\delta}^{\left(  l\right)  }\varphi(x)={\int
\limits_{\mathbb{Q}_{p}^{N}}}\frac{\varphi\left(  y\right)  -\varphi\left(
x\right)  }{w_{\delta}\left(  \Vert y-x\Vert_{p}\right)  }d^{N}y={\int
\limits_{\mathbb{Q}_{p}^{N}}}\frac{%
{\textstyle\sum\limits_{\boldsymbol{j}\in G_{l}}}
\varphi\left(  \boldsymbol{j}\right)  \Omega\left(  p^{l}\left\Vert
y-\boldsymbol{j}\right\Vert _{p}\right)  -\varphi\left(  \boldsymbol{i}%
\right)  \Omega\left(  p^{l}\left\Vert x-\boldsymbol{i}\right\Vert
_{p}\right)  }{w_{\delta}\left(  \Vert y-x\Vert_{p}\right)  }d^{N}y\\
=%
{\textstyle\sum\limits_{\substack{\boldsymbol{j}\in G_{l}\\\boldsymbol{j}%
\neq\boldsymbol{i}}}}
\text{ }{\int\limits_{\mathbb{Q}_{p}^{N}}}\frac{\varphi\left(  \boldsymbol{j}%
\right)  \Omega\left(  p^{l}\left\Vert y-\boldsymbol{j}\right\Vert
_{p}\right)  }{w_{\delta}\left(  \Vert y-x\Vert_{p}\right)  }d^{N}%
y+{\int\limits_{\mathbb{Q}_{p}^{N}}}\frac{\varphi\left(  \boldsymbol{i}%
\right)  \left\{  \Omega\left(  p^{l}\left\Vert y-\boldsymbol{i}\right\Vert
_{p}\right)  -\Omega\left(  p^{l}\left\Vert x-\boldsymbol{i}\right\Vert
_{p}\right)  \right\}  }{w_{\delta}\left(  \Vert y-x\Vert_{p}\right)  }%
d^{N}y\\
=%
{\textstyle\sum\limits_{\substack{\boldsymbol{j}\in G_{l}\\\boldsymbol{j}%
\neq\boldsymbol{i}}}}
A_{\boldsymbol{i},\boldsymbol{j}}\left(  l\right)  \varphi\left(
\boldsymbol{j}\right)  +{\int\limits_{\mathbb{Q}_{p}^{N}\setminus
\boldsymbol{i}+p^{l}\mathbb{Z}_{p}^{N}}}\frac{\varphi\left(  \boldsymbol{i}%
\right)  \left\{  \Omega\left(  p^{l}\left\Vert y-\boldsymbol{i}\right\Vert
_{p}\right)  -1\right\}  }{w_{\delta}\left(  \Vert y-x\Vert_{p}\right)  }%
d^{N}y.
\end{gather*}
Now%
\begin{align*}
{\int\limits_{\mathbb{Q}_{p}^{N}\setminus\left(  \boldsymbol{i}+p^{l}%
\mathbb{Z}_{p}^{N}\right)  }}\frac{\varphi\left(  \boldsymbol{i}\right)
\left\{  \Omega\left(  p^{l}\left\Vert y-\boldsymbol{i}\right\Vert
_{p}\right)  -1\right\}  }{w_{\delta}\left(  \Vert y-x\Vert_{p}\right)  }%
d^{N}y  &  ={\int\limits_{\mathbb{Q}_{p}^{N}\setminus p^{l}\mathbb{Z}_{p}^{N}%
}}\frac{\varphi\left(  \boldsymbol{i}\right)  \left\{  \Omega\left(
p^{l}\left\Vert z\right\Vert _{p}\right)  -1\right\}  }{w_{\delta}\left(
\Vert z+\left(  \boldsymbol{i}-x\right)  \Vert_{p}\right)  }d^{N}z\\
&  =-\varphi\left(  \boldsymbol{i}\right)  {\int\limits_{\mathbb{Q}_{p}%
^{N}\setminus p^{l}\mathbb{Z}_{p}^{N}}}\frac{d^{N}z}{w_{\delta}\left(  \Vert
z\Vert_{p}\right)  }.
\end{align*}

\end{proof}

\section{Energy functionals}

\subsection{Energy functionals in the coordinate space}

For $\varphi\in\mathcal{D}_{\mathbb{R}}(\mathbb{Q}_{p}^{N})$, and $\delta>N$,
$\gamma>0$, $\alpha_{2}\geq0$, we define the energy functional:%
\begin{equation}
E_{0}(\varphi):=E_{0}(\varphi;\delta,\gamma,\alpha_{2})=\frac{\gamma}{4}\text{
\ }%
{\textstyle\iint\limits_{\mathbb{Q}_{p}^{N}\times\mathbb{Q}_{p}^{N}}}
\text{ }\frac{\left\{  \varphi\left(  x\right)  -\varphi\left(  y\right)
\right\}  ^{2}}{w_{\delta}\left(  \left\Vert x-y\right\Vert _{p}\right)
}d^{N}xd^{N}y+\frac{\alpha_{2}}{2}\int\limits_{\mathbb{Q}_{p}^{N}}\varphi
^{2}\left(  x\right)  d^{N}x\geq0. \label{Energy_Functioal_E_0}%
\end{equation}
Then $E_{0}$ is a well-defined real-valued functional on $\mathcal{D}%
_{\mathbb{R}}$. Notice that $E_{0}(\varphi)=0$ if an only if $\varphi=0$. The
restriction of $E_{0}$ to $\mathcal{D}_{\mathbb{R}}^{l}$ (denoted as
$E_{0}^{\left(  l\right)  }$) provides a natural discretization of $E_{0}$.

\begin{remark}
\label{Nota_discretization}The functional
\[
E_{m}^{\prime}(\varphi):=\int\limits_{\mathbb{Q}_{p}^{N}}\varphi^{m}\left(
x\right)  d^{N}x\text{ for }m\in\mathbb{N}\smallsetminus\left\{  0\right\}
\text{, }\varphi\in\mathcal{D}_{\mathbb{R}}^{l}\text{,}%
\]
discretizes as
\[
E_{m}^{\prime}(\varphi)=p^{-lN}%
{\textstyle\sum\limits_{\boldsymbol{i}\in G_{l}}}
\varphi^{m}\left(  \boldsymbol{i}\right)  \text{.}%
\]

\end{remark}

\begin{lemma}
\label{Lemma5}For $\varphi\in\mathcal{D}_{\mathbb{R}}^{l}$, the following
formula holds true:%
\[
E_{0}^{\left(  l\right)  }(\varphi)=p^{-lN}\left(  \frac{\gamma}{2}d\left(
l,w_{\delta}\right)  +\frac{\alpha_{2}}{2}\right)
{\textstyle\sum\limits_{\boldsymbol{i}\in G_{l}}}
\varphi^{2}\left(  \boldsymbol{i}\right)  -\frac{\gamma}{2}p^{_{-lN}}%
{\textstyle\sum\limits_{\boldsymbol{i},\boldsymbol{j}\in G_{l}}}
A_{\boldsymbol{i},\boldsymbol{j}}(l)\varphi\left(  \boldsymbol{i}\right)
\varphi\left(  \boldsymbol{j}\right)  .
\]

\end{lemma}

\begin{proof}
We set%
\begin{align*}
E_{0}^{\prime}(\varphi)  &  :=\frac{\gamma}{4}\text{ }%
{\textstyle\iint\limits_{\mathbb{Q}_{p}^{N}\times\mathbb{Q}_{p}^{N}}}
\text{ }\frac{\left\{  \varphi\left(  x\right)  -\varphi\left(  y\right)
\right\}  ^{2}}{w_{\delta}\left(  \left\Vert x-y\right\Vert _{p}\right)
}d^{N}xd^{N}y\\
&  =\frac{\gamma}{4}\text{ }%
{\textstyle\iint\limits_{\mathbb{Q}_{p}^{N}\times\mathbb{Q}_{p}^{N}}}
\text{ }\frac{\left\{
{\textstyle\sum\limits_{\boldsymbol{i}\in G_{l}}}
\varphi\left(  \boldsymbol{i}\right)  \left[  \Omega\left(  p^{l}\left\Vert
x-\boldsymbol{i}\right\Vert _{p}\right)  -\Omega\left(  p^{l}\left\Vert
y-\boldsymbol{i}\right\Vert _{p}\right)  \right]  \right\}  ^{2}}{w_{\delta
}\left(  \left\Vert x-y\right\Vert _{p}\right)  }d^{N}xd^{N}y.
\end{align*}
Now, by using that for $\boldsymbol{i}\neq\boldsymbol{j}$,%
\[
\Omega\left(  p^{l}\left\Vert x-\boldsymbol{i}\right\Vert _{p}\right)
\Omega\left(  p^{l}\left\Vert y-\boldsymbol{j}\right\Vert _{p}\right)
=1\text{ }\Rightarrow\text{ }\Omega\left(  p^{l}\left\Vert x-\boldsymbol{j}%
\right\Vert _{p}\right)  \Omega\left(  p^{l}\left\Vert y-\boldsymbol{i}%
\right\Vert _{p}\right)  =0,
\]
we get that
\begin{gather*}
\left\{
{\textstyle\sum\limits_{\boldsymbol{i}\in G_{l}}}
\varphi\left(  \boldsymbol{i}\right)  \left[  \Omega\left(  p^{l}\left\Vert
x-\boldsymbol{i}\right\Vert _{p}\right)  -\Omega\left(  p^{l}\left\Vert
y-\boldsymbol{i}\right\Vert _{p}\right)  \right]  \right\}  ^{2}=\\%
{\textstyle\sum\limits_{\boldsymbol{i}\in G_{l}}}
\varphi^{2}\left(  \boldsymbol{i}\right)  \left[  \Omega\left(  p^{l}%
\left\Vert x-\boldsymbol{i}\right\Vert _{p}\right)  -\Omega\left(
p^{l}\left\Vert y-\boldsymbol{i}\right\Vert _{p}\right)  \right]  ^{2}-\\
2%
{\textstyle\sum\limits_{\substack{\boldsymbol{i},\boldsymbol{j}\in
G_{l}\\\boldsymbol{i}\neq\boldsymbol{j}}}}
\varphi\left(  \boldsymbol{i}\right)  \varphi\left(  \boldsymbol{j}\right)
\Omega\left(  p^{l}\left\Vert x-\boldsymbol{i}\right\Vert _{p}\right)
\Omega\left(  p^{l}\left\Vert y-\boldsymbol{j}\right\Vert _{p}\right)  .
\end{gather*}
Therefore%
\[
E_{0}^{\prime}(\varphi)=\frac{\gamma}{4}\text{ }%
{\textstyle\sum\limits_{\boldsymbol{i}\in G_{l}}}
E_{\boldsymbol{i}}^{\left(  1\right)  }(\varphi)-\frac{\gamma}{2}\text{ }%
{\textstyle\sum\limits_{\substack{\boldsymbol{i},\boldsymbol{j}\in
G_{l}\\\boldsymbol{i}\neq\boldsymbol{j}}}}
E_{\boldsymbol{i},\boldsymbol{j}}^{(2)}(\varphi),
\]
where%
\begin{gather*}
E_{\boldsymbol{i}}^{\left(  1\right)  }(\varphi):=\varphi^{2}\left(
\boldsymbol{i}\right)
{\textstyle\iint\limits_{\mathbb{Q}_{p}^{N}\times\mathbb{Q}_{p}^{N}}}
\text{ }\frac{\left[  \Omega\left(  p^{l}\left\Vert x-\boldsymbol{i}%
\right\Vert _{p}\right)  -\Omega\left(  p^{l}\left\Vert y-\boldsymbol{i}%
\right\Vert _{p}\right)  \right]  ^{2}}{w_{\delta}\left(  \left\Vert
x-y\right\Vert _{p}\right)  }d^{N}xd^{N}y=\\
\varphi^{2}\left(  \boldsymbol{i}\right)  \int\limits_{\left\Vert x\right\Vert
_{p}>p^{-l}}\text{ }\int\limits_{\left\Vert y\right\Vert _{p}\leq p^{-l}%
}\text{ }\frac{d^{N}xd^{N}y}{w_{\delta}\left(  \left\Vert x-y\right\Vert
_{p}\right)  }+\varphi^{2}\left(  \boldsymbol{i}\right)  \int
\limits_{\left\Vert x\right\Vert _{p}\leq p^{-l}}\text{ }\int
\limits_{\left\Vert y\right\Vert _{p}>p^{-l}}\text{ }\frac{d^{N}xd^{N}%
y}{w_{\delta}\left(  \left\Vert x-y\right\Vert _{p}\right)  }=\\
2\varphi^{2}\left(  \boldsymbol{i}\right)  \int\limits_{\left\Vert
x\right\Vert _{p}>p^{-l}}\text{ }\int\limits_{\left\Vert y\right\Vert _{p}\leq
p^{-l}}\text{ }\frac{d^{N}xd^{N}y}{w_{\delta}\left(  \left\Vert x-y\right\Vert
_{p}\right)  }=2p^{-lN}\varphi^{2}\left(  \boldsymbol{i}\right)  d\left(
l,w_{\delta}\right)  .
\end{gather*}
And \ for $\boldsymbol{i},\boldsymbol{j}\in G_{l}$, with $\boldsymbol{i}%
\neq\boldsymbol{j}$,
\begin{align*}
E_{\boldsymbol{i},\boldsymbol{j}}^{(2)}(\varphi)  &  :=\varphi\left(
\boldsymbol{i}\right)  \varphi\left(  \boldsymbol{j}\right)
{\textstyle\iint\limits_{\mathbb{Q}_{p}^{N}\times\mathbb{Q}_{p}^{N}}}
\frac{\Omega\left(  p^{l}\left\Vert x-\boldsymbol{i}\right\Vert _{p}\right)
\Omega\left(  p^{l}\left\Vert y-\boldsymbol{j}\right\Vert _{p}\right)
}{w_{\delta}\left(  \left\Vert x-y\right\Vert _{p}\right)  }d^{N}xd^{N}y\\
&  =\frac{p^{_{-2lN}}}{w_{\delta}\left(  \left\Vert \boldsymbol{i}%
-\boldsymbol{j}\right\Vert _{p}\right)  }\varphi\left(  \boldsymbol{i}\right)
\varphi\left(  \boldsymbol{j}\right)  .
\end{align*}
Consequently,
\begin{align}
E_{0}^{\prime}(\varphi)  &  =\frac{\gamma}{2}p^{-lN}d\left(  l,w_{\delta
}\right)
{\textstyle\sum\limits_{\boldsymbol{i}\in G_{l}}}
\varphi^{2}\left(  \boldsymbol{i}\right)  -\frac{\gamma}{2}%
{\textstyle\sum\limits_{\substack{\boldsymbol{i},\boldsymbol{j}\in
G_{l}\\\boldsymbol{i}\neq\boldsymbol{j}}}}
\text{ }\frac{p^{_{-2lN}}}{w_{\delta}\left(  \left\Vert \boldsymbol{i}%
-\boldsymbol{j}\right\Vert _{p}\right)  }\varphi\left(  \boldsymbol{i}\right)
\varphi\left(  \boldsymbol{j}\right) \nonumber\\
&  =\frac{\gamma}{2}p^{-lN}d\left(  l,w_{\delta}\right)
{\textstyle\sum\limits_{\boldsymbol{i}\in G_{l}}}
\varphi^{2}\left(  \boldsymbol{i}\right)  -\frac{\gamma}{2}p^{_{-lN}}%
{\textstyle\sum\limits_{\boldsymbol{i},\boldsymbol{j}\in G_{l}}}
A_{\boldsymbol{i},\boldsymbol{j}}(l)\varphi\left(  \boldsymbol{i}\right)
\varphi\left(  \boldsymbol{j}\right)  . \label{Eq_S_phi}%
\end{align}

The announced formula follows from (\ref{Eq_S_phi}) by using Remark
\ref{Nota_discretization}.
\end{proof}

We now set $U\left(  l\right)  :=U=\left[  U_{\boldsymbol{i},\boldsymbol{j}%
}(l)\right]  _{\boldsymbol{i},\boldsymbol{j}\in G_{l}}$, where%
\[
U_{\boldsymbol{i},\boldsymbol{j}}(l):=\left(  \frac{\gamma}{2}d\left(
l,w_{\delta}\right)  +\frac{\alpha_{2}}{2}\right)  \delta_{\boldsymbol{i}%
,\boldsymbol{j}}-\frac{\gamma}{2}A_{\boldsymbol{i},\boldsymbol{j}}(l),
\]
where $\delta_{\boldsymbol{i},\boldsymbol{j}}$\ denotes the Kronecker delta.
Notice that $U=\left(  \frac{\gamma}{2}d\left(  l,w_{\delta}\right)
+\frac{\alpha_{2}}{2}\right)  \mathbb{I}-\frac{\gamma}{2}A$ is the matrix of
the operator%
\[
-\frac{\gamma}{2}\boldsymbol{W}_{\delta}+\frac{\alpha_{2}}{2}%
\]
acting on $\mathcal{D}_{\mathbb{R}}^{l}$, in the basis (\ref{Basis}), cf.
Lemma \ref{Lemma4}.

\begin{lemma}
\label{Lemma6}With the above notation the following formula holds true:%
\begin{equation}
E_{0}^{\left(  l\right)  }(\varphi)=\left[  \varphi\left(  \boldsymbol{i}%
\right)  \right]  _{\boldsymbol{i}\in G_{l}}^{T}p^{-lN}U(l)\left[
\varphi\left(  \boldsymbol{i}\right)  \right]  _{\boldsymbol{i}\in G_{l}}=%
{\textstyle\sum\limits_{\boldsymbol{i},\boldsymbol{j}\in G_{l}}}
p^{-lN}U_{\boldsymbol{i},\boldsymbol{j}}(l)\varphi\left(  \boldsymbol{i}%
\right)  \varphi\left(  \boldsymbol{j}\right)  \geq0, \label{Eq_formula_E_0}%
\end{equation}
for $\varphi\in\mathcal{D}_{\mathbb{R}}^{l}$, where $U$ is a symmetric,
positive definite matrix. Consequently $p^{-lN}U(l)$ is a diagonalizable and
invertible matrix.
\end{lemma}

\[
-\frac{\gamma}{2}\boldsymbol{W}_{\delta}^{\left(  l\right)  }\varphi\left(
x\right)  +\left(  \frac{\gamma}{2}d\left(  l,w_{\delta}\right)  +\alpha
_{2}\right)  \varphi\left(  x\right)  =J(x).
\]

\subsection{The Fourier transform in $\mathcal{D}^{l}(\mathbb{Q}_{p}^{N})$}

We denote by $\mathcal{D}^{l}(\mathbb{Q}_{p}^{N}):=\mathcal{D}^{l}$ the
$\mathbb{C}$-vector space of the test functions $\varphi\in\mathcal{D}%
(\mathbb{Q}_{p}^{N})$ having the form: $\varphi\left(  x\right)
=\sum_{\boldsymbol{i}\in G_{l}}\varphi\left(  \boldsymbol{i}\right)
\Omega\left(  p^{l}\left\Vert x-\boldsymbol{i}\right\Vert _{p}\right)  $,
$\varphi\left(  \boldsymbol{i}\right)  \in\mathbb{C}$. Alternatively,
$\mathcal{D}^{l}$ the $\mathbb{C}$-vector space of the test functions
$\varphi\in\mathcal{D}(\mathbb{Q}_{p}^{N})$ satisfying:

\begin{enumerate}
\item supp $\varphi=B_{l}^{N}$;

\item for any $x\in B_{l}^{N}$, $\varphi\mid_{x+p^{l}\mathbb{Z}_{p}^{N}%
}=\varphi\left(  x\right)  $.
\end{enumerate}

Then by using that $\mathcal{F}_{x\rightarrow\kappa}\left(  \Omega\left(
p^{l}\left\Vert x-\boldsymbol{i}\right\Vert _{p}\right)  \right)
\allowbreak=p^{-lN}\chi_{p}\left(  \boldsymbol{i\cdot}\kappa\right)
\Omega\left(  p^{-l}\left\Vert \kappa\right\Vert _{p}\right)  $, we get that%
\begin{equation}
\widehat{\varphi}\left(  \kappa\right)  =p^{-lN}\Omega\left(  p^{-l}\left\Vert
\kappa\right\Vert _{p}\right)  \sum_{\boldsymbol{i}\in G_{l}}\varphi\left(
\boldsymbol{i}\right)  \chi_{p}\left(  \boldsymbol{i\cdot}\kappa\right)  .
\label{Eq_10}%
\end{equation}
By using the identity $\Omega\left(  p^{-l}\left\Vert \kappa\right\Vert
_{p}\right)  =\sum_{\boldsymbol{j}\in G_{l}}$ $\Omega\left(  p^{l}\left\Vert
\kappa-\boldsymbol{j}\right\Vert _{p}\right)  $ in (\ref{Eq_10}),
\begin{align}
\widehat{\varphi}\left(  \kappa\right)   &  =\sum_{\boldsymbol{j}\in G_{l}%
}\left\{  p^{-lN}\sum_{\boldsymbol{i}\in G_{l}}\varphi\left(  \boldsymbol{i}%
\right)  \chi_{p}\left(  \boldsymbol{i\cdot j}\right)  \right\}  \Omega\left(
p^{l}\left\Vert \kappa-\boldsymbol{j}\right\Vert _{p}\right) \nonumber\\
&  =:\sum_{\boldsymbol{j}\in G_{l}}\widehat{\varphi}\left(  \boldsymbol{j}%
\right)  \Omega\left(  p^{l}\left\Vert \kappa-\boldsymbol{j}\right\Vert
_{p}\right)  . \label{Eq_11}%
\end{align}
Conversely,%
\begin{align}
\varphi\left(  x\right)   &  =\sum_{\boldsymbol{j}\in G_{l}}\left\{
p^{-lN}\sum_{\boldsymbol{i}\in G_{l}}\widehat{\varphi}\left(  \boldsymbol{i}%
\right)  \chi_{p}\left(  -\boldsymbol{i\cdot j}\right)  \right\}
\Omega\left(  p^{l}\left\Vert x-\boldsymbol{j}\right\Vert _{p}\right)
\nonumber\\
&  =\sum_{\boldsymbol{j}\in G_{l}}\varphi\left(  \boldsymbol{j}\right)
\Omega\left(  p^{l}\left\Vert x-\boldsymbol{j}\right\Vert _{p}\right)  .
\label{Eq_11A}%
\end{align}
It follows from (\ref{Eq_11})-(\ref{Eq_11A}) that the Fourier transform is an
automorphism of the $\mathbb{C}$-vector space $\mathcal{D}^{l}$.

\begin{remark}
\label{Nota1}(i) For $\varphi\in\mathcal{D}_{\mathbb{R}}^{l}(\mathbb{Q}%
_{p}^{N})$, $\overline{\widehat{\varphi}\left(  \kappa\right)  }%
=\widehat{\varphi}\left(  -\kappa\right)  $ and
\begin{equation}
\left\vert \widehat{\varphi}\left(  \kappa\right)  \right\vert ^{2}%
=\sum_{\boldsymbol{i}\in G_{l}}\left\vert \widehat{\varphi}\left(
\boldsymbol{i}\right)  \right\vert ^{2}\Omega\left(  p^{l}\left\Vert
\kappa-\boldsymbol{i}\right\Vert _{p}\right)  . \label{Eq_14}%
\end{equation}

(ii) The formulae
\begin{equation}
\widehat{\varphi}\left(  \boldsymbol{j}\right)  =p^{-lN}\sum_{\boldsymbol{i}%
\in G_{l}}\varphi\left(  \boldsymbol{i}\right)  \chi_{p}\left(
\boldsymbol{i\cdot j}\right)  \text{,}\ \ \varphi\left(  \boldsymbol{j}%
\right)  =p^{-lN}\sum_{\boldsymbol{i}\in G_{l}}\widehat{\varphi}\left(
\boldsymbol{i}\right)  \chi_{p}\left(  -\boldsymbol{i\cdot j}\right)
\label{Eq_14A}%
\end{equation}
give the discrete Fourier transform its inverse in the additive group $G_{l}$.
\end{remark}

\subsection{Lizorkin spaces of second kind}

The space%
\[
\mathcal{L}:=\mathcal{L}(\mathbb{Q}_{p}^{N})=\left\{  \varphi\in
\mathcal{D}(\mathbb{Q}_{p}^{N});%
{\textstyle\int\limits_{\mathbb{Q}_{p}^{N}}}
\varphi\left(  x\right)  d^{N}x=0\right\}
\]
is called \textit{the }$p$\textit{-adic Lizorkin space of second kind}. The
real Lizorkin space of second kind is $\mathcal{L}_{\mathbb{R}}:=\mathcal{L}%
_{\mathbb{R}}(\mathbb{Q}_{p}^{N})=\mathcal{L}(\mathbb{Q}_{p}^{N}%
)\cap\mathcal{D}_{\mathbb{R}}(\mathbb{Q}_{p}^{N})$. If%

\[
\mathcal{FL}:=\mathcal{FL}(\mathbb{Q}_{p}^{N})=\left\{  \widehat{\varphi}%
\in\mathcal{D}(\mathbb{Q}_{p}^{N});\widehat{\varphi}\left(  0\right)
=0\right\}  ,
\]
then the Fourier transform gives rise to an isomorphism of $\mathbb{C}$-vector
spaces from $\mathcal{L}$\ into $\mathcal{FL}$. The topological dual
$\mathcal{L}^{\prime}:=\mathcal{L}^{\prime}(\mathbb{Q}_{p}^{N})$\ of the space
$\mathcal{L}$ is called \textit{the }$p$\textit{-adic Lizorkin space of
distributions of second kind. }The real version is denoted as\textit{
}$\mathcal{L}_{_{\mathbb{R}}}^{\prime}:=\mathcal{L}_{_{\mathbb{R}}}^{\prime
}(\mathbb{Q}_{p}^{N})$\textit{. }

Let $\boldsymbol{A}(\partial)$ be a pseudo-differential operator defined as
\[
\boldsymbol{A}(\partial)\varphi\left(  x\right)  =\mathcal{F}_{\kappa
\rightarrow x}^{-1}(A(\left\Vert \kappa\right\Vert _{p})\mathcal{F}%
_{x\rightarrow\kappa}\mathcal{\varphi})\text{, for }\varphi\in\mathcal{D}%
_{\mathbb{R}}(\mathbb{Q}_{p}^{N}),
\]
where $A(\left\Vert \kappa\right\Vert _{p})$ is a real-valued and radial
function satisfying
\[
A(\left\Vert \kappa\right\Vert _{p})=0\text{ if and only if }\kappa=0\text{.}%
\]
Then, the Lizorkin space $\mathcal{L}_{\mathbb{R}}$ is invariant under
$\boldsymbol{A}(\partial)$. For further details about Lizorkin spaces and
pseudo-differential operators, the reader may consult \cite[Sections 7.3,
9.2]{A-K-S}.

We now define for $l\in\mathbb{N}\smallsetminus\left\{  0\right\}  $,
\[
\mathcal{L}^{l}:=\mathcal{L}^{l}(\mathbb{Q}_{p}^{N})=\left\{  \varphi\left(
x\right)  =\sum_{\boldsymbol{i}\in G_{l}}\varphi\left(  \boldsymbol{i}\right)
\Omega\left(  p^{l}\left\Vert x-\boldsymbol{i}\right\Vert _{p}\right)
,\varphi\left(  \boldsymbol{i}\right)  \in\mathbb{C};p^{-lN}\sum
_{\boldsymbol{i}\in G_{l}}\varphi\left(  \boldsymbol{i}\right)  =0\right\}  ,
\]
resp. $\mathcal{L}_{\mathbb{R}}^{l}:=\mathcal{L}_{\mathbb{R}}^{l}%
(\mathbb{Q}_{p}^{N})=\mathcal{L}^{l}\cap\mathcal{D}_{\mathbb{R}}^{l}$, and%
\[
\mathcal{FL}^{l}:=\mathcal{FL}^{l}(\mathbb{Q}_{p}^{N})=\left\{  \widehat
{\varphi}\left(  \kappa\right)  =\sum_{\boldsymbol{i}\in G_{l}}\widehat
{\varphi}\left(  \boldsymbol{i}\right)  \Omega\left(  p^{l}\left\Vert
\kappa-\boldsymbol{i}\right\Vert _{p}\right)  ,\widehat{\varphi}\left(
\boldsymbol{i}\right)  \in\mathbb{C};\widehat{\varphi}\left(  \boldsymbol{0}%
\right)  =0\right\}  ,
\]
By the formulae (\ref{Eq_14A}), the Fourier transform $\mathcal{F}%
:\mathcal{L}^{l}\rightarrow\mathcal{F}\mathcal{L}^{l}$ is an automorphism of
$\mathbb{C}$-vector spaces. The multiplication by the function $A(\left\Vert
\kappa\right\Vert _{p})$ gives rise to a linear transformation from
$\mathcal{L}^{l}$ onto itself. Consequently, $\boldsymbol{A}(\partial
):\mathcal{L}^{l}\rightarrow\mathcal{L}^{l}$ is a well-defined linear operator.

\subsection{Energy functionals in the momenta space}

By using (\ref{EQ_oper_W_def})-(\ref{EQ_oper_W_pseudo}), for $\varphi
\in\mathcal{D}_{\mathbb{R}}$, we have%
\[%
{\textstyle\iint\limits_{\mathbb{Q}_{p}^{N}\times\mathbb{Q}_{p}^{N}}}
\text{ }\frac{\left\{  \varphi\left(  x\right)  -\varphi\left(  y\right)
\right\}  ^{2}}{w_{\delta}\left(  \left\Vert x-y\right\Vert _{p}\right)
}d^{N}xd^{N}y=2%
{\textstyle\int\limits_{\mathbb{Q}_{p}^{N}}}
\text{ }\varphi\left(  x\right)  \left(  -\boldsymbol{W}_{\delta}\right)
\varphi\left(  x\right)  d^{N}x.
\]
Then
\begin{align*}
E_{0}(\varphi)  &  =\frac{\gamma}{2}\text{\ }%
{\textstyle\int\limits_{\mathbb{Q}_{p}^{N}}}
\text{ }\varphi\left(  x\right)  \left(  -\boldsymbol{W}_{\delta}\right)
\varphi\left(  x\right)  d^{N}x+\frac{\alpha_{2}}{2}%
{\textstyle\int\limits_{\mathbb{Q}_{p}^{N}}}
\varphi^{2}\left(  x\right)  d^{N}x\\
&  =\frac{\gamma}{2}\text{\ }%
{\textstyle\int\limits_{\mathbb{Q}_{p}^{N}}}
\text{ }\varphi\left(  x\right)  \boldsymbol{\boldsymbol{W}}\left(
\partial,\delta\right)  \varphi\left(  x\right)  d^{N}x+\frac{\alpha_{2}}{2}%
{\textstyle\int\limits_{\mathbb{Q}_{p}^{N}}}
\varphi^{2}\left(  x\right)  d^{N}x\\
&  =\frac{\gamma}{2}\text{\ }%
{\textstyle\int\limits_{\mathbb{Q}_{p}^{N}}}
\text{ }A_{w_{\delta}}(\left\Vert \kappa\right\Vert _{p})\left\vert
\widehat{\varphi}\left(  \kappa\right)  \right\vert ^{2}d^{N}\kappa
+\frac{\alpha_{2}}{2}%
{\textstyle\int\limits_{\mathbb{Q}_{p}^{N}}}
\left\vert \widehat{\varphi}\left(  \kappa\right)  \right\vert ^{2}d^{N}%
\kappa\\
&  =\text{\ }%
{\textstyle\int\limits_{\mathbb{Q}_{p}^{N}}}
\text{ }\left(  \frac{\gamma}{2}A_{w_{\delta}}(\left\Vert \kappa\right\Vert
_{p})+\frac{\alpha_{2}}{2}\right)  \left\vert \widehat{\varphi}\left(
\kappa\right)  \right\vert ^{2}d^{N}\kappa.
\end{align*}
Now, for $\varphi\in\mathcal{D}_{\mathbb{R}}^{l}$ by using (\ref{Eq_14}), we
have
\begin{align*}
E_{0}(\varphi)  &  =p^{-lN}%
{\textstyle\sum\limits_{\boldsymbol{j}\in G_{l}\smallsetminus\left\{
\boldsymbol{0}\right\}  }}
\text{\ }\left(  \frac{\gamma}{2}A_{w_{\delta}}(\left\Vert \boldsymbol{j}%
\right\Vert _{p})+\frac{\alpha_{2}}{2}\right)  \left\vert \widehat{\varphi
}\left(  \boldsymbol{j}\right)  \right\vert ^{2}\\
&  +\left\vert \widehat{\varphi}\left(  \boldsymbol{0}\right)  \right\vert
^{2}\left\{
{\textstyle\int\limits_{p^{l}\mathbb{Z}_{p}^{N}}}
\text{ }\left(  \frac{\gamma}{2}A_{w_{\delta}}(\left\Vert z\right\Vert
_{p})+\frac{\alpha_{2}}{2}\right)  d^{N}z\right\}  ,
\end{align*}
where $\widehat{\varphi}\left(  \boldsymbol{j}\right)  =$ $\widehat{\varphi
}_{1}\left(  \boldsymbol{j}\right)  +\sqrt{-1}\widehat{\varphi}_{2}\left(
\boldsymbol{j}\right)  \in\mathbb{C}$. Here we use the alternative notation
$\widehat{\varphi_{1}}\left(  \boldsymbol{j}\right)  =\operatorname{Re}\left(
\widehat{\varphi}\left(  \boldsymbol{j}\right)  \right)  $, $\widehat{\varphi
}_{2}\left(  \boldsymbol{j}\right)  =\operatorname{Im}\left(  \widehat
{\varphi}\left(  \boldsymbol{j}\right)  \right)  $ which more convenient for us.

\begin{remark}
Notice that
\[
\mathcal{FL}_{\mathbb{R}}^{l}=\left\{  \widehat{\varphi}\left(  \kappa\right)
=\sum_{\boldsymbol{i}\in G_{l}}\widehat{\varphi}\left(  \boldsymbol{i}\right)
\Omega\left(  p^{l}\left\Vert \kappa-\boldsymbol{i}\right\Vert _{p}\right)
,\widehat{\varphi}\left(  \boldsymbol{i}\right)  \in\mathbb{C};\widehat
{\varphi}\left(  0\right)  =0,\text{ }\overline{\widehat{\varphi}\left(
\kappa\right)  }=\widehat{\varphi}\left(  -\kappa\right)  \right\}  ,
\]
and that the condition $\overline{\widehat{\varphi}\left(  \kappa\right)
}=\widehat{\varphi}\left(  -\kappa\right)  $ implies that $\widehat{\varphi
}_{1}\left(  -\boldsymbol{i}\right)  =\widehat{\varphi}_{1}\left(
\boldsymbol{i}\right)  $\ and $\widehat{\varphi}_{2}\left(  -\boldsymbol{i}%
\right)  =-\widehat{\varphi}_{2}\left(  \boldsymbol{i}\right)  $ for any
$\boldsymbol{i}\in G_{l}$. This implies that $\mathcal{FL}_{\mathbb{R}}^{l}$
is $\mathbb{R}$-vector space of dimension $\#G_{l}$ $-1$.
\end{remark}

\begin{remark}
\label{Nota_Basis}We set $G_{l}\smallsetminus\left\{  \boldsymbol{0}\right\}
:=G_{l}^{+}\bigsqcup G_{l}^{-}$, where the subsets $G_{l}^{+}$, $G_{l}^{-}$
satisfy that
\[%
\begin{array}
[c]{lll}%
G_{l}^{+} & \rightarrow & G_{l}^{-}\\
\boldsymbol{i} & \rightarrow & -\boldsymbol{i}%
\end{array}
\]
is a bijection. We recall here that $G_{l}$ is a finite additive group. Since
$\#G_{l}^{+}=\#G_{l}^{-}$ necessarily $\#\left(  G_{l}\smallsetminus\left\{
\boldsymbol{0}\right\}  \right)  =p^{2Nl}-1$ is even, and thus $p\geq3$. Then
any function from $\mathcal{FL}_{\mathbb{R}}^{l}$ can be uniquely represented
as
\[
\widehat{\varphi}\left(  \kappa\right)  =\sum_{\boldsymbol{i}\in G_{l}^{+}%
}\widehat{\varphi}_{1}\left(  \boldsymbol{i}\right)  \Omega_{+}\left(
p^{l}\left\Vert \kappa-\boldsymbol{i}\right\Vert _{p}\right)  +\widehat
{\varphi}_{2}\left(  \boldsymbol{i}\right)  \Omega_{-}\left(  p^{l}\left\Vert
\kappa-\boldsymbol{i}\right\Vert _{p}\right)  ,
\]
where
\[
\Omega_{+}\left(  p^{l}\left\Vert \kappa-\boldsymbol{i}\right\Vert
_{p}\right)  :=\Omega\left(  p^{l}\left\Vert \kappa-\boldsymbol{i}\right\Vert
_{p}\right)  +\Omega\left(  p^{l}\left\Vert \kappa+\boldsymbol{i}\right\Vert
_{p}\right)  ,
\]
and
\[
\Omega_{-}\left(  p^{l}\left\Vert \kappa-\boldsymbol{i}\right\Vert
_{p}\right)  :=\sqrt{-1}\left\{  \Omega\left(  p^{l}\left\Vert \kappa
-\boldsymbol{i}\right\Vert _{p}\right)  -\Omega\left(  p^{l}\left\Vert
\kappa+\boldsymbol{i}\right\Vert _{p}\right)  \right\}  .
\]

\end{remark}

We take\ $\varphi\in\mathcal{L}_{\mathbb{R}}^{l}$, then $\widehat{\varphi
}\left(  0\right)  =0$, and
\begin{align*}
E_{0}^{\left(  l\right)  }(\varphi)  &  =p^{-lN}%
{\textstyle\sum\limits_{\boldsymbol{j}\in G_{l}\smallsetminus\left\{
\boldsymbol{0}\right\}  }}
\text{\ }\left(  \frac{\gamma}{2}A_{w_{\delta}}(\left\Vert \boldsymbol{j}%
\right\Vert _{p})+\frac{\alpha_{2}}{2}\right)  \left(  \widehat{\varphi_{1}%
}^{2}\left(  \boldsymbol{j}\right)  +\widehat{\varphi_{2}}^{2}\left(
\boldsymbol{j}\right)  \right) \\
&  =2p^{-lN}%
{\textstyle\sum\limits_{r\in\left\{  1,2\right\}  }}
\text{ \ }%
{\textstyle\sum\limits_{\boldsymbol{j}\in G_{l}^{+}}}
\text{\ }\left(  \frac{\gamma}{2}A_{w_{\delta}}(\left\Vert j\right\Vert
_{p})+\frac{\alpha_{2}}{2}\right)  \widehat{\varphi_{r}}^{2}\left(
\boldsymbol{j}\right)  .
\end{align*}
By using that $\mathcal{L}_{\mathbb{R}}^{l}\simeq\mathcal{FL}_{\mathbb{R}}%
^{l}$ we get that $E_{0}^{\left(  l\right)  }$ is a real-valued functional
defined on $\mathcal{FL}_{\mathbb{R}}^{l}\simeq\mathbb{R}^{\left(
\#G_{l}-1\right)  }$.

We now define\ the diagonal matrix $B^{\left(  r\right)  }=\left[
B_{\boldsymbol{i},\boldsymbol{j}}^{\left(  r\right)  }\right]
_{\boldsymbol{i},\boldsymbol{j}\in G_{l}^{+}}$, $r=1$, $2$, where%
\[
B_{\boldsymbol{i},\boldsymbol{j}}^{\left(  r\right)  }:=\left\{
\begin{array}
[c]{lll}%
\frac{\gamma}{2}A_{w_{\delta}}(\left\Vert \boldsymbol{j}\right\Vert
_{p})+\frac{\alpha_{2}}{2} & \text{if} & \boldsymbol{i}=\boldsymbol{j}\\
&  & \\
0 & \text{if} & \boldsymbol{i}\neq\boldsymbol{j}.
\end{array}
\right.
\]
Notice that $B_{\boldsymbol{i},\boldsymbol{j}}^{\left(  1\right)
}=B_{\boldsymbol{i},\boldsymbol{j}}^{\left(  2\right)  }$. We set%
\begin{equation}
B(l):=B(l,\delta,\gamma,\alpha_{2})=\left[
\begin{array}
[c]{ll}%
B^{\left(  1\right)  } & \boldsymbol{0}\\
\boldsymbol{0} & B^{\left(  2\right)  }.
\end{array}
\right]  \label{Eq_Matrix_B_1}%
\end{equation}
The matrix $B=\left[  B_{\boldsymbol{i},\boldsymbol{j}}\right]  $ is a
diagonal of size $2\left(  \#G_{l}^{+}\right)  \times2\left(  \#G_{l}%
^{+}\right)  $. In addition, the indices $\boldsymbol{i},\boldsymbol{j}$ run
through two disjoint copies of $G_{l}^{+}$. Then we have the following result:

\begin{lemma}
\label{Lemma10}Assume that $\alpha_{2}\geq0$. With the above notation the
following formula holds true:%
\begin{equation}
E_{0}^{\left(  l\right)  }(\varphi):=E_{0}^{\left(  l\right)  }\left(
\widehat{\varphi_{1}}\left(  \boldsymbol{j}\right)  ,\widehat{\varphi_{2}%
}\left(  \boldsymbol{j}\right)  ;\boldsymbol{j}\in G_{l}^{+}\right)  =\left[
\begin{array}
[c]{l}%
\left[  \widehat{\varphi_{1}}\left(  \boldsymbol{j}\right)  \right]
_{\boldsymbol{j}\in G_{l}^{+}}\\
\left[  \widehat{\varphi_{2}}\left(  \boldsymbol{j}\right)  \right]
_{\boldsymbol{j}\in G_{l}^{+}}%
\end{array}
\right]  ^{T}2p^{-lN}B(l)\left[
\begin{array}
[c]{l}%
\left[  \widehat{\varphi_{1}}\left(  \boldsymbol{j}\right)  \right]
_{\boldsymbol{j}\in G_{l}^{+}}\\
\left[  \widehat{\varphi_{2}}\left(  \boldsymbol{j}\right)  \right]
_{\boldsymbol{j}\in G_{l}^{+}}%
\end{array}
\right]  \geq0, \label{Eq_E_0_Momenta}%
\end{equation}
for $\varphi\in\mathcal{L}_{\mathbb{R}}^{l}\simeq\mathcal{FL}_{\mathbb{R}}%
^{l}\simeq\mathbb{R}^{\left(  \#G_{l}\text{ }-1\right)  }$, where
$2p^{-lN}B(l)$ is a diagonal, positive definite, invertible matrix.
\end{lemma}

\section{Gaussian measures}

We recall that we are taking $\delta>N$, $\gamma>0$, $\alpha_{2}\geq0$. We
define the partition function attached to the energy functional $E_{0}$ as
\[
\mathcal{Z}(\delta,\gamma,\alpha_{2})=\int\limits_{\mathcal{FL}_{\mathbb{R}%
}(\mathbb{Q}_{p}^{N})}D(\varphi)e^{-E_{0}\left(  \varphi\right)  },
\]
where $D(\varphi)$ is a \textquotedblleft spurious measure\textquotedblright%
\ on $\mathcal{FL}_{\mathbb{R}}(\mathbb{Q}_{p}^{N})$. For the sake of
simplicity we use the notation $\mathcal{Z}=\mathcal{Z}(\delta,\gamma
,\alpha_{2})$ wherever possible.

As the discrete version of $\mathcal{Z}(\delta,\gamma,\alpha_{2})$ we take
\[
\mathcal{Z}^{\left(  l\right)  }(\delta,\gamma,\alpha_{2}):=\int
\limits_{\mathcal{FL}_{\mathbb{R}}^{l}(\mathbb{Q}_{p}^{N})}D_{l}\left(
\varphi\right)  e^{-E_{0}\left(  \varphi\right)  }.
\]
We also use the notation $\mathcal{Z}^{\left(  l\right)  }=\mathcal{Z}%
^{\left(  l\right)  }(\delta,\gamma,\alpha_{2})$. Now we define $\mathcal{Z}%
^{\left(  l\right)  }(\delta,\gamma,\alpha_{2})$ as
\begin{multline*}
\mathcal{Z}^{\left(  l\right)  }(\delta,\gamma,\alpha_{2})=\\
\int\limits_{\mathbb{R}^{\left(  p^{2lN}-1\right)  }}\exp\left(  -\left[
\begin{array}
[c]{l}%
\left[  \widehat{\varphi_{1}}\left(  \boldsymbol{j}\right)  \right]
_{\boldsymbol{j}\in G_{l}^{+}}\\
\left[  \widehat{\varphi_{2}}\left(  \boldsymbol{j}\right)  \right]
_{\boldsymbol{j}\in G_{l}^{+}}%
\end{array}
\right]  ^{T}2p^{-lN}B(l)\left[
\begin{array}
[c]{l}%
\left[  \widehat{\varphi_{1}}\left(  \boldsymbol{j}\right)  \right]
_{\boldsymbol{j}\in G_{l}^{+}}\\
\left[  \widehat{\varphi_{2}}\left(  \boldsymbol{j}\right)  \right]
_{\boldsymbol{j}\in G_{l}^{+}}%
\end{array}
\right]  \right)  \prod\limits_{\boldsymbol{i}\in G_{l}^{+}}d\widehat
{\varphi_{1}}\left(  \boldsymbol{i}\right)  d\widehat{\varphi_{2}}\left(
\boldsymbol{i}\right)  ,
\end{multline*}
where $%
{\textstyle\prod\nolimits_{\boldsymbol{i}\in G_{l}^{+}}}
d\widehat{\varphi_{1}}\left(  \boldsymbol{i}\right)  d\widehat{\varphi_{2}%
}\left(  \boldsymbol{i}\right)  $ is the Lebesgue measure on $\mathbb{R}%
^{\left(  p^{2lN}-1\right)  }$.

The integral $\mathcal{Z}^{\left(  l\right)  }$ is the natural discretization
of $\mathcal{Z}$. From a classical point of view, one should expect that
$\mathcal{Z}=\lim_{l\rightarrow\infty}\mathcal{Z}^{\left(  l\right)  }$ in
some sense. The goal of this section is to study these matters in a rigorous
mathematical way. Our main result is the construction of rigorous mathematical
version of the spurious measure $D(\varphi)$.

By Lemma \ref{Lemma10}, $\mathcal{Z}^{\left(  l\right)  }$ is a Gaussian
integral, then%
\[
\mathcal{Z}^{\left(  l\right)  }=\frac{\left(  2\pi\right)  ^{\frac{\left(
p^{2lN}-1\right)  }{2}}}{\sqrt{\det4p^{-lN}B(l)}}=\left(  \frac{\pi}%
{2}\right)  ^{\frac{\left(  p^{2lN}-1\right)  }{2}}\frac{p^{\frac{lN\left(
p^{2lN}-1\right)  }{2}}}{\sqrt{\det B}}.
\]

\begin{definition}
We define the following family of Gaussian measures:%
\begin{gather}
d\mathbb{P}_{l}\left(  \left[
\begin{array}
[c]{l}%
\left[  \widehat{\varphi_{1}}\left(  \boldsymbol{j}\right)  \right]
_{\boldsymbol{j}\in G_{l}^{+}}\\
\left[  \widehat{\varphi_{2}}\left(  \boldsymbol{j}\right)  \right]
_{\boldsymbol{j}\in G_{l}^{+}}%
\end{array}
\right]  ;\delta,\gamma,\alpha_{2}\right)  :=d\mathbb{P}_{l}\left(  \left[
\begin{array}
[c]{l}%
\left[  \widehat{\varphi_{1}}\left(  \boldsymbol{j}\right)  \right]
_{\boldsymbol{j}\in G_{l}^{+}}\\
\left[  \widehat{\varphi_{2}}\left(  \boldsymbol{j}\right)  \right]
_{\boldsymbol{j}\in G_{l}^{+}}%
\end{array}
\right]  \right) \nonumber\\
=\frac{1}{\mathcal{Z}^{\left(  l\right)  }}\exp(-\left[
\begin{array}
[c]{l}%
\left[  \widehat{\varphi_{1}}\left(  \boldsymbol{j}\right)  \right]
_{\boldsymbol{j}\in G_{l}^{+}}\\
\left[  \widehat{\varphi_{2}}\left(  \boldsymbol{j}\right)  \right]
_{\boldsymbol{j}\in G_{l}^{+}}%
\end{array}
\right]  ^{T}2p^{-lN}B(l)\left[
\begin{array}
[c]{l}%
\left[  \widehat{\varphi_{1}}\left(  \boldsymbol{j}\right)  \right]
_{\boldsymbol{j}\in G_{l}^{+}}\\
\left[  \widehat{\varphi_{2}}\left(  \boldsymbol{j}\right)  \right]
_{\boldsymbol{j}\in G+_{l}}%
\end{array}
\right]  )\prod\limits_{\boldsymbol{i}\in G_{l}^{+}}d\widehat{\varphi_{1}%
}\left(  \boldsymbol{i}\right)  d\widehat{\varphi_{2}}\left(  \boldsymbol{i}%
\right)  \label{Eq_P_L}%
\end{gather}
in $\mathcal{FL}_{\mathbb{R}}^{l}\simeq\mathbb{R}^{\left(  p^{2lN}-1\right)
}$, for $l\in\mathbb{N}\smallsetminus\left\{  0\right\}  $.
\end{definition}

Thus for any Borel subset $A$ of $\mathbb{R}^{\left(  p^{2lN}-1\right)
}\simeq\mathcal{FL}_{\mathbb{R}}^{l}$ and any continuous and bounded function
$f:\mathcal{FL}_{\mathbb{R}}^{l}\rightarrow\mathbb{R}$ the integral%
\[%
{\textstyle\int\limits_{A}}
f\left(  \left[
\begin{array}
[c]{l}%
\left[  \widehat{\varphi_{1}}\left(  \boldsymbol{j}\right)  \right]
_{\boldsymbol{j}\in G_{l}^{+}}\\
\left[  \widehat{\varphi_{2}}\left(  \boldsymbol{j}\right)  \right]
_{\boldsymbol{j}\in G_{l}^{+}}%
\end{array}
\right]  \right)  d\mathbb{P}_{l}\left(  \left[
\begin{array}
[c]{l}%
\left[  \widehat{\varphi_{1}}\left(  \boldsymbol{j}\right)  \right]
_{\boldsymbol{j}\in G_{l}^{+}}\\
\left[  \widehat{\varphi_{2}}\left(  \boldsymbol{j}\right)  \right]
_{\boldsymbol{j}\in G_{l}^{+}}%
\end{array}
\right]  \right)  =:%
{\textstyle\int\limits_{A}}
f\left(  \widehat{\varphi}\right)  d\mathbb{P}_{l}\left(  \widehat{\varphi
}\right)
\]
is well-defined.

We define $\mathcal{I}=\cup_{l\in\mathbb{N}\setminus\left\{  0\right\}  }%
G_{l}^{+}$. Then $\mathcal{I}$ is a countable set. Given a finite subset $J$
of $\mathcal{I}$ there is $l_{0}\in\mathbb{N}\setminus\left\{  0\right\}  $
such that $G_{l_{0}}^{+}$ is the smallest set of the form $G_{l}^{+}%
$\ containing $J$. To each finite subset $J$ of $\mathcal{I}$ we attach a
collection of Gaussian random variables
\[
\left.  \left[
\begin{array}
[c]{l}%
\left[  \widehat{\varphi_{1}}\left(  \boldsymbol{j}\right)  \right]
_{\boldsymbol{j}\in G_{l_{0}}^{+}}\\
\left[  \widehat{\varphi_{2}}\left(  \boldsymbol{j}\right)  \right]
_{\boldsymbol{j}\in G_{l_{0}}^{+}}%
\end{array}
\right]  \right\vert _{\widehat{\varphi_{2}}\left(  \boldsymbol{j}\right)
=\widehat{\varphi_{1}}\left(  \boldsymbol{j}\right)  =0,\boldsymbol{j}\in
G_{l_{0}}^{+}\setminus J}%
\]
having joint probability distribution%
\[
\mathbb{P}_{J}=\left.  \mathbb{P}_{l_{0}}\left(  \left[
\begin{array}
[c]{l}%
\left[  \widehat{\varphi_{1}}\left(  \boldsymbol{j}\right)  \right]
_{\boldsymbol{j}\in G_{l_{0}}^{+}}\\
\left[  \widehat{\varphi_{2}}\left(  \boldsymbol{j}\right)  \right]
_{\boldsymbol{j}\in G_{l_{0}}^{+}}%
\end{array}
\right]  \right)  \right\vert _{\widehat{\varphi_{2}}\left(  \boldsymbol{j}%
\right)  =\widehat{\varphi_{1}}\left(  \boldsymbol{j}\right)
=0,\boldsymbol{j}\in G_{l_{0}}^{+}\setminus J}.
\]

Notice that $\mathbb{P}_{G_{l}^{+}}\left(  \left[
\begin{array}
[c]{l}%
\left[  \widehat{\varphi_{1}}\left(  \boldsymbol{j}\right)  \right]
_{\boldsymbol{j}\in G_{l}^{+}}\\
\left[  \widehat{\varphi_{2}}\left(  \boldsymbol{j}\right)  \right]
_{\boldsymbol{j}\in G_{l}^{+}}%
\end{array}
\right]  \right)  =\mathbb{P}_{l}\left(  \left[
\begin{array}
[c]{l}%
\left[  \widehat{\varphi_{1}}\left(  \boldsymbol{j}\right)  \right]
_{\boldsymbol{j}\in G_{l}^{+}}\\
\left[  \widehat{\varphi_{2}}\left(  \boldsymbol{j}\right)  \right]
_{\boldsymbol{j}\in G_{l}^{+}}%
\end{array}
\right]  \right)  $. The family of Gaussian measures $\left\{  \mathbb{P}%
_{J}\left(  \left[
\begin{array}
[c]{l}%
\left[  \widehat{\varphi_{1}}\left(  \boldsymbol{j}\right)  \right]
_{\boldsymbol{j}\in G_{l}^{+}}\\
\left[  \widehat{\varphi_{2}}\left(  \boldsymbol{j}\right)  \right]
_{\boldsymbol{j}\in G_{l}^{+}}%
\end{array}
\right]  \right)  ;J\subset\mathcal{I}\right\}  $ is consistent, i.e.
$\mathbb{P}_{J}(A)=\mathbb{P}_{K}(A\times\mathbb{R}^{\#K-\#J})$, for $J\subset
K$, see e.g. \cite[Chapter IV, Section 3.1, Lemma 1]{Gelfand-Vilenkin}. We now
apply Kolmogorov's consistency theorem and its proof, see e.g. \cite[Theorem
2.1]{Simon-1}, to obtain the following result:

\begin{lemma}
\label{Lemma11}There exists a probability measure space $\left(
X,\mathcal{F},\mathbb{P}\right)  $ and random variables
\[
\left[
\begin{array}
[c]{l}%
\left[  \widehat{\varphi_{1}}\left(  \boldsymbol{j}\right)  \right]
_{\boldsymbol{j}\in G_{l}^{+}}\\
\left[  \widehat{\varphi_{2}}\left(  \boldsymbol{j}\right)  \right]
_{\boldsymbol{j}\in G_{l}^{+}}%
\end{array}
\right]  \text{, for }l\in\mathbb{N}\smallsetminus\left\{  0\right\}  \text{,}%
\]
such that $\mathbb{P}_{l}$ is the joint probability distribution of\ $\left[
\begin{array}
[c]{l}%
\left[  \widehat{\varphi_{1}}\left(  \boldsymbol{j}\right)  \right]
_{\boldsymbol{j}\in G_{l}^{+}}\\
\left[  \widehat{\varphi_{2}}\left(  \boldsymbol{j}\right)  \right]
_{\boldsymbol{j}\in G_{l}^{+}}%
\end{array}
\right]  $. The space $\left(  X,\mathcal{F},\mathbb{P}\right)  $ is unique up
to isomorphisms of probability measure spaces. Furthermore, for any bounded
continuous function $f$ supported in $\mathcal{FL}_{\mathbb{R}}^{l}$, we have%
\[%
{\textstyle\int\limits_{\mathcal{FL}_{\mathbb{R}}^{l}}}
f\left(  \widehat{\varphi}\right)  d\mathbb{P}_{l}\left(  \widehat{\varphi
}\right)  =%
{\textstyle\int\limits_{\mathcal{FL}_{\mathbb{R}}^{l}}}
f\left(  \widehat{\varphi}\right)  d\mathbb{P}\left(  \widehat{\varphi
}\right)  .
\]

\end{lemma}

\subsection{A quick detour into the $p$-adic noise calculus}

In this section we introduce a Gel'fand triple and construct \ some Gaussian
measures in the non-Archimedean setting.

\subsubsection{A bilinear form in $\mathcal{D}_{\mathbb{R}}\left(
\mathbb{Q}_{p}^{N}\right)  $}

For $\delta>N$, $\gamma$, $\alpha_{2}>0$, we define the operator
\[%
\begin{array}
[c]{lll}%
\mathcal{D}\left(  \mathbb{Q}_{p}^{N}\right)  & \rightarrow & L^{2}\left(
\mathbb{Q}_{p}^{N}\right) \\
&  & \\
\varphi & \rightarrow & \left(  \frac{\gamma}{2}\boldsymbol{W}\left(
\partial,\delta\right)  +\frac{\alpha_{2}}{2}\right)  ^{-1}\varphi,
\end{array}
\]
where $\left(  \frac{\gamma}{2}\boldsymbol{W}\left(  \partial,\delta\right)
+\frac{\alpha_{2}}{2}\right)  ^{-1}\varphi\left(  x\right)  :=\mathcal{F}%
_{\kappa\rightarrow x}^{-1}\left(  \frac{\mathcal{F}_{x\rightarrow\kappa
}\varphi}{\frac{\gamma}{2}A_{w_{\delta}}(\left\Vert \kappa\right\Vert
_{p})+\frac{\alpha_{2}}{2}}\right)  $.

We define the distribution%
\[
G(x):=G(x;\delta,\gamma,\alpha_{2})=\mathcal{F}_{\kappa\rightarrow x}%
^{-1}\left(  \frac{1}{\frac{\gamma}{2}A_{w_{\delta}}(\left\Vert \kappa
\right\Vert _{p})+\frac{\alpha_{2}}{2}}\right)  \in\mathcal{D}^{\prime}\left(
\mathbb{Q}_{p}^{N}\right)  .
\]
By using the fact that $\frac{1}{\frac{\gamma}{2}A_{w_{\delta}}(\left\Vert
\kappa\right\Vert _{p})+\frac{\alpha_{2}}{2}}$ is radial and $(\mathcal{F}%
(\mathcal{F}\varphi))(\kappa)=\varphi(-\kappa)$ one \ verifies that%
\[
G(x)\in\mathcal{D}_{\mathbb{R}}^{\prime}\left(  \mathbb{Q}_{p}^{N}\right)  .
\]
Now we define the following bilinear form $\mathbb{B}:=\mathbb{B}%
(\delta,\gamma,\alpha_{2})$:%
\[%
\begin{array}
[c]{lll}%
\mathbb{B}:\mathcal{D}_{\mathbb{R}}\left(  \mathbb{Q}_{p}^{N}\right)
\times\mathcal{D}_{\mathbb{R}}\left(  \mathbb{Q}_{p}^{N}\right)  & \rightarrow
& \mathbb{R}\\
&  & \\
\left(  \varphi,\theta\right)  & \rightarrow & \left\langle \varphi,\left(
\frac{\gamma}{2}\boldsymbol{W}\left(  \partial,\delta\right)  +\frac
{\alpha_{2}}{2}\right)  ^{-1}\theta\right\rangle ,
\end{array}
\]
where $\left\langle \cdot,\cdot\right\rangle $ denotes the scalar product in
$L^{2}\left(  \mathbb{Q}_{p}^{N}\right)  $.

\begin{lemma}
\label{Lemma12}$\mathbb{B}$ is a positive, continuous bilinear form from
$\mathcal{D}_{\mathbb{R}}\left(  \mathbb{Q}_{p}^{N}\right)  \times
\mathcal{D}_{\mathbb{R}}\left(  \mathbb{Q}_{p}^{N}\right)  $\ into
$\mathbb{R}$.
\end{lemma}

\begin{proof}
We first notice that for $\varphi\in\mathcal{D}_{\mathbb{R}}\left(
\mathbb{Q}_{p}^{N}\right)  $, we have%
\[
\mathbb{B}(\varphi,\varphi)=%
{\textstyle\int\limits_{\mathbb{Q}_{p}^{N}}}
\frac{\left\vert \widehat{\varphi}\left(  \kappa\right)  \right\vert ^{2}%
d^{N}\kappa}{\frac{\gamma}{2}A_{w_{\delta}}(\left\Vert \kappa\right\Vert
_{p})+\frac{\alpha_{2}}{2}}\geq0.
\]
Then $\mathbb{B}(\varphi,\varphi)=0$ implies that $\varphi$ is zero almost
everywhere. Since $\varphi$ is a locally constant function, $\mathbb{B}%
(\varphi,\varphi)=0$ if and only if $\varphi=0$.

For $\left(  \varphi,\theta\right)  \in\mathcal{D}_{\mathbb{R}}\left(
\mathbb{Q}_{p}^{N}\right)  \times\mathcal{D}_{\mathbb{R}}\left(
\mathbb{Q}_{p}^{N}\right)  $, the Cauchy-Schwarz inequality implies that%
\begin{equation}
\left\vert \mathbb{B}\left(  \varphi,\theta\right)  \right\vert \leq\left\Vert
\varphi\right\Vert _{2}\left(
{\textstyle\int\limits_{\mathbb{Q}_{p}^{N}}}
\frac{\left\vert \widehat{\theta}\left(  \kappa\right)  \right\vert ^{2}%
d^{N}\kappa}{\left(  \frac{\gamma}{2}A_{w_{\delta}}(\left\Vert \kappa
\right\Vert _{p})+\frac{\alpha_{2}}{2}\right)  ^{2}}\right)  ^{\frac{1}{2}%
}\leq\frac{2}{\alpha_{2}}\left\Vert \varphi\right\Vert _{2}\left\Vert
\theta\right\Vert _{2}. \label{Eq_19}%
\end{equation}
Now take two sequences in $\mathcal{D}_{\mathbb{R}}\left(  \mathbb{Q}_{p}%
^{N}\right)  $ such that $\varphi_{n}$ $\underrightarrow{\mathcal{D}%
_{\mathbb{R}}}$ $\varphi$ and $\theta_{n}$ $\underrightarrow{\mathcal{D}%
_{\mathbb{R}}}$ $\theta$ with $\varphi$, $\theta\in\mathcal{D}_{\mathbb{R}%
}\left(  \mathbb{Q}_{p}^{N}\right)  $. We recall that the convergence of these
sequences means that there is an positive integer $l$ such that $\varphi_{n}$,
$\varphi$, $\theta_{n}$, $\theta\in\mathcal{D}_{\mathbb{R}}^{l}$, and
\[
\varphi_{n}-\varphi\text{ }\underrightarrow{\text{unif.}}\text{ }0\text{ \ and
\ }\theta_{n}-\theta\text{ }\underrightarrow{\text{unif.}}\text{ }0\text{ in
}p^{-l}\mathbb{Z}_{p}^{N}.
\]
Then
\begin{align*}
\varphi_{n}\left(  x\right)  -\varphi\left(  x\right)   &  =%
{\textstyle\sum\limits_{\boldsymbol{i}\in G_{l}}}
\left(  \varphi_{n}\left(  \boldsymbol{i}\right)  -\varphi\left(
\boldsymbol{i}\right)  \right)  \Omega\left(  p^{l}\left\Vert x-\boldsymbol{i}%
\right\Vert _{p}\right)  \text{, and}\\
\theta_{n}\left(  x\right)  -\theta\left(  x\right)   &  =%
{\textstyle\sum\limits_{\boldsymbol{i}\in G_{l}}}
\left(  \theta_{n}\left(  \boldsymbol{i}\right)  -\theta\left(  \boldsymbol{i}%
\right)  \right)  \Omega\left(  p^{l}\left\Vert x-\boldsymbol{i}\right\Vert
_{p}\right)
\end{align*}
and by (\ref{Eq_19}),%
\begin{gather*}
\left\vert \mathbb{B}\left(  \varphi_{n}-\varphi,\theta_{n}-\theta\right)
\right\vert \leq\frac{2}{\alpha_{2}}\left\Vert \varphi_{n}-\varphi\right\Vert
_{2}\left\Vert \theta_{n}-\theta\right\Vert _{2}\\
\leq\frac{2p^{-lN}}{\alpha_{2}}\sqrt{%
{\textstyle\sum\limits_{\boldsymbol{i}\in G_{l}}}
\left\vert \varphi_{n}\left(  \boldsymbol{i}\right)  -\varphi\left(
\boldsymbol{i}\right)  \right\vert ^{2}}\sqrt{%
{\textstyle\sum\limits_{\boldsymbol{i}\in G_{l}}}
\left\vert \theta_{n}\left(  \boldsymbol{i}\right)  -\theta\left(
\boldsymbol{i}\right)  \right\vert ^{2}}\\
\leq\frac{2p^{-lN}\#G_{l}}{\alpha_{2}}\left(  \max_{\boldsymbol{i}\in G_{l}%
}\left\vert \varphi_{n}\left(  \boldsymbol{i}\right)  -\varphi\left(
\boldsymbol{i}\right)  \right\vert \right)  \left(  \max_{\boldsymbol{i}\in
G_{l}}\left\vert \theta_{n}\left(  \boldsymbol{i}\right)  -\theta\left(
\boldsymbol{i}\right)  \right\vert \right)  \rightarrow0
\end{gather*}
as $n\rightarrow\infty$. This fact implies the continuity of $\mathbb{B}$ in
$\mathcal{D}_{\mathbb{R}}\left(  \mathbb{Q}_{p}^{N}\right)  \times
\mathcal{D}_{\mathbb{R}}\left(  \mathbb{Q}_{p}^{N}\right)  $.
\end{proof}

In the next sections we only use the restriction of $\mathbb{B}$ to
$\mathcal{L}_{\mathbb{R}}\left(  \mathbb{Q}_{p}^{N}\right)  \times
\mathcal{L}_{\mathbb{R}}\left(  \mathbb{Q}_{p}^{N}\right)  $.

\begin{lemma}
\label{Lemma13}For $\varphi\in\mathcal{L}_{\mathbb{R}}^{l}\simeq
\mathcal{FL}_{\mathbb{R}}^{l}$,%
\[
\mathbb{B}_{l}(\varphi,\varphi):=\mathbb{B}(\varphi,\varphi)=\left[
\begin{array}
[c]{l}%
\left[  \widehat{\varphi_{1}}\left(  \boldsymbol{j}\right)  \right]
_{\boldsymbol{j}\in G_{l}^{+}}\\
\left[  \widehat{\varphi_{2}}\left(  \boldsymbol{j}\right)  \right]
_{\boldsymbol{j}\in G_{l}^{+}}%
\end{array}
\right]  ^{T}2p^{-lN}B^{-1}(l)\left[
\begin{array}
[c]{l}%
\left[  \widehat{\varphi_{1}}\left(  \boldsymbol{j}\right)  \right]
_{\boldsymbol{j}\in G_{l}^{+}}\\
\left[  \widehat{\varphi_{2}}\left(  \boldsymbol{j}\right)  \right]
_{\boldsymbol{j}\in G_{l}^{+}}%
\end{array}
\right]  ,
\]
where $B(l)$ is the matrix defined in (\ref{Eq_Matrix_B_1}).
\end{lemma}

\begin{proof}
The proof is similar to the proof of Lemma \ref{Lemma10}. We first notice
that
\[
\mathbb{B}(\varphi,\varphi)=%
{\textstyle\int\limits_{\mathbb{Q}_{p}^{N}}}
\text{ }\frac{\left\vert \widehat{\varphi}\left(  \kappa\right)  \right\vert
^{2}}{\frac{\gamma}{2}A_{w_{\delta}}(\left\Vert \kappa\right\Vert _{p}%
)+\frac{\alpha_{2}}{2}}d^{N}\kappa.
\]
By using (\ref{Eq_14}), we get that%
\begin{equation}
\mathbb{B}_{l}(\varphi,\varphi)=2p^{-lN}%
{\textstyle\sum\limits_{r\in\left\{  1,2\right\}  }}
\
{\textstyle\sum\limits_{\boldsymbol{j}\in G_{l}^{+}}}
\ \frac{\widehat{\varphi_{r}}^{2}\left(  \boldsymbol{j}\right)  }{\frac
{\gamma}{2}A_{w_{\delta}}(\left\Vert \boldsymbol{j}\right\Vert _{p}%
)+\frac{\alpha_{2}}{2}}. \label{Eq_18}%
\end{equation}
Now, the announced formula follows from (\ref{Eq_18}).
\end{proof}

Given a finite dimensional subspace $\mathcal{Y}\subset\mathcal{L}%
_{\mathbb{R}}(\mathbb{Q}_{p}^{N})$, we denote by $\mathbb{B}_{\mathcal{Y}}$
the restriction of $\mathbb{B}$ to $\mathcal{Y}\times\mathcal{Y}$. In the case
$\mathcal{Y}=\mathcal{L}_{\mathbb{R}}^{l}$, we use the notation $\mathbb{B}%
_{l}$, which agrees with the notation introduced in Lemma \ref{Lemma13}.

\begin{lemma}
\label{Lemma14}Given finite dimensional subspace $\mathcal{Y}\subset
\mathcal{L}_{\mathbb{R}}(\mathbb{Q}_{p}^{N})$, there is a positive integer
$l=l(\mathcal{Y})$ such that $\mathcal{Y}\subset\mathcal{L}_{\mathbb{R}}%
^{l}\simeq\mathcal{FL}_{\mathbb{R}}^{l}$, and there is a subset
$J=J(\mathcal{Y})\subset G_{l}^{+}$ such that
\begin{equation}
\mathbb{B}_{\mathcal{Y}}(\varphi,\varphi)=2p^{-lN}%
{\textstyle\sum\limits_{r\in\left\{  1,2\right\}  }}
\
{\textstyle\sum\limits_{\boldsymbol{j}\in J}}
\ \frac{\widehat{\varphi}_{r}^{2}\left(  \boldsymbol{j}\right)  }{\frac
{\gamma}{2}A_{w_{\delta}}(\left\Vert j\right\Vert _{p})+\frac{\alpha_{2}}{2}}.
\label{Eq_22}%
\end{equation}
Furthermore,
\begin{equation}
\mathbb{B}_{\mathcal{Y}}=\mathbb{B}_{l}\mid_{\left\{  \widehat{\varphi}%
_{1}\left(  \boldsymbol{j}\right)  =0,\widehat{\varphi}_{2}\left(
\boldsymbol{j}\right)  =0;\boldsymbol{j}\notin J\right\}  }. \label{Eq_23}%
\end{equation}

\end{lemma}

\begin{proof}
Since $\mathcal{L}_{\mathbb{R}}=\cup_{l=1}^{\infty}\mathcal{L}_{\mathbb{R}%
}^{l}$ and $\mathcal{L}_{\mathbb{R}}^{l}\subset\mathcal{L}_{\mathbb{R}}^{m}$
for $m>l$, there is is a positive integer $l=l(\mathcal{Y})$ such that
$\mathcal{Y}\subset\mathcal{L}_{\mathbb{R}}^{l}$. Then there is a subset
$J\subset G_{l}^{+}$ such that $\left\{  \Omega_{\pm}\left(  p^{l}\left\Vert
x-\boldsymbol{i}\right\Vert _{p}\right)  \right\}  _{\boldsymbol{i}\in J}$ is
a basis of $\mathcal{Y}$, and so the formula (\ref{Eq_22}) holds. The
assertion (\ref{Eq_23}) follows from \ (\ref{Eq_18}).
\end{proof}

\begin{corollary}
\label{Cor1}The collection $\left\{  \mathbb{B}_{\mathcal{Y}};\mathcal{Y}%
\text{ finite dimensional subspace of\ }\mathcal{L}_{\mathbb{R}}\right\}  $ is
completely determined by the collection $\left\{  \mathbb{B}_{l}%
;l\in\mathbb{N}\smallsetminus\left\{  0\right\}  \right\}  $. In the sense
that given any $\mathbb{B}_{\mathcal{Y}}$ there is an integer $l$ and a subset
$J\subset G_{l}^{+}$, the case $J=\emptyset$ is included, such that
$\mathbb{B}_{\mathcal{Y}}=\mathbb{B}_{l}\mid_{\left\{  \widehat{\varphi}%
_{1}\left(  \boldsymbol{j}\right)  =0,\widehat{\varphi}_{2}\left(
\boldsymbol{j}\right)  =0;\boldsymbol{j}\notin J\right\}  }$.
\end{corollary}

\subsubsection{ Gaussian measures in the non-Archimedean framework}

We recall that $\mathcal{D}(\mathbb{Q}_{p}^{N})$ is a nuclear space, cf.
\cite[Section 4]{Bruhat}, and thus $\mathcal{L}_{\mathbb{R}}(\mathbb{Q}%
_{p}^{N})$ is a nuclear space, since any subspace of a nuclear space is also
nuclear, see e.g. \cite[Proposition 50.1]{Treves}.

The spaces
\[
\mathcal{L}_{\mathbb{R}}\left(  \mathbb{Q}_{p}^{N}\right)  \hookrightarrow
L_{\mathbb{R}}^{2}\left(  \mathbb{Q}_{p}^{N}\right)  \hookrightarrow
\mathcal{L}_{\mathbb{R}}^{\prime}\left(  \mathbb{Q}_{p}^{N}\right)
\]
form a Gel'fand triple, that is, $\mathcal{L}_{\mathbb{R}}\left(
\mathbb{Q}_{p}^{N}\right)  $ is a nuclear space which is densely and
continuously embedded in $L_{\mathbb{R}}^{2}$ (see \cite[Theorem 7.4.3]%
{A-K-S}) and $\left\Vert g\right\Vert _{2}^{2}=\left\langle g,g\right\rangle $
for $g\in\mathcal{L}_{\mathbb{R}}\left(  \mathbb{Q}_{p}^{N}\right)  $.

We denote by $\mathcal{B}:=\mathcal{B}(\mathcal{L}_{\mathbb{R}}^{\prime
}\left(  \mathbb{Q}_{p}^{N}\right)  )$ the $\sigma$-algebra generated by the
cylinder subsets of $\mathcal{L}_{\mathbb{R}}^{\prime}\left(  \mathbb{Q}%
_{p}^{N}\right)  $. The mapping
\[%
\begin{array}
[c]{cccc}%
\mathcal{C}: & \mathcal{L}_{\mathbb{R}}\left(  \mathbb{Q}_{p}^{N}\right)  &
\rightarrow & \mathbb{C}\\
& f & \rightarrow & e^{-\frac{1}{2}\mathbb{B}(f,f)}%
\end{array}
\]
defines a characteristic functional, i.e. $\mathcal{C}$ is continuous,
positive definite and $\mathcal{C}\left(  0\right)  =1$. The continuity
follows from Lemma \ref{Lemma12}. The fact that $\mathbb{B}$ defines an inner
product in $L^{2}\left(  \mathbb{Q}_{p}^{N}\right)  $ implies that the
functional $\mathcal{C}$ is positive definite.

\begin{definition}
\label{Def_white_noise_space}By the Bochner-Minlos theorem, see e.g.
\cite{Ber-Kon}, \cite{Hida et al}, \cite{Huang-Yang}, there exists a unique
probability measure $\mathbb{P}:=\mathbb{P}\left(  \delta,\gamma,\alpha
_{2}\right)  $ called \textit{the canonical Gaussian measure} on $\left(
\mathcal{L}_{\mathbb{R}}^{\prime}\left(  \mathbb{Q}_{p}^{N}\right)
,\mathcal{B}\right)  $, given by its characteristic functional as%
\begin{equation}%
{\textstyle\int\limits_{\mathcal{L}_{\mathbb{R}}^{\prime}\left(
\mathbb{Q}_{p}^{N}\right)  }}
e^{\sqrt{-1}\langle W,f\rangle}d\mathbb{P}(W)=e^{-\frac{1}{2}\mathbb{B}%
(f,f)}\text{,}\ \ f\in\mathcal{L}_{\mathbb{R}}\left(  \mathbb{Q}_{p}%
^{N}\right)  \text{.}\label{Eq_Char_func_2A}%
\end{equation}
We set $\left(  L_{\mathbb{R}}^{\rho}\right)  :=L^{\rho}\left(  \mathcal{L}%
_{\mathbb{R}}^{\prime}\left(  \mathbb{Q}_{p}^{N}\right)  ,\mathbb{P}\right)
$, $\rho\in\left[  1,\infty\right)  $, to denote the real vector space of
measu\-rable functions $\Psi:\mathcal{L}_{\mathbb{R}}^{\prime}\left(
\mathbb{Q}_{p}^{N}\right)  \rightarrow\mathbb{R}$ satisfying%
\[
\left\Vert \Psi\right\Vert _{\left(  L_{\mathbb{R}}^{\rho}\right)  }^{\rho}=%
{\textstyle\int\limits_{\mathcal{L}_{\mathbb{R}}^{\prime}\left(
\mathbb{Q}_{p}^{N}\right)  }}
\left\vert \Psi\left(  W\right)  \right\vert ^{\rho}d\mathbb{P}(W)<\infty
\text{.}%
\]

\end{definition}

\subsubsection{\label{Section_Further_Remarks}Further remarks on the cylinder
measure $\mathbb{P}$}

We set $\mathbb{L}\left(  \varphi\right)  =\exp\frac{-1}{2}\mathbb{B}%
(\varphi,\varphi)$, for $\varphi\in\mathcal{L}_{\mathbb{R}}$. The functional
$\mathbb{L}$ is positive definite, continuous and $\mathbb{L}(0)=1$. By taking
the restriction of \ $\mathbb{L}$ to a finite dimensional subspace
$\mathcal{Y}$ of $\mathcal{L}_{\mathbb{R}}$, one obtains a positive definite,
continuous functional $\mathbb{L}_{\mathcal{Y}}(\varphi)$ on $\mathcal{Y}$. By
the Bochner theorem, see e.g. \cite[Chapter II, Section 3.2]{Gelfand-Vilenkin}%
, this function is the Fourier transform of a probability measure
$\mathbb{P}_{_{\mathcal{Y}}}$ defined in the dual space $\mathcal{Y}^{\prime
}\subset\mathcal{L}_{\mathbb{R}}^{\prime}$ of $\mathcal{Y}$. By identifying
$\mathcal{Y}^{\prime}$ with $\mathcal{L}_{\mathbb{R}}^{\prime}\left(
\mathbb{Q}_{p}^{N}\right)  /\mathcal{Y}^{0}$, where $\mathcal{Y}^{0}%
$\ consists of all linear functionals $T$ which vanish on $\mathcal{Y}$, we
get that $\mathbb{P}_{_{\mathcal{Y}}}$ is a probability measure in the finite
dimensional space \ $\mathcal{L}_{\mathbb{R}}^{\prime}\left(  \mathbb{Q}%
_{p}^{N}\right)  /\mathcal{Y}^{0}$. The measure $\mathbb{P}$ is constructed
from the family of probability measures $\left\{  \mathbb{P}_{_{\mathcal{Y}}%
};\mathcal{Y\subset L}_{\mathbb{R}}\text{, finite dimensional space}\right\}
$. These measures are compatible and satisfy a suitable continuity condition,
and they give rise to a cylinder measure $\mathbb{P}$ in $\mathcal{L}%
_{\mathbb{R}}^{\prime}$. Since $\mathcal{L}_{\mathbb{R}}$ is a nuclear space,
this cylinder measure is countably additive.\ For further details about the
construction of the measure $\mathbb{P}$, the reader may consult \cite[Chapter
IV, Section 4.2, proof of Theorem 1]{Gelfand-Vilenkin}.

Now, by using the formula%
\[
\mathbb{L}\left(  \varphi\right)  =%
{\textstyle\int\limits_{\mathcal{L}_{\mathbb{R}}^{\prime}\left(
\mathbb{Q}_{p}^{N}\right)  /\mathcal{Y}^{0}}}
e^{\sqrt{-1}\left\langle W,\varphi\right\rangle }d\mathbb{P}_{_{\mathcal{Y}}%
}\left(  \varphi\right)  \text{ for }\varphi\in\mathcal{Y}\text{,}%
\]
see \cite[Chapter IV, Section 4.1]{Gelfand-Vilenkin}, and the fact that
$\mathbb{L}\left(  \varphi\right)  =\exp\frac{-1}{2}\mathbb{B}(\varphi
,\varphi)$, for $\varphi\in\mathcal{Y}$, one gets that $\mathbb{P}%
_{_{\mathcal{Y}}}$ is a Gaussian probability measure in $\mathcal{Y}$, with
mean zero, and correlation function $\mathbb{B}$, i.e. if $\mathcal{Y}$\ has
dimension $n$, then%
\[
\mathbb{P}_{_{\mathcal{Y}}}\left(  A\right)  =\frac{1}{\left(  2\pi\right)
^{\frac{n}{2}}}%
{\textstyle\int\limits_{A}}
e^{-\frac{1}{2}\mathbb{B}(\psi,\psi)}d\mathbb{\psi}\text{,}%
\]
where $d\mathbb{\psi}$ is the Lebesgue measure in $\mathcal{Y}$ corresponding
to the scalar product $\mathbb{B}$, and $A\subset\mathcal{Y}$ is a measurable
subset. In conclusion, the cylinder measure $\mathbb{P}$ is uniquely
determined by the family of Gaussian measures
\[
\left\{  \mathbb{P}_{_{\mathcal{Y}}};\mathcal{Y\subset L}_{\mathbb{R}}\text{,
finite dimensional space}\right\}  ,
\]
or equivalently by the sequence%
\begin{equation}
\left\{  \mathbb{B}_{_{\mathcal{Y}}};\mathcal{Y\subset L}_{\mathbb{R}}\text{,
finite dimensional space}\right\}  , \label{Eq_sequence_B}%
\end{equation}
where $\mathbb{B}_{_{\mathcal{Y}}}$ denotes the restriction of the scalar
product to $\mathbb{B}$ to $\mathcal{Y}$. This is a consequence of the fact
that any finite dimensional Gaussian measure, with mean zero, is completely
determined by its correlation matrix.

\subsection{Existence of a measure on $\mathcal{FL}_{\mathbb{R}}^{l}\left(
\mathbb{Q}_{p}^{N}\right)  $}

\begin{theorem}
\label{Theorem1}Assume that $\delta>N$, $\gamma>0$, $\alpha_{2}>0$. (i) The
cylinder probability measure $\mathbb{P}=\mathbb{P}\left(  \delta
,\gamma,\alpha_{2}\right)  $ is uniquely determined\ by the sequence
$\mathbb{P}_{l}=\mathbb{P}_{l}\left(  \delta,\gamma,\alpha_{2}\right)  $,
$l\in\mathbb{N}\smallsetminus\left\{  0\right\}  $, of Gaussian measures. (ii)
Let $f:\mathcal{F}\mathcal{L}_{\mathbb{R}}\left(  \mathbb{Q}_{p}^{N}\right)
\rightarrow\mathbb{R}$ be a continuous and bounded function. Then%
\[
\lim_{l\rightarrow\infty}%
{\textstyle\int\limits_{\mathcal{FL}_{\mathbb{R}}^{l}\left(  \mathbb{Q}%
_{p}^{N}\right)  }}
f\left(  \widehat{\varphi}\right)  d\mathbb{P}_{l}\left(  \widehat{\varphi
}\right)  =%
{\textstyle\int\limits_{\mathcal{FL}_{\mathbb{R}}\left(  \mathbb{Q}_{p}%
^{N}\right)  }}
f\left(  \widehat{\varphi}\right)  d\mathbb{P}\left(  \widehat{\varphi
}\right)  .
\]

\end{theorem}

\begin{proof}
(i) We use the notation and results given in Section
\ref{Section_Further_Remarks}. By the Corollary \ref{Cor1}, the sequence
(\ref{Eq_sequence_B}) is completely determined by the sequence $\left\{
2p^{-lN}\mathbb{B}_{l};l\in\mathbb{N}\smallsetminus\left\{  0\right\}
\right\}  $, i.e. by the sequence $\left\{  \mathbb{P}_{l};l\in\mathbb{N}%
\smallsetminus\left\{  0\right\}  \right\}  $. Notice that the covariance
matrix of $\mathbb{P}_{l}$ is $2p^{-lN}B^{-1}(l)=2p^{-lN}\mathbb{B}_{l}$, cf.
Lemma \ref{Lemma13}. Then the cylinder measure $\mathbb{P}$ is exactly the
probability measure announced in Lemma \ref{Lemma11}.

(ii) By\ using the formula given in Lemma \ref{Lemma11}, for any bounded
continuous function $f$ supported in $\mathcal{FL}_{\mathbb{R}}^{l}$, we have%
\begin{equation}%
{\textstyle\int\limits_{\mathcal{FL}_{\mathbb{R}}^{l}}}
f\left(  \widehat{\varphi}\right)  d\mathbb{P}_{l}\left(  \widehat{\varphi
}\right)  =%
{\textstyle\int\limits_{\mathcal{FL}_{\mathbb{R}}^{l}}}
f\left(  \widehat{\varphi}\right)  d\mathbb{P}\left(  \widehat{\varphi
}\right)  . \label{Eq_25}%
\end{equation}
By the uniqueness of the probability space $\left(  X,\mathcal{F}%
;\mathbb{P}\right)  $ in Lemma \ref{Lemma11}, we can identify the $\sigma
$-algebra $\mathcal{F}$ with $\mathcal{B}(\mathcal{L}_{\mathbb{R}}^{\prime
}\left(  \mathbb{Q}_{p}^{N}\right)  )$, the $\sigma$-algebra generated by the
cylinder subsets of $\mathcal{L}_{\mathbb{R}}^{\prime}\left(  \mathbb{Q}%
_{p}^{N}\right)  $. Then $\mathcal{FL}_{\mathbb{R}}^{l}$ belongs to
$\mathcal{B}(\mathcal{L}_{\mathbb{R}}^{\prime}\left(  \mathbb{Q}_{p}%
^{N}\right)  )$, and $\mathcal{FL}_{\mathbb{R}}=\cup_{l}\mathcal{FL}%
_{\mathbb{R}}^{l}$ also belongs to $\mathcal{B}(\mathcal{L}_{\mathbb{R}%
}^{\prime}\left(  \mathbb{Q}_{p}^{N}\right)  )$. Now by taking the limit
$l\rightarrow\infty$ in (\ref{Eq_25}), we get the announced formula.
\end{proof}

\subsection{Further comments on Theorem \ref{Theorem1}}

By using the Gel'fand triple,
\[
\mathcal{D}_{\mathbb{R}}\left(  \mathbb{Q}_{p}^{N}\right)  \hookrightarrow
L_{\mathbb{R}}^{2}\left(  \mathbb{Q}_{p}^{N}\right)  \hookrightarrow
\mathcal{D}_{\mathbb{R}}^{\prime}\left(  \mathbb{Q}_{p}^{N}\right)  ,
\]
and the fact that $\mathcal{D}(\mathbb{Q}_{p}^{N})$ is a nuclear space, cf.
\cite[Section 4]{Bruhat}, it follows from Lemma \ref{Lemma12} that%
\[%
\begin{array}
[c]{cccc}%
\mathcal{C}: & \mathcal{D}_{\mathbb{R}}\left(  \mathbb{Q}_{p}^{N}\right)   &
\rightarrow & \mathbb{C}\\
& f & \rightarrow & e^{-\frac{1}{2}\mathbb{B}(f,f)}%
\end{array}
\]
defines a characteristic functional, then by the Bochner-Minlos theorem, there
exists a unique probability measure $\mathbb{S}:=\mathbb{S}\left(
\delta,\gamma,\alpha_{2}\right)  $ on $\left(  \mathcal{D}_{\mathbb{R}%
}^{\prime}\left(  \mathbb{Q}_{p}^{N}\right)  ,\mathcal{B}_{0}\right)  $ given
by%
\[%
{\textstyle\int\limits_{\mathcal{D}_{\mathbb{R}}^{\prime}\left(
\mathbb{Q}_{p}^{N}\right)  }}
e^{\sqrt{-1}\langle W,f\rangle}d\mathbb{S}(W)=e^{-\frac{1}{2}\mathbb{B}%
(f,f)}\text{,}\ \ f\in\mathcal{D}_{\mathbb{R}}\left(  \mathbb{Q}_{p}%
^{N}\right)  \text{,}%
\]
where $\mathcal{B}_{0}:=\mathcal{B}_{0}(\mathcal{D}_{\mathbb{R}}^{\prime
}\left(  \mathbb{Q}_{p}^{N}\right)  )$ the $\sigma$-algebra generated by the
cylinder subsets of $\mathcal{D}_{\mathbb{R}}^{\prime}\left(  \mathbb{Q}%
_{p}^{N}\right)  $. Therefore
\[
\mathbb{P=}\frac{1_{\mathcal{L}_{\mathbb{R}}^{\prime}\left(  \mathbb{Q}%
_{p}^{N}\right)  }\mathbb{S}}{\int_{\mathcal{L}_{\mathbb{R}}^{\prime}\left(
\mathbb{Q}_{p}^{N}\right)  }d\mathbb{S}}.
\]

\section{Partition functions and generating functionals}

In this section we introduce a family of $\mathcal{P}(\varphi)$-theories,
where%
\begin{equation}
\mathcal{P}(X)=a_{3}X^{3}+a_{4}X^{4}+\ldots+a_{2k}X^{2D}\in\mathbb{R}\left[
X\right]  \text{, with }D\geq2\text{, } \label{Poly_interactions}%
\end{equation}
satisfying $\mathcal{P}(\alpha)\geq0$ for any $\alpha\in\mathbb{R}$. Notice
that this implies that for $\varphi\in\mathcal{D}_{\mathbb{R}}^{l}$ and
$\alpha_{4}>0$, $\exp\left(  -\frac{\alpha_{4}}{2}\int\mathcal{P}%
(\varphi)d^{N}x\right)  \leq1$. This fact follows from Remark
\ref{Nota_discretization}. Each of these theories corresponds to a thermally
fluctuating field which is defined by means of a functional integral
representation of the partition function. All the thermodynamic quantities and
correlation functions of the system can be obtained by functional
differentiation from a generating functional as in the classical case, see
e.g. \cite{Kleinert et al}, \cite{Mussardo}. In this section, we provide
mathematical rigorous definitions of all these objects.

\subsection{Partition functions}

We assume that $\varphi\in\mathcal{L}_{\mathbb{R}}\left(  \mathbb{Q}_{p}%
^{N}\right)  $ represents a field that performs thermal fluctuations. We also
assume that in the normal phase the expectation value of the field $\varphi$
is zero. Then the fluctuations take place around zero. The size of these
fluctuations is controlled by the energy functional:%
\[
E(\varphi):=E_{0}(\varphi)+E_{\text{int}}(\varphi),
\]
where the first terms is defined in (\ref{Energy_Functioal_E_0}), and the
second term is%
\[
E_{\text{int}}(\varphi):=\frac{\alpha_{4}}{4}%
{\displaystyle\int\limits_{\mathbb{Q}_{p}^{N}}}
\mathcal{P}\left(  \varphi\left(  x\right)  \right)  d^{N}x\text{, \ }%
\alpha_{4}\geq0\text{,}%
\]
corresponds to the interaction energy.

All the thermodynamic properties of the system attached to the field $\varphi$
are described by the partition function of the fluctuating field, which is
given classically by a functional integral%
\[
\mathcal{Z}^{\text{phys}}=%
{\displaystyle\int}
D\left(  \varphi\right)  e^{-\frac{E(\varphi)}{K_{B}T}},
\]
where $D\left(  \varphi\right)  $ is a `spurious measure' on the space of
fields, $K_{B}$ is the Boltzmann's constant and $T$ is the temperature. We use
the normalization $K_{B}T=1$. When the coupling constant $\alpha_{4}=0$,
$\mathcal{Z}^{\text{phys}}$ reduced to the free-field partition function%
\[
\mathcal{Z}_{0}^{\text{phys}}=%
{\displaystyle\int}
D\left(  \varphi\right)  e^{-E_{0}(\varphi)}.
\]
It is more convenient to use a normalize partition function $\frac
{\mathcal{Z}^{\text{phys}}}{\mathcal{Z}_{0}^{\text{phys}}}$.

\begin{definition}
Assume that $\delta>N$, and $\gamma$, $\alpha_{2}>0$. The free-partition
function is defined as%
\[
\mathcal{Z}_{0}=\mathcal{Z}_{0}(\delta,\gamma,\alpha_{2})=%
{\displaystyle\int\limits_{\mathcal{L}_{\mathbb{R}}\left(  \mathbb{Q}_{p}%
^{N}\right)  }}
d\mathbb{P}\left(  \varphi\right)  .
\]
The discrete free-partition function is defined as%
\[
\mathcal{Z}_{0}^{\left(  l\right)  }=\mathcal{Z}_{0}^{\left(  l\right)
}(\delta,\gamma,\alpha_{2})=%
{\displaystyle\int\limits_{\mathcal{L}_{\mathbb{R}}^{l}\left(  \mathbb{Q}%
_{p}^{N}\right)  }}
d\mathbb{P}_{l}\left(  \varphi\right)
\]
for $l\in\mathbb{N}\smallsetminus\left\{  0\right\}  $.
\end{definition}

By Lemma \ref{Lemma11}, $\lim_{l\rightarrow\infty}\mathcal{Z}_{0}^{\left(
l\right)  }=\mathcal{Z}_{0}$. Notice that the term $e^{-E_{0}(\varphi)}$ is
used to construct the measure $\mathbb{P}\left(  \varphi\right)  $.

\begin{definition}
Assume that $\delta>N$, and $\gamma$, $\alpha_{2}$, $\alpha_{4}>0$. The
partition function is defined as
\[
\mathcal{Z}=\mathcal{Z}(\delta,\gamma,\alpha_{2},\alpha_{4})=%
{\displaystyle\int\limits_{\mathcal{L}_{\mathbb{R}}\left(  \mathbb{Q}_{p}%
^{N}\right)  }}
e^{-E_{\text{int}}\left(  \varphi\right)  }d\mathbb{P}\left(  \varphi\right)
.
\]
The discrete partition functions are defined as
\[
\mathcal{Z}^{\left(  l\right)  }=\mathcal{Z}^{\left(  l\right)  }%
(\delta,\gamma,\alpha_{2},\alpha_{4})=%
{\displaystyle\int\limits_{\mathcal{L}_{\mathbb{R}}^{l}\left(  \mathbb{Q}%
_{p}^{N}\right)  }}
e^{-E_{\text{int}}\left(  \varphi\right)  }d\mathbb{P}_{l}\left(
\varphi\right)  ,
\]
for $l\in\mathbb{N}\smallsetminus\left\{  0\right\}  $.
\end{definition}

Notice that $e^{-E_{\text{int}}\left(  \varphi\right)  }$ is bounded and
(sequentially) continuous in $\mathcal{L}_{\mathbb{R}}$, and consequently in
$\mathcal{L}_{\mathbb{R}}^{l}$ for any $l$. Indeed, take $\varphi_{n}$
$\underrightarrow{\mathcal{D}_{\mathbb{R}}}\ 0$, $\mathcal{L}_{\mathbb{R}}$ is
endowed with the topology of $\mathcal{D}_{\mathbb{R}}$. Then there is $l$
such that $\varphi_{n}\in\mathcal{L}_{\mathbb{R}}^{l}$ for every $n$, and
$\varphi_{n}$ $\underrightarrow{\text{unif.}}\ 0$, i.e.
\[
\varphi_{n}(x)=%
{\textstyle\sum\limits_{\boldsymbol{i}\in G_{l}}}
\varphi^{\left(  n\right)  }\left(  \boldsymbol{i}\right)  \Omega\left(
p^{l}\left\Vert x-\boldsymbol{i}\right\Vert _{p}\right)  \text{, and }%
\max_{\boldsymbol{i}\in G_{l}}\left\{  \varphi^{\left(  n\right)  }\left(
\boldsymbol{i}\right)  \right\}  \rightarrow0\text{ as }n\rightarrow\infty.
\]
Which implies that $E_{\text{int}}\left(  \varphi_{n}\right)  \rightarrow0$.
Again by Lemma \ref{Lemma11}, $\lim_{l\rightarrow\infty}\mathcal{Z}^{\left(
l\right)  }=\mathcal{Z}$.

\subsection{Correlation functions}

From a mathematical perspective a $\mathcal{P}\left(  \varphi\right)  $-theory
is given by a cylinder probability measure of the form%
\begin{equation}
\frac{1_{\mathcal{L}_{\mathbb{R}}}\left(  \varphi\right)  e^{-E_{\text{int}%
}\left(  \varphi\right)  }d\mathbb{P}}{%
{\displaystyle\int\nolimits_{\mathcal{L}_{\mathbb{R}}\left(  \mathbb{Q}%
_{p}^{N}\right)  }}
e^{-E_{\text{int}}\left(  \varphi\right)  }d\mathbb{P}}=\frac{1_{\mathcal{L}%
_{\mathbb{R}}}\left(  \varphi\right)  e^{-E_{\text{int}}\left(  \varphi
\right)  }d\mathbb{P}}{\mathcal{Z}} \label{Eq_Measure}%
\end{equation}
in the space of fields $\mathcal{L}_{\mathbb{R}}\left(  \mathbb{Q}_{p}%
^{N}\right)  $. It is important to mention that we do not require the Wick
regularization operation in $e^{-E_{\text{int}}\left(  \varphi\right)  }$
because we are restricting the fields to be test functions.

\begin{definition}
\label{Definition_G_m}The $m$-point correlation functions of a field
$\varphi\in\mathcal{L}_{\mathbb{R}}\left(  \mathbb{Q}_{p}^{N}\right)  $ are
\ defined as%
\[
G^{\left(  m\right)  }\left(  x_{1},\ldots,x_{m}\right)  =\frac{1}%
{\mathcal{Z}}%
{\displaystyle\int\limits_{\mathcal{L}_{\mathbb{R}}\left(  \mathbb{Q}_{p}%
^{N}\right)  }}
\left(
{\displaystyle\prod\limits_{i=1}^{m}}
\varphi\left(  x_{i}\right)  \right)  e^{-E_{\text{int}}\left(  \varphi
\right)  }d\mathbb{P}.
\]
The discrete $m$-point correlation functions of a field $\varphi\in
\mathcal{L}_{\mathbb{R}}^{l}\left(  \mathbb{Q}_{p}^{N}\right)  $ are defined
as%
\[
G_{l}^{\left(  m\right)  }\left(  x_{1},\ldots,x_{m}\right)  =\frac
{1}{\mathcal{Z}^{\left(  l\right)  }}%
{\displaystyle\int\limits_{\mathcal{L}_{\mathbb{R}}^{l}\left(  \mathbb{Q}%
_{p}^{N}\right)  }}
\left(
{\displaystyle\prod\limits_{i=1}^{m}}
\varphi\left(  x_{i}\right)  \right)  e^{-E_{\text{int}}\left(  \varphi
\right)  }d\mathbb{P}_{l},
\]
for $l\in\mathbb{N}\smallsetminus\left\{  0\right\}  $.
\end{definition}

\begin{lemma}
\label{Lemma15}The discrete $m$-point correlation functions $G_{l}^{\left(
m\right)  }\left(  x_{1},\ldots,x_{m}\right)  $ of a field $\varphi
\in\mathcal{L}_{\mathbb{R}}\left(  \mathbb{Q}_{p}^{N}\right)  $ are test
functions in $x_{1},\ldots,x_{m}$.
\end{lemma}

\begin{proof}
There is an positive integer $l=l(\varphi)$ such that $\varphi\in
\mathcal{L}_{\mathbb{R}}^{l}$ and $x_{1},\ldots,x_{m}\in B_{l}^{N}$. By using
that%
\begin{equation}
\varphi\left(  x_{i}\right)  =%
{\textstyle\sum\limits_{\boldsymbol{j}\in G_{l}}}
\varphi\left(  \boldsymbol{j}\right)  \Omega\left(  p^{l}\left\Vert
x_{i}-\boldsymbol{j}\right\Vert _{p}\right)  , \label{Eq_phi_expansion}%
\end{equation}
one gets that $%
{\textstyle\prod\nolimits_{i=1}^{m}}
\varphi\left(  x_{i}\right)  $ is a finite sum of terms of the form%
\[%
{\displaystyle\prod\limits_{k=1}^{m}}
\varphi\left(  \boldsymbol{j}_{k}\right)  \Omega\left(  p^{l}\left\Vert
x_{k}-\boldsymbol{j}_{k}\right\Vert _{p}\right)  =:F(\varphi\left(
\boldsymbol{j}_{1}\right)  ,\ldots,\varphi\left(  \boldsymbol{j}_{m}\right)
)\Theta_{l}\left(  x_{1},\ldots,x_{m}\right)  ,
\]
where $F(\varphi\left(  \boldsymbol{j}_{1}\right)  ,\ldots,\varphi\left(
\boldsymbol{j}_{m}\right)  )$ is a polynomial function defined in
$\mathcal{L}_{\mathbb{R}}^{l}$, $\boldsymbol{j}_{k}\in G_{l}$, and $\Theta
_{l}\left(  x\right)  =\Theta_{l}\left(  x_{1},\ldots,x_{m}\right)  $ is the
characteristic function of the polydisc $B_{-l}^{N}(\boldsymbol{j}_{1}%
)\times\cdots\times B_{-l}^{N}(\boldsymbol{j}_{m})$. Now, by using that
$\exp\left(  -E_{\text{int}}\left(  \varphi\right)  \right)  =\exp
(-\frac{\alpha_{4}}{4}p^{-lN}\sum_{k=3}^{2D}\sum_{\boldsymbol{j}\in G_{l}%
}a_{k}\varphi^{k}\left(  \boldsymbol{j}\right)  )$, the correlation function
$G_{l}^{\left(  m\right)  }\left(  x_{1},\ldots,x_{m}\right)  $ is a finite
sum of test functions of the form%
\begin{align*}
&  \Theta_{l}\left(  x\right)
{\textstyle\int\limits_{\mathcal{L}_{\mathbb{R}}^{l}}}
\left\{  F(\varphi\left(  \boldsymbol{j}_{1}\right)  ,\ldots,\varphi\left(
\boldsymbol{j}_{m}\right)  )\exp(-\frac{\alpha_{4}}{4}p^{-lN}\sum_{k=3}%
^{2D}\sum_{\boldsymbol{j}\in G_{l}}a_{k}\varphi^{k}\left(  \boldsymbol{j}%
\right)  )\right\}  d\mathbb{P}_{l}\mathbb{=}\\
&  \Theta_{l}\left(  x\right)
{\textstyle\int\limits_{\mathcal{L}_{\mathbb{R}}^{l}}}
\left\{  F(\varphi\left(  \boldsymbol{j}_{1}\right)  ,\ldots,\varphi\left(
\boldsymbol{j}_{m}\right)  )\exp(-\frac{\alpha_{4}}{4}p^{-lN}\sum_{k=3}%
^{2D}\sum_{\boldsymbol{j}\in G_{l}}a_{k}\varphi^{k}\left(  \boldsymbol{j}%
\right)  )\right\}  d\mathbb{P},
\end{align*}
where the convergence of the integrals is guaranteed by the fact that the
integrands are bounded functions, cf. Lemma \ref{Lemma11}.
\end{proof}

Notice that the pointwise limit $G^{\left(  m\right)  }\left(  x_{1}%
,\ldots,x_{m}\right)  =\lim_{l\rightarrow\infty}G_{l}^{\left(  m\right)
}\left(  x_{1},\ldots,x_{m}\right)  $ is not a test function due to the fact
that $\Theta_{l}\left(  x\right)  $ has an arbitrary small exponent of local
constancy when $l$ tends to infinity.

\subsection{Generating functionals}

We now introduce a current $J(x)\in\mathcal{L}_{\mathbb{R}}\left(
\mathbb{Q}_{p}^{N}\right)  $ and add to the energy functional $E(\varphi)$ a
linear interaction \ energy of this current with the field $\varphi\left(
x\right)  $,%
\[
E_{\text{source}}(\varphi,J):=-%
{\textstyle\int\limits_{\mathbb{Q}_{p}^{N}}}
\varphi\left(  x\right)  J(x)d^{N}x\text{,}%
\]
in this way we get a new energy functional%
\[
E(\varphi,J):=E(\varphi)+E_{\text{source}}(\varphi,J).
\]
Notice that $E_{\text{source}}(\varphi,J)=-\left\langle \varphi,J\right\rangle
$, where $\left\langle \cdot,\cdot\right\rangle $ denotes the scalar product
of $L^{2}(\mathbb{Q}_{p}^{N})$. This scalar product extends to the pairing
between \ $\mathcal{L}_{\mathbb{R}}^{\prime}\left(  \mathbb{Q}_{p}^{N}\right)
$ and $\mathcal{L}_{\mathbb{R}}\left(  \mathbb{Q}_{p}^{N}\right)  $

\begin{definition}
\label{Definition_Z_J}Assume that $\delta>N$, and $\gamma$, $\alpha_{2}$,
$\alpha_{4}>0$. The partition function corresponding to the energy functional
$E(\varphi,J)$ is defined as
\[
\mathcal{Z}(J;\delta,\gamma,\alpha_{2},\alpha_{4})=\frac{1}{\mathcal{Z}_{0}}%
{\displaystyle\int\limits_{\mathcal{L}_{\mathbb{R}}\left(  \mathbb{Q}_{p}%
^{N}\right)  }}
e^{-E_{\text{int}}\left(  \varphi\right)  +\left\langle \varphi,J\right\rangle
}\text{ }d\mathbb{P}\text{,}%
\]
and the discrete versions
\[
\mathcal{Z}^{(l)}(J;\delta,\gamma,\alpha_{2},\alpha_{4})=\frac{1}%
{\mathcal{Z}_{0}^{\left(  l\right)  }}%
{\displaystyle\int\limits_{\mathcal{L}_{\mathbb{R}}^{l}\left(  \mathbb{Q}%
_{p}^{N}\right)  }}
e^{-E_{\text{int}}\left(  \varphi\right)  +\left\langle \varphi,J\right\rangle
}\text{ }d\mathbb{P}_{l}\text{,}%
\]
for $l\in\mathbb{N}\smallsetminus\left\{  0\right\}  $.
\end{definition}

For the sake of simplicity we will the notation $\mathcal{Z}(J)=\mathcal{Z}%
(J;\delta,\gamma,\alpha_{2},\alpha_{4})$, $\mathcal{Z}^{\left(  l\right)
}(J)=\mathcal{Z}^{(l)}(J;\delta,\gamma,\alpha_{2},\alpha_{4})$.

\begin{remark}
\label{Nota_3}In this section, we need some functionals from the space
\[
\left(  L_{\mathbb{R}}^{\rho}\right)  =L^{\rho}\left(  \mathcal{L}%
_{\mathbb{R}}^{\prime}\left(  \mathbb{Q}_{p}^{N}\right)  ,d\mathbb{P}\right)
\text{, \ }\rho\in\left[  1,\infty\right)  ,
\]
see Definition \ref{Def_white_noise_space}. Let $F\left(  X_{1},\ldots
,X_{n}\right)  $ be a real-valued polynomial, and $\xi=\left(  \xi_{1}%
,\ldots,\xi_{n}\right)  $ , with $\xi_{i}\in\mathcal{L}_{\mathbb{R}}\left(
\mathbb{Q}_{p}^{N}\right)  $ for $i=1,\ldots,n$, then the functional%
\[
F_{\xi}(W):=F\left(  \left\langle W,\xi_{1}\right\rangle ,\ldots,\left\langle
W,\xi_{n}\right\rangle \right)  \text{, }W\in\mathcal{L}_{\mathbb{R}}^{\prime
}\left(  \mathbb{Q}_{p}^{N}\right)  ,
\]
belongs to $\left(  L_{\mathbb{R}}^{\rho}\right)  $, $\rho\in\left[
1,\infty\right)  $, see e.g. \cite[Proposition 1.6]{Hida et al}. The
functional $\exp C\left\langle \cdot,\phi\right\rangle $, for $C\in\mathbb{R}%
$, $\phi\in\mathcal{L}_{\mathbb{R}}$ belongs to $\left(  L_{\mathbb{R}}^{\rho
}\right)  $, $\rho\in\left[  1,\infty\right)  $, see e.g. \cite[Proposition
1.7]{Hida et al}. The $\mathbb{R}$-algebra $\mathcal{A}$ generated by the
functionals $F_{\xi}$, $\exp C\left\langle \cdot,\phi\right\rangle $ is dense
in $\left(  L_{\mathbb{R}}^{\rho}\right)  $, $\rho\in\left[  1,\infty\right)
$, see e.g. \cite[Theorem 1.9]{Hida et al}.
\end{remark}

\begin{lemma}
\label{Lemma16}Given $\varphi\in\mathcal{L}_{\mathbb{R}}\left(  \mathbb{Q}%
_{p}^{N}\right)  $, $m\geq1$, and $e_{i}\geq0$ for $i=1,\ldots,m$,\ we define%
\[
\mathcal{I}(\varphi)=%
{\displaystyle\int\limits_{\left(  \mathbb{Q}_{p}^{N}\right)  ^{m}}}
\left(
{\displaystyle\prod\limits_{i=1}^{m}}
\varphi^{e_{i}}\left(  x_{i}\right)  \right)
{\displaystyle\prod\limits_{i=1}^{m}}
d^{N}x_{i}.
\]
Then $\mathcal{I}\in\mathcal{A}$.
\end{lemma}

\begin{proof}
There is an integer $l$ such that $\varphi\in\mathcal{L}_{\mathbb{R}}^{l}$. By
using (\ref{Eq_phi_expansion}), and the fact that the functions $\Omega\left(
p^{l}\left\Vert x_{i}-\boldsymbol{j}\right\Vert _{p}\right)  $,
$\boldsymbol{j}\in G_{l}$, are orthogonal with respect to the scalar product
$\left\langle \cdot,\cdot\right\rangle $ in $L_{\mathbb{R}}^{2}(\mathbb{Q}%
_{p}^{N})$, we have%
\begin{align*}
\varphi\left(  x_{i}\right)   &  =%
{\textstyle\sum\limits_{\boldsymbol{j}\in G_{l}}}
p^{lN}\left\langle \varphi\left(  x_{i}\right)  ,\Omega\left(  p^{l}\left\Vert
x_{i}-\boldsymbol{j}\right\Vert _{p}\right)  \right\rangle \Omega\left(
p^{l}\left\Vert x_{i}-\boldsymbol{j}\right\Vert _{p}\right) \\
&  =%
{\textstyle\sum\limits_{\boldsymbol{j}\in G_{l}}}
p^{lN}\left\langle W_{\boldsymbol{j}},\varphi\right\rangle \Omega\left(
p^{l}\left\Vert x_{i}-\boldsymbol{j}\right\Vert _{p}\right)  ,
\end{align*}
where $W_{\boldsymbol{j}}\in\mathcal{L}_{\mathbb{R}}^{\prime}\left(
\mathbb{Q}_{p}^{N}\right)  $, for $\boldsymbol{j}\in G_{l}$. Consequently,%
\[
\varphi^{e_{i}}\left(  x_{i}\right)  =%
{\textstyle\sum\limits_{\boldsymbol{j}\in G_{l}}}
p^{lNe_{i}}\left\langle W_{\boldsymbol{j}},\varphi\right\rangle ^{e_{i}}%
\Omega\left(  p^{l}\left\Vert x_{i}-\boldsymbol{j}\right\Vert _{p}\right)
\]
and $%
{\textstyle\prod\nolimits_{i=1}^{m}}
\varphi^{e_{i}}\left(  x_{i}\right)  $ is a finite sum of terms of the form%
\[
\left(
{\displaystyle\prod\limits_{k=1}^{m}}
p^{lNe_{i_{k}}}\left\langle W_{\boldsymbol{j}_{k}},\varphi\right\rangle
^{e_{i_{k}}}\right)
{\displaystyle\prod\limits_{k=1}^{m}}
\Omega\left(  p^{l}\left\Vert x_{k}-\boldsymbol{j}_{k}\right\Vert _{p}\right)
,
\]
where $i_{k}\in\left\{  1,\ldots,m\right\}  $,\ $\boldsymbol{j}_{k}\in G_{l}$.
Now $\mathcal{I}(\varphi)$ is a finite sum of terms of the form%
\begin{align*}
&  \left(
{\displaystyle\prod\limits_{k=1}^{m}}
p^{lNe_{i_{k}}}\left\langle W_{\boldsymbol{j}_{k}},\varphi\right\rangle
^{e_{i_{k}}}\right)
{\displaystyle\int\limits_{\left(  \mathbb{Q}_{p}^{N}\right)  ^{m}}}
{\displaystyle\prod\limits_{k=1}^{m}}
\Omega\left(  p^{l}\left\Vert x_{k}-\boldsymbol{j}_{k}\right\Vert _{p}\right)
%
{\displaystyle\prod\limits_{i=1}^{m}}
d^{N}x_{i}\\
&  =p^{-lNm}\left(
{\displaystyle\prod\limits_{k=1}^{m}}
p^{lNe_{i_{k}}}\left\langle W_{\boldsymbol{j}_{k}},\varphi\right\rangle
^{e_{i_{k}}}\right)  \in\mathcal{A},
\end{align*}
and therefore $\mathcal{I}\in\mathcal{A}$.
\end{proof}

\begin{lemma}
\label{Lemma17}With the above notation, the following assertions hold true:

\noindent(i) $1_{\mathcal{L}_{\mathbb{R}}}\left(  \varphi\right)
e^{-E_{\text{int}}\left(  \varphi\right)  +\left\langle \varphi,J\right\rangle
}\in\left(  L_{\mathbb{R}}^{1}\right)  $. In particular, $\mathcal{Z}%
(J)<\infty$;

\noindent(ii)
\[
\lim_{l\rightarrow\infty}%
{\displaystyle\int\limits_{\mathcal{L}_{\mathbb{R}}^{l}\left(  \mathbb{Q}%
_{p}^{N}\right)  }}
e^{\left\langle \varphi,J\right\rangle }\text{ }d\mathbb{P}_{l}=%
{\displaystyle\int\limits_{\mathcal{L}_{\mathbb{R}}\left(  \mathbb{Q}_{p}%
^{N}\right)  }}
e^{\left\langle \varphi,J\right\rangle }\text{ }d\mathbb{P};
\]

\noindent(iii) $\mathcal{Z}^{\left(  l\right)  }(J)<\infty$ for any
$l\in\mathbb{N}\smallsetminus\left\{  0\right\}  $;

\noindent(iv) $\lim_{l\rightarrow\infty}\mathcal{Z}^{\left(  l\right)
}(J)=\mathcal{Z}(J)$.
\end{lemma}

\begin{proof}
(i) The result follows from%
\[%
{\textstyle\int\limits_{\mathcal{L}_{\mathbb{R}}\left(  \mathbb{Q}_{p}%
^{N}\right)  }}
e^{-E_{\text{int}}\left(  \varphi\right)  +\left\langle \varphi,J\right\rangle
}d\mathbb{P}\left(  \varphi\right)  \mathbb{\leq}%
{\textstyle\int\limits_{\mathcal{L}_{\mathbb{R}}\left(  \mathbb{Q}_{p}%
^{N}\right)  }}
e^{\left\langle \varphi,J\right\rangle }d\mathbb{P}\left(  \varphi\right)
\mathbb{\leq}%
{\textstyle\int\limits_{\mathcal{L}_{\mathbb{R}}^{\prime}\left(
\mathbb{Q}_{p}^{N}\right)  }}
e^{\left\langle W,J\right\rangle }d\mathbb{P(}W\mathbb{)<\infty}\text{,}%
\]
by using Remark \ref{Nota_3}.

(ii) For each $l\in\mathbb{N}\smallsetminus\left\{  0\right\}  $, we take
$\left\{  K_{n_{l}}\right\}  $ to be a increasing sequence of compact subsets
of $\mathcal{L}_{\mathbb{R}}^{l}\left(  \mathbb{Q}_{p}^{N}\right)  $ having
$\mathcal{L}_{\mathbb{R}}^{l}\left(  \mathbb{Q}_{p}^{N}\right)  $ as its
limit. Set%
\[
\mathcal{I}^{\left(  l,n\right)  }(J):=%
{\displaystyle\int\limits_{\mathcal{L}_{\mathbb{R}}^{l}\left(  \mathbb{Q}%
_{p}^{N}\right)  }}
1_{K_{n_{l}}}\left(  \varphi\right)  e^{\left\langle \varphi,J\right\rangle
}\text{ }d\mathbb{P}_{l}.
\]
Since the integrand $1_{K_{n_{l}}}\left(  \varphi\right)  e^{\left\langle
\varphi,J\right\rangle }$ is continuous and bounded, by Lemma \ref{Lemma11},\
\[
\mathcal{I}^{\left(  l,n\right)  }(J)=%
{\displaystyle\int\limits_{\mathcal{L}_{\mathbb{R}}^{l}\left(  \mathbb{Q}%
_{p}^{N}\right)  }}
1_{K_{n_{l}}}\left(  \varphi\right)  e^{\left\langle \varphi,J\right\rangle
}\text{ }d\mathbb{P}.
\]
The result follows by the dominated convergence theorem, by taking first the
limit $n_{l}\rightarrow\infty$, and then the limit $l\rightarrow\infty$, and
using the fact that $e^{\left\langle \varphi,J\right\rangle }$ is integrable.

(iii) By Lemma \ref{Lemma11} and Remark \ref{Nota_3},%
\[%
{\displaystyle\int\limits_{\mathcal{L}_{\mathbb{R}}^{l}\left(  \mathbb{Q}%
_{p}^{N}\right)  }}
e^{\left\langle \varphi,J\right\rangle }\text{ }d\mathbb{P}_{l}\left(
\varphi\right)  =%
{\displaystyle\int\limits_{\mathcal{L}_{\mathbb{R}}^{l}\left(  \mathbb{Q}%
_{p}^{N}\right)  }}
e^{\left\langle \varphi,J\right\rangle }\text{ }d\mathbb{P}\left(
\varphi\right)  \mathbb{\leq}%
{\displaystyle\int\limits_{\mathcal{L}_{\mathbb{R}}^{\prime}\left(
\mathbb{Q}_{p}^{N}\right)  }}
e^{\left\langle W,J\right\rangle }\text{ }d\mathbb{P}\left(  W\right)
\mathbb{<\infty}\text{.}%
\]
We now use that%
\[
\mathcal{Z}^{\left(  l\right)  }(J)\leq\frac{%
{\displaystyle\int\limits_{\mathcal{L}_{\mathbb{R}}^{l}\left(  \mathbb{Q}%
_{p}^{N}\right)  }}
e^{\left\langle \varphi,J\right\rangle }\text{ }d\mathbb{P}_{l}}{%
{\displaystyle\int\limits_{\mathcal{L}_{\mathbb{R}}^{l}\left(  \mathbb{Q}%
_{p}^{N}\right)  }}
\text{ }d\mathbb{P}_{l}}.
\]

(iv) It is sufficient to show that%
\[
\lim_{l\rightarrow\infty}%
{\displaystyle\int\limits_{\mathcal{L}_{\mathbb{R}}^{l}\left(  \mathbb{Q}%
_{p}^{N}\right)  }}
e^{-E_{\text{int}}\left(  \varphi\right)  +\left\langle \varphi,J\right\rangle
}\text{ }d\mathbb{P}_{l}=%
{\displaystyle\int\limits_{\mathcal{L}_{\mathbb{R}}\left(  \mathbb{Q}_{p}%
^{N}\right)  }}
e^{-E_{\text{int}}\left(  \varphi\right)  +\left\langle \varphi,J\right\rangle
}\text{ }d\mathbb{P}.
\]
This identity is established by using the reasoning given in the second part.
\end{proof}

\begin{definition}
For $\theta\in\mathcal{L}_{\mathbb{R}}\left(  \mathbb{Q}_{p}^{N}\right)  $,
the functional derivative ${\Huge D}_{\theta}\mathcal{Z}(J)$ of $\mathcal{Z}%
(J)$ is defined as
\[
{\Huge D}_{\theta}\mathcal{Z}(J)=\lim_{\epsilon\rightarrow0}\frac
{\mathcal{Z}(J+\epsilon\theta)-\mathcal{Z}(J)}{\epsilon}=\left[  \frac
{d}{d\epsilon}\mathcal{Z}(J+\epsilon\theta)\right]  _{\epsilon=0}.
\]

\end{definition}

\begin{lemma}
\label{Lemma18}Let $\theta_{1}$,\ldots,$\theta_{m}$ be test functions from
$\mathcal{L}_{\mathbb{R}}\left(  \mathbb{Q}_{p}^{N}\right)  $. The functional
derivative ${\Huge D}_{\theta_{1}}\cdots{\Huge D}_{\theta_{m}}\mathcal{Z}(J)$
exists, and the following formula holds true:
\begin{equation}
{\Huge D}_{\theta_{1}}\cdots{\Huge D}_{\theta_{m}}\mathcal{Z}(J)=\frac
{1}{\mathcal{Z}_{0}}\text{ \ }%
{\textstyle\int\limits_{\mathcal{L}_{\mathbb{R}}\left(  \mathbb{Q}_{p}%
^{N}\right)  }}
e^{-E_{\text{int}}(\varphi)+\langle\varphi,J\rangle}\left(
{\textstyle\prod\limits_{i=1}^{m}}
\left\langle \varphi,\theta_{i}\right\rangle \right)  d\mathbb{P}(\varphi).
\label{Eq_30}%
\end{equation}
Furthermore, the functional derivative ${\Huge D}_{\theta_{1}}\cdots
{\Huge D}_{\theta_{m}}\mathcal{Z}(J)$ can be uniquely identified with the
distribution%
\begin{equation}%
{\textstyle\prod\limits_{i=1}^{m}}
\theta_{i}\left(  x_{i}\right)  \rightarrow\frac{1}{\mathcal{Z}_{0}}\text{ \ }%
{\textstyle\idotsint\limits_{\mathbb{Q}_{p}^{N}\times\cdots\times
\mathbb{Q}_{p}^{N}}}
\text{ }%
{\textstyle\prod\limits_{i=1}^{m}}
\theta_{i}\left(  x_{i}\right)  \left\{
{\textstyle\int\limits_{\mathcal{L}_{\mathbb{R}}\left(  \mathbb{Q}_{p}%
^{N}\right)  }}
e^{-E_{\text{int}}(\varphi)+\langle\varphi,J\rangle}%
{\textstyle\prod\limits_{i=1}^{m}}
\varphi\left(  x_{i}\right)  d\mathbb{P}(\varphi)\right\}
{\textstyle\prod\limits_{i=1}^{m}}
d^{N}x_{i} \label{Eq_31}%
\end{equation}
from $\mathcal{L}_{\mathbb{R}}^{\prime}\left(  \left(  \mathbb{Q}_{p}%
^{N}\right)  ^{m}\right)  $.
\end{lemma}

\begin{proof}
We first compute%
\[
\left[  \frac{d}{d\epsilon}\mathcal{Z}(J+\epsilon\theta_{m})\right]
_{\epsilon=0}=\frac{1}{\mathcal{Z}_{0}}\lim_{\epsilon\rightarrow0}%
{\textstyle\int\limits_{\mathcal{L}_{\mathbb{R}}\left(  \mathbb{Q}_{p}%
^{N}\right)  }}
e^{-E_{\text{int}}(\varphi)+\langle\varphi,J\rangle}\left(  \frac
{e^{\epsilon\langle\varphi,\theta_{m}\rangle}-1}{\epsilon}\right)
\ d\mathbb{P}(\varphi).
\]
We consider the case $\epsilon\rightarrow0^{+}$, the other limit is treated in
a similar way. For $\epsilon>0$ sufficiently small, by using the mean value
theorem,
\[
\frac{e^{\epsilon\langle\varphi,\theta_{m}\rangle}-1}{\epsilon}=\left\langle
\varphi,\theta_{m}\right\rangle e^{\epsilon_{0}\langle\varphi,\theta
_{m}\rangle}\text{ where }\epsilon_{0}\in\left(  0,\epsilon\right)  .
\]
Then, by using\ $e^{-E_{\text{int}}(\varphi)}\leq1$\ and Remark \ref{Nota_3},
\[
e^{-E_{\text{int}}(\varphi)+\langle\varphi,J\rangle}\left(  \frac
{e^{\epsilon\langle\varphi,\theta_{m}\rangle}-1}{\epsilon}\right)
=\left\langle \varphi,\theta_{m}\right\rangle e^{-E_{\text{int}}%
(\varphi)+\langle\varphi,J+\epsilon_{0}\theta_{m}\rangle}%
\]
is an integrable function. Now, by applying the dominated convergence theorem,%
\begin{equation}
{\Huge D}_{\theta_{m}}\mathcal{Z}(J)=\left[  \frac{d}{d\epsilon}%
\mathcal{Z}(J+\epsilon\theta_{m})\right]  _{\epsilon=0}=\frac{1}%
{\mathcal{Z}_{0}}\text{ \ }%
{\textstyle\int\limits_{\mathcal{L}_{\mathbb{R}}\left(  \mathbb{Q}_{p}%
^{N}\right)  }}
e^{-E_{\text{int}}(\varphi)+\langle\varphi,J\rangle}\left\langle
\varphi,\theta_{m}\right\rangle \ d\mathbb{P}(\varphi). \label{Eq_Der_Funt}%
\end{equation}
By Remark \ref{Nota_3}, $e^{-E_{\text{int}}(\varphi)+\langle\varphi,J\rangle
}\left\langle \varphi,\theta_{m}\right\rangle \in\left(  L_{\mathbb{R}}%
^{1}\right)  $, then, further derivatives can be computed using
(\ref{Eq_Der_Funt}).

Finally, formula (\ref{Eq_31}) is obtained from (\ref{Eq_30}) by using
Fubini's theorem and Remark \ref{Nota_Nuclear}:%
\[
{\Huge D}_{\theta_{1}}\cdots{\Huge D}_{\theta_{m}}\mathcal{Z}(J)=\frac
{1}{\mathcal{Z}_{0}}\text{ }%
{\textstyle\int\limits_{\mathcal{L}_{\mathbb{R}}\left(  \mathbb{Q}_{p}%
^{N}\right)  }}
e^{-E_{\text{int}}(\varphi)+\langle\varphi,J\rangle}\left\{  \text{ }%
{\textstyle\idotsint\limits_{\mathbb{Q}_{p}^{N}\times\cdots\times
\mathbb{Q}_{p}^{N}}}
\text{ }%
{\textstyle\prod\limits_{i=1}^{m}}
\theta_{i}\left(  x_{i}\right)  \varphi\left(  x_{i}\right)
{\textstyle\prod\limits_{i=1}^{m}}
d^{N}x_{i}\right\}  d\mathbb{P}(\varphi).
\]

\end{proof}

\begin{remark}
\label{Nota_eqqui_Der}In an alternative way, one can define the functional
derivative $\frac{\delta}{\delta J\left(  y\right)  }\mathcal{Z}(J)$\ of
$\mathcal{Z}(J)$ as the distribution from $\mathcal{L}_{\mathbb{R}}^{\prime
}\left(  \mathbb{Q}_{p}^{N}\right)  $ satisfying%
\[%
{\textstyle\int\limits_{\mathbb{Q}_{p}^{N}}}
\theta\left(  y\right)  \left(  \frac{\delta}{\delta J\left(  y\right)
}\mathcal{Z}(J)\right)  \left(  y\right)  d^{N}y=\left[  \frac{d}{d\epsilon
}\mathcal{Z}(J+\epsilon\theta)\right]  _{\epsilon=0}.
\]
Using this notation and formula (\ref{Eq_31}), we obtain that%
\[
\frac{\delta}{\delta J\left(  x_{1}\right)  }\cdots\frac{\delta}{\delta
J\left(  x_{m}\right)  }\mathcal{Z}(J)=\frac{1}{\mathcal{Z}_{0}}\text{ \ }%
{\textstyle\int\limits_{\mathcal{L}_{\mathbb{R}}\left(  \mathbb{Q}_{p}%
^{N}\right)  }}
e^{-E_{\text{int}}(\varphi)+\langle\varphi,J\rangle}\left(
{\textstyle\prod\limits_{i=1}^{m}}
\varphi\left(  x_{i}\right)  \right)  d\mathbb{P}(\varphi)\in\mathcal{L}%
_{\mathbb{R}}^{\prime}\left(  \left(  \mathbb{Q}_{p}^{N}\right)  ^{m}\right)
.
\]

\end{remark}

\begin{remark}
\label{Nota_10}Consider the probability measure space \ $\left(
\mathcal{L}_{\mathbb{R}}\left(  \mathbb{Q}_{p}^{N}\right)  ,\mathcal{B}%
\cap\mathcal{L}_{\mathbb{R}},\frac{1}{\mathcal{Z}_{0}}\mathbb{P}\right)  $,
where $\mathcal{B}\cap\mathcal{L}_{\mathbb{R}}$ denotes the $\sigma$-algebra
generated by the cylinder subsets of $\mathcal{L}_{\mathbb{R}}$. Given
$\theta_{1}$,\ldots,$\theta_{m}$ test functions from $\mathcal{L}_{\mathbb{R}%
}\left(  \mathbb{Q}_{p}^{N}\right)  $, \ we attach them the following random
variable:%
\[%
\begin{array}
[c]{lll}%
\mathcal{L}_{\mathbb{R}}\left(  \mathbb{Q}_{p}^{N}\right)  & \rightarrow &
\mathbb{R}\\
\varphi & \rightarrow &
{\textstyle\prod\limits_{i=1}^{m}}
\left\langle \varphi,\theta_{i}\right\rangle .
\end{array}
\]
The expected value of this variable is given by%
\[
{\Huge D}_{\theta_{1}}\cdots{\Huge D}_{\theta_{m}}\mathcal{Z}(J)\mid
_{J=0}=\frac{1}{\mathcal{Z}_{0}}\text{ \ }%
{\textstyle\int\limits_{\mathcal{L}_{\mathbb{R}}\left(  \mathbb{Q}_{p}%
^{N}\right)  }}
e^{-E_{\text{int}}(\varphi)}\left(
{\textstyle\prod\limits_{i=1}^{m}}
\left\langle \varphi,\theta_{i}\right\rangle \right)  d\mathbb{P}(\varphi).
\]
An alternative description of the expected value is given by%
\[
\frac{\delta}{\delta J\left(  x_{1}\right)  }\cdots\frac{\delta}{\delta
J\left(  x_{m}\right)  }\mathcal{Z}(J)\mid_{J=0}=\frac{1}{\mathcal{Z}_{0}%
}\text{ \ }%
{\textstyle\int\limits_{\mathcal{L}_{\mathbb{R}}\left(  \mathbb{Q}_{p}%
^{N}\right)  }}
e^{-E_{\text{int}}(\varphi)}\left(
{\textstyle\prod\limits_{i=1}^{m}}
\varphi\left(  x_{i}\right)  \right)  d\mathbb{P}(\varphi).
\]

\end{remark}

As a conclusion we have the following result:

\begin{proposition}
\label{Prop1}The correlations functions $G^{\left(  m\right)  }\left(
x_{1},\ldots,x_{m}\right)  \in\mathcal{L}_{\mathbb{R}}^{\prime}\left(  \left(
\mathbb{Q}_{p}^{N}\right)  ^{m}\right)  $ are given by
\[
G^{\left(  m\right)  }\left(  x_{1},\ldots,x_{m}\right)  =\frac{\mathcal{Z}%
_{0}}{\mathcal{Z}}\frac{\delta}{\delta J\left(  x_{1}\right)  }\cdots
\frac{\delta}{\delta J\left(  x_{m}\right)  }\mathcal{Z}(J)\mid_{J=0}.
\]

\end{proposition}

\subsection{Free-field theory}

\subsubsection{The propagators}

We take $\delta>N$, and $\gamma$, $\alpha_{2}>0$ as before. For $J\in
\mathcal{L}_{\mathbb{R}}$, the equation%
\begin{equation}
\left(  \frac{\gamma}{2}W\left(  \partial,\delta\right)  +\frac{\alpha_{2}}%
{2}\right)  \varphi_{0}=J \label{Equation_J}%
\end{equation}
has unique solution $\varphi_{0}\in\mathcal{L}_{\mathbb{R}}$. Indeed,
$\widehat{\varphi_{0}}\left(  \kappa\right)  =\frac{\widehat{J}\left(
\kappa\right)  }{\frac{\gamma}{2}A_{w_{\delta}}(\left\Vert \kappa\right\Vert
_{p})+\frac{\alpha_{2}}{2}}$ is a test function satisfying $\widehat
{\varphi_{0}}\left(  0\right)  =0$. On the other hand, solving equation
(\ref{Equation_J}) in $\mathcal{D}_{\mathbb{R}}^{\prime}$, we have
\[
\varphi_{0}\left(  x\right)  =\mathcal{F}_{\kappa\rightarrow x}^{-1}(\frac
{1}{\frac{\gamma}{2}A_{w_{\delta}}(\left\Vert \kappa\right\Vert _{p}%
)+\frac{\alpha_{2}}{2}})\ast J(x)=G(\left\Vert x\right\Vert _{p})\ast
J(x)\text{,}%
\]
where $\mathcal{F}_{\kappa\rightarrow x}^{-1}$ denotes the Fourier transform
from $\mathcal{D}^{\prime}$ into $\mathcal{D}^{\prime}$, which means that
equation (\ref{Equation_J}) has a unique solution $\varphi_{0}\left(
x\right)  =G(\left\Vert x\right\Vert _{p})\ast J(x)$ in $\mathcal{L}%
_{\mathbb{R}}$, where $G(\left\Vert x\right\Vert _{p})$ is the `standard Green
function'. This means that the UV and IF behavior \ of the propagators are not
altered if we use Lizorkin spaces in the construction of $p$-adic QFTs.

We now discuss the singular behavior of the Green function in the case \ of
Taibleson-Vladimirov operator:%
\[
G(x;\beta,\gamma,\alpha_{2})=\mathcal{F}_{\kappa\rightarrow x}^{-1}(\frac
{1}{\frac{\gamma}{2}\left\Vert \kappa\right\Vert _{p}^{\beta}+\frac{\alpha
_{2}}{2}}),
\]
where $\beta,\gamma,\alpha_{2}>0$. In this case $G(x;\beta,\gamma,\alpha_{2})$
is continuous on $\mathbb{Q}_{p}^{N}\smallsetminus\left\{  0\right\}  $. If
$\beta>N$, then $G(x;\beta,\gamma,\alpha_{2})$ is continuous. For $0<\beta\leq
N$, $G(x;\beta,\gamma,\alpha_{2})$ is locally constant on $\mathbb{Q}_{p}%
^{N}\smallsetminus\left\{  0\right\}  $, and
\[
\left\vert G(x;\beta,\gamma,\alpha_{2})\right\vert \leq\left\{
\begin{array}
[c]{cc}%
C\left\Vert x\right\Vert _{p}^{\beta-N} & \text{for }0<\beta<N\\
& \\
C_{0}-C_{1}\ln\left\Vert x\right\Vert _{p} & \text{for }N=\beta\text{,}%
\end{array}
\right.
\]
for $\left\Vert x\right\Vert _{p}\leq1$, \ where $C$, $C_{0}$, $C_{1}$ are
positive constants; $\left\vert G(x;\beta,\gamma,\alpha_{2})\right\vert \leq
C_{1}\left\Vert x\right\Vert _{p}^{-\beta-N}$ as $\left\Vert x\right\Vert
_{p}\rightarrow\infty$. Finally, $G(x;\beta,\gamma,\alpha_{2})\geq0$ on
$\mathbb{Q}_{p}^{N}\smallsetminus\left\{  0\right\}  $, see e.g.
\cite[Proposition 11.1]{KKZuniga}.

The behavior at the origin of the Green functions considered here depends in
an intricate way on the parameters of the QFT considered and on the dimension.
This behavior plays a central role in the renormalization of the QFTs
\ presented here. The renormalization will be considered in a forthcoming article.

\begin{theorem}
\label{Prop2}Set $\mathcal{Z}_{0}(J):=\mathcal{Z}(J;\delta,\gamma,\alpha
_{2},0)$, then
\[
\mathcal{Z}_{0}(J)=\mathcal{N}_{0}^{\prime}\exp\left\{  \int_{\mathbb{Q}%
_{p}^{N}}\int_{\mathbb{Q}_{p}^{N}}J(x)G(\left\Vert x-y\right\Vert
_{p})J(y)d^{N}x\text{ }d^{N}y\right\}  ,
\]
where $\mathcal{N}_{0}^{\prime}$ denotes a normalization constant.
\end{theorem}

\begin{proof}
We take $\varphi_{0}$, $J\in\mathcal{L}_{\mathbb{R}}$, where $\varphi_{0}$ is
the solution of equation(\ref{Equation_J}). We now change variables in
$\mathcal{Z}_{0}(J)$ as $\varphi=\varphi_{0}+\varphi^{\prime}$,%
\begin{align*}
\mathcal{Z}_{0}(J) &  =\frac{1}{\mathcal{Z}_{0}}%
{\displaystyle\int\limits_{\mathcal{L}_{\mathbb{R}}\left(  \mathbb{Q}_{p}%
^{N}\right)  }}
e^{\left\langle \varphi,J\right\rangle }\text{ }d\mathbb{P=}\frac
{e^{\left\langle \varphi_{0},J\right\rangle }}{\mathcal{Z}_{0}}%
{\displaystyle\int\limits_{\mathcal{L}_{\mathbb{R}}\left(  \mathbb{Q}_{p}%
^{N}\right)  }}
e^{\left\langle \varphi^{\prime},J\right\rangle }\text{ }d\mathbb{P}^{\prime
}\left(  \varphi^{\prime}\right)  \\
&  =\left(  \frac{1}{\mathcal{Z}_{0}}\text{ }%
{\displaystyle\int\limits_{\mathcal{L}_{\mathbb{R}}\left(  \mathbb{Q}_{p}%
^{N}\right)  }}
e^{\left\langle \varphi^{\prime},\left(  \frac{\gamma}{2}W\left(
\partial,\delta\right)  +\frac{\alpha_{2}}{2}\right)  \varphi_{0}\right\rangle
}\text{ }d\mathbb{P}^{\prime}\left(  \varphi^{\prime}\right)  \right)
e^{\left\langle G\ast J,J\right\rangle }\\
&  =\mathcal{N}_{0}^{\prime}e^{\left\langle G\ast J,J\right\rangle
}=\mathcal{N}_{0}^{\prime}\exp\left\{  \int_{\mathbb{Q}_{p}^{N}}%
\int_{\mathbb{Q}_{p}^{N}}J(x)G(\left\Vert x-y\right\Vert _{p})J(y)d^{N}x\text{
}d^{N}y\right\}  .
\end{align*}
Furthermore, by using (\ref{Eq_Char_func_2A}), the characteristic functional
of the measure $\mathbb{P}^{\prime}$\ is
\[%
{\textstyle\int\limits_{\mathcal{L}_{\mathbb{R}}^{\prime}\left(
\mathbb{Q}_{p}^{N}\right)  }}
e^{\sqrt{-1}\langle T,f\rangle}d\mathbb{P}^{\prime}(T)=e^{-\sqrt{-1}%
\langle\varphi_{0},f\rangle-\frac{1}{2}\mathbb{B}(f,f)},\ \ f\in
\mathcal{L}_{\mathbb{R}}\left(  \mathbb{Q}_{p}^{N}\right)  ,
\]
which means that $\mathbb{P}^{\prime}$ is a Gaussian measure with mean
functional $\langle\varphi_{0},\cdot\rangle$ and correlation functional
$\mathbb{B}(\cdot,\cdot)$.
\end{proof}

The correlation functions $G_{0}^{\left(  m\right)  }(x_{1},\ldots,x_{m})$ of
the free-field theory are obtained \ from the functional derivatives of
$\mathcal{Z}_{0}(J)$ at $J=0$:

\begin{theorem}
\label{Prop3}%
\begin{gather*}
G_{0}^{\left(  m\right)  }(x_{1},\ldots,x_{m})=\left[  \frac{\delta}{\delta
J\left(  x_{1}\right)  }\cdots\frac{\delta}{\delta J\left(  x_{m}\right)
}\mathcal{Z}_{0}(J)\right]  _{J=0}\\
=\mathcal{N}_{0}^{\prime}\text{ }\frac{\delta}{\delta J\left(  x_{1}\right)
}\cdots\frac{\delta}{\delta J\left(  x_{m}\right)  }\exp\left\{
\int_{\mathbb{Q}_{p}^{N}}\int_{\mathbb{Q}_{p}^{N}}J(x)G(\left\Vert
x-y\right\Vert _{p})J(y)d^{N}x\text{ }d^{N}y\right\}  \mid_{J=0}.
\end{gather*}

\end{theorem}

\begin{remark}
The random variable $\varphi\left(  x_{i}\right)  $ corresponds to the random
variable $\left\langle W,\varphi\right\rangle $, for some $W=W(x_{i}%
)\in\mathcal{L}_{\mathbb{R}}^{\prime}\left(  \mathbb{Q}_{p}^{N}\right)  $, see
Remark \ref{Nota_10}, which is Gaussian with mean zero and variance
$\left\Vert \varphi\right\Vert _{2}^{2}$, see e.g. \cite[Lemma 2.1.5]{Obata}.
Then, the correlation functions $G_{0}^{\left(  m\right)  }(x_{1},\ldots
,x_{m})$ obey to Wick's theorem:%
\begin{equation}
\frac{1}{\mathcal{Z}_{0}}%
{\displaystyle\int\limits_{\mathcal{L}_{\mathbb{R}}\left(  \mathbb{Q}_{p}%
^{N}\right)  }}
{\displaystyle\prod\limits_{i=1}^{m}}
\varphi\left(  x_{i}\right)  d\mathbb{P=}\left\{
\begin{array}
[c]{lll}%
0 & \text{if} & m\text{ is not even}\\
&  & \\%
{\textstyle\sum\limits_{\text{pairings}}}
\mathbb{E}(\varphi\left(  x_{i_{1}}\right)  \varphi\left(  x_{j_{1}}\right)
)\cdots\mathbb{E}(\varphi\left(  x_{i_{n}}\right)  \varphi\left(  x_{j_{n}%
}\right)  ) & \text{if} & m=2n,
\end{array}
\right.  \label{Wick-Expansion}%
\end{equation}
where
\[
\mathbb{E}(\varphi\left(  x_{i}\right)  \varphi\left(  x_{j}\right)
):=\frac{1}{\mathcal{Z}_{0}}%
{\displaystyle\int\limits_{\mathcal{L}_{\mathbb{R}}\left(  \mathbb{Q}_{p}%
^{N}\right)  }}
\varphi\left(  x_{i}\right)  \varphi\left(  x_{j}\right)  d\mathbb{P}%
\]
and $%
{\textstyle\sum\limits_{\text{pairings}}}
$ means the sum over all $\frac{\left(  2n!\right)  }{2^{n}n!}$ ways of
writing $1,\ldots,2n$ as $n$ distinct (unordered) pairs $(i_{1},j_{1})$%
,\ldots,$(i_{n},j_{n})$, see e.g. \cite[Proposition 1.2]{Simon-0}.

For $n=2$, $G_{0}^{\left(  2\right)  }$ is the free two-point function or the
free propagator of the field:%
\begin{align*}
G_{0}^{\left(  2\right)  }\left(  x_{1},x_{2}\right)   &  =\mathcal{N}%
_{0}^{\prime}\text{ }\frac{\delta}{\delta J\left(  x_{1}\right)  }\frac
{\delta}{\delta J\left(  x_{2}\right)  }\exp\left\{  \int_{\mathbb{Q}_{p}^{N}%
}\int_{\mathbb{Q}_{p}^{N}}J(x)G(\left\Vert x-y\right\Vert _{p})J(y)d^{N}%
x\text{ }d^{N}y\right\}  \mid_{J=0}\\
&  =2\mathcal{N}_{0}^{\prime}\text{ }G(\left\Vert x_{1}-x_{2}\right\Vert
_{p})\in\mathcal{L}_{\mathbb{R}}^{\prime}(\mathbb{Q}_{p}^{N}\times
\mathbb{Q}_{p}^{N}).
\end{align*}
By using Wick's theorem all the $2n$-point functions can be expressed as sums
of products of two-point functions:%
\[
G_{0}^{\left(  2n\right)  }(x_{1},\ldots,x_{2n})=%
{\textstyle\sum\limits_{\text{pairings}}}
G(\left\Vert x_{i_{1}}-x_{j_{1}}\right\Vert _{p})\cdots G(\left\Vert x_{i_{n}%
}-x_{j_{n}}\right\Vert _{p}).
\]
Notice that $G_{0}^{\left(  2n\right)  }(x_{1},\ldots,x_{2n})$ is singular at
$x_{i_{1}}-x_{j_{1}}=\cdots=x_{i_{n}}-x_{j_{n}}=0$, where $\left(  i_{k}%
,j_{k}\right)  $ runs over all the possible pairings of the variables
$x_{1},\ldots,x_{2n}$. This set is a closed subset of $\mathbb{Q}_{p}^{2N}$.
\end{remark}

\subsection{Perturbation expansions for $\varphi^{4}$-theories}

In this section we assume that $\mathcal{P}(\varphi)=\varphi^{4}$. This
hypothesis allow us to provide explicit formulas which completely similar to
the classical ones, see e.g. \cite[Chapter 2]{Kleinert et al}. At any rate,
the techniques presented here can be applied to polynomial interactions of
type (\ref{Poly_interactions}).

The existence of a convergent power series expansion for $Z(J)$ (\textit{the
perturbation expansion}) in the coupling parameter $\alpha_{4}$ follows from
the fact that $\exp\left(  -E_{\text{int}}(\varphi)+\left\langle
\varphi,J\right\rangle \right)  $ is an integrable function, see Lemma
\ref{Lemma17} (i), by using the dominated convergence theorem, more precisely,
we have%
\begin{gather}
\mathcal{Z}(J)=\mathcal{Z}_{0}(J)+\frac{1}{\mathcal{Z}_{0}}%
{\displaystyle\sum\limits_{m=1}^{\infty}}
\frac{1}{m!}\left(  \frac{-\alpha_{4}}{4}\right)  ^{m}%
{\displaystyle\int\limits_{\mathcal{L}_{\mathbb{R}}\left(  \mathbb{Q}_{p}%
^{N}\right)  }}
\left\{
{\displaystyle\int\limits_{\left(  \mathbb{Q}_{p}^{N}\right)  ^{m}}}
\left(
{\textstyle\prod\limits_{i=1}^{m}}
\varphi^{4}\left(  z_{i}\right)  \right)  e^{\left\langle \varphi
,J\right\rangle }%
{\textstyle\prod\limits_{i=1}^{m}}
d^{N}z_{i}\right\}  d\mathbb{P}(\varphi)\nonumber\\
=:\mathcal{Z}_{0}(J)+%
{\displaystyle\sum\limits_{m=1}^{\infty}}
\mathcal{Z}_{m}(J), \label{Eq_35}%
\end{gather}
where%
\[
\mathcal{Z}_{0}(J)=\frac{1}{\mathcal{Z}_{0}}%
{\displaystyle\int\limits_{\mathcal{L}_{\mathbb{R}}\left(  \mathbb{Q}_{p}%
^{N}\right)  }}
e^{\left\langle \varphi,J\right\rangle }d\mathbb{P}(\varphi).
\]
In the case $m\geq1$, by using that $\mathcal{A}$ is an algebra (see Remark
\ref{Nota_3} and Lemma \ref{Lemma16}), we can apply Fubini's theorem to obtain
that%
\begin{align*}
\mathcal{Z}_{m}(J)  &  :=\frac{1}{\mathcal{Z}_{0}\text{ }m!}\left(
\frac{-\alpha_{4}}{4}\right)  ^{m}%
{\displaystyle\int\limits_{\mathcal{L}_{\mathbb{R}}\left(  \mathbb{Q}_{p}%
^{N}\right)  }}
\left\{
{\displaystyle\int\limits_{\left(  \mathbb{Q}_{p}^{N}\right)  ^{m}}}
\left(
{\textstyle\prod\limits_{i=1}^{m}}
\varphi^{4}\left(  z_{i}\right)  \right)  e^{\left\langle \varphi
,J\right\rangle }%
{\textstyle\prod\limits_{i=1}^{m}}
d^{N}z_{i}\right\}  d\mathbb{P}(\varphi)\\
&  =\frac{1}{\mathcal{Z}_{0}\text{ }m!}\left(  \frac{-\alpha_{4}}{4}\right)
^{m}%
{\displaystyle\int\limits_{\left(  \mathbb{Q}_{p}^{N}\right)  ^{m}}}
\left\{
{\displaystyle\int\limits_{\mathcal{L}_{\mathbb{R}}\left(  \mathbb{Q}_{p}%
^{N}\right)  }}
\left(
{\textstyle\prod\limits_{i=1}^{m}}
\varphi^{4}\left(  z_{i}\right)  \right)  e^{\left\langle \varphi
,J\right\rangle }d\mathbb{P}(\varphi)\right\}
{\textstyle\prod\limits_{i=1}^{m}}
d^{N}z_{i}.
\end{align*}
Then
\begin{equation}
\mathcal{Z}_{m}(0)=\frac{1}{m!}\left(  \frac{-\alpha_{4}}{4}\right)  ^{m}%
{\displaystyle\int\limits_{\left(  \mathbb{Q}_{p}^{N}\right)  ^{m}}}
G_{0}^{\left(  4m\right)  }\left(  z_{1},z_{1},z_{1},z_{1},\ldots,z_{m}%
,z_{m},z_{m},z_{m}\right)
{\textstyle\prod\limits_{i=1}^{m}}
d^{N}z_{i}, \label{Eq_36}%
\end{equation}
for $m\geq1$. Therefore from (\ref{Eq_35})-(\ref{Eq_36}), with $J=0$, and
using $\mathcal{Z}=\mathcal{Z}(0)$, $\mathcal{Z}_{m}(0):=\mathcal{Z}_{m}$, for
$m\geq1$,%
\[
\mathcal{Z}=1+%
{\displaystyle\sum\limits_{m=1}^{\infty}}
\mathcal{Z}_{m}.
\]
Now by using Propositions \ref{Prop1}, \ref{Prop3}\ and (\ref{Eq_35}),%

\begin{gather}
G^{\left(  n\right)  }\left(  x_{1},\ldots,x_{n}\right)  =\frac{\mathcal{Z}%
_{0}}{\mathcal{Z}}\left[  \frac{\delta}{\delta J\left(  x_{1}\right)  }%
\cdots\frac{\delta}{\delta J\left(  x_{n}\right)  }\mathcal{Z}(J)\right]
_{J=0}\nonumber\\
=\frac{\mathcal{Z}_{0}}{\mathcal{Z}}\left[  \frac{\delta}{\delta J\left(
x_{1}\right)  }\cdots\frac{\delta}{\delta J\left(  x_{n}\right)  }%
\mathcal{Z}_{0}(J)\right]  _{J=0}+\frac{\mathcal{Z}_{0}}{\mathcal{Z}}\left[
\frac{\delta}{\delta J\left(  x_{1}\right)  }\cdots\frac{\delta}{\delta
J\left(  x_{n}\right)  }%
{\displaystyle\sum\limits_{m=1}^{\infty}}
\mathcal{Z}_{m}(J)\right]  _{J=0}\nonumber\\
=\frac{\mathcal{Z}_{0}}{\mathcal{Z}}G_{0}^{\left(  n\right)  }\left(
x_{1},\ldots,x_{n}\right)  +\frac{\mathcal{Z}_{0}}{\mathcal{Z}}\left[
\frac{\delta}{\delta J\left(  x_{1}\right)  }\cdots\frac{\delta}{\delta
J\left(  x_{n}\right)  }%
{\displaystyle\sum\limits_{m=1}^{\infty}}
\mathcal{Z}_{m}(J)\right]  _{J=0}. \label{Eq_36A}%
\end{gather}

\begin{lemma}%
\begin{gather*}
\frac{\delta}{\delta J\left(  x_{1}\right)  }\cdots\frac{\delta}{\delta
J\left(  x_{n}\right)  }%
{\displaystyle\sum\limits_{m=1}^{\infty}}
\mathcal{Z}_{m}(J)=\\
\frac{1}{\mathcal{Z}_{0}}\text{ }%
{\displaystyle\sum\limits_{m=1}^{\infty}}
\frac{1}{m!}\left(  \frac{-\alpha_{4}}{4}\right)  ^{m}%
{\displaystyle\int\limits_{\left(  \mathbb{Q}_{p}^{N}\right)  ^{m}}}
\left\{
{\displaystyle\int\limits_{\mathcal{L}_{\mathbb{R}}\left(  \mathbb{Q}_{p}%
^{N}\right)  }}
\left(
{\textstyle\prod\limits_{i=1}^{m}}
\varphi^{4}\left(  z_{i}\right)  \right)  \left(
{\textstyle\prod\limits_{i=1}^{n}}
\varphi\left(  x_{i}\right)  \right)  e^{\left\langle \varphi,J\right\rangle
}d\mathbb{P}(\varphi)\right\}
{\textstyle\prod\limits_{i=1}^{m}}
d^{N}z_{i}.
\end{gather*}

\end{lemma}

\begin{proof}
We recall that by the proof of Lemma \ref{Lemma16},
\[
\mathcal{J}\left(  \varphi\right)  :=%
{\displaystyle\int\limits_{\left(  \mathbb{Q}_{p}^{N}\right)  ^{m}}}
\left(
{\textstyle\prod\limits_{i=1}^{m}}
\varphi^{4}\left(  z_{i}\right)  \right)
{\textstyle\prod\limits_{i=1}^{m}}
d^{N}z_{i}%
\]
is a \ finite sum of terms of the form
\[
\left(
{\displaystyle\prod\limits_{k=1}^{m}}
p^{lNe_{i_{k}}}\left\langle \varphi,W_{\boldsymbol{j}_{k}}\right\rangle
^{e_{i_{k}}}\right)
{\displaystyle\prod\limits_{k=1}^{m}}
\Omega\left(  p^{l}\left\Vert x_{k}-\boldsymbol{j}_{k}\right\Vert _{p}\right)
,
\]
then by the definition of $\mathcal{Z}_{m}(J)$ and Fubini's theorem, it is
sufficient to compute%
\begin{gather*}
\frac{\delta}{\delta J\left(  x_{1}\right)  }\cdots\frac{\delta}{\delta
J\left(  x_{n}\right)  }%
{\displaystyle\sum\limits_{m=1}^{\infty}}
\frac{1}{\mathcal{Z}_{0}\text{ }m!}\left(  \frac{-\alpha_{4}}{4}\right)
^{m}\times\\%
{\displaystyle\int\limits_{\mathcal{L}_{\mathbb{R}}\left(  \mathbb{Q}_{p}%
^{N}\right)  }}
\left\{  \left(
{\displaystyle\prod\limits_{k=1}^{m}}
p^{lNe_{i_{k}}}\left\langle \varphi,W_{\boldsymbol{j}_{k}}\right\rangle
^{e_{i_{k}}}\right)  e^{\left\langle \varphi,J\right\rangle }\right\}
d\mathbb{P}(\varphi).
\end{gather*}
We first establish that%
\begin{gather*}
{\Huge D}_{\theta_{1}}\left\{
{\displaystyle\sum\limits_{m=1}^{\infty}}
\frac{1}{\mathcal{Z}_{0}m!}\left(  \frac{-\alpha_{4}}{4}\right)  ^{m}%
{\displaystyle\int\limits_{\mathcal{L}_{\mathbb{R}}\left(  \mathbb{Q}_{p}%
^{N}\right)  }}
\left\{  \left(
{\displaystyle\prod\limits_{k=1}^{m}}
p^{lNe_{i_{k}}}\left\langle \varphi,W_{\boldsymbol{j}_{k}}\right\rangle
^{e_{i_{k}}}\right)  e^{\left\langle \varphi,J\right\rangle }\right\}
d\mathbb{P}(\varphi)\right\}  \\
=%
{\displaystyle\sum\limits_{m=1}^{\infty}}
\frac{1}{\mathcal{Z}_{0}m!}\left(  \frac{-\alpha_{4}}{4}\right)  ^{m}%
{\displaystyle\int\limits_{\mathcal{L}_{\mathbb{R}}\left(  \mathbb{Q}_{p}%
^{N}\right)  }}
\left\{  \left(
{\displaystyle\prod\limits_{k=1}^{m}}
p^{lNe_{i_{k}}}\left\langle \varphi,W_{\boldsymbol{j}_{k}}\right\rangle
^{e_{i_{k}}}\right)  \left\langle \varphi,\theta_{1}\right\rangle
e^{\left\langle \varphi,J\right\rangle }\right\}  d\mathbb{P}(\varphi),
\end{gather*}
by using the reasoning given in the proof of Lemma \ref{Lemma18}. Since
\[
\left(
{\displaystyle\prod\limits_{k=1}^{m}}
p^{lNe_{i_{k}}}\left\langle \varphi,W_{\boldsymbol{j}_{k}}\right\rangle
^{e_{i_{k}}}\right)  \left\langle \varphi,\theta_{1}\right\rangle
e^{\left\langle \varphi,J\right\rangle }\text{ is an integrable function,}%
\]
cf. Remark \ref{Nota_3}, further derivatives can be calculated in the same
way. Consequently,%
\begin{gather}
{\Huge D}_{\theta_{1}}\cdots{\Huge D}_{\theta_{m}}%
{\displaystyle\sum\limits_{m=1}^{\infty}}
\mathcal{Z}_{m}(J)=\label{Eq_36B}\\
\frac{1}{\mathcal{Z}_{0}}%
{\displaystyle\sum\limits_{m=1}^{\infty}}
\frac{1}{m!}\left(  \frac{-\alpha_{4}}{4}\right)  ^{m}%
{\displaystyle\int\limits_{\mathcal{L}_{\mathbb{R}}\left(  \mathbb{Q}_{p}%
^{N}\right)  }}
\left\{  \left(
{\textstyle\prod\limits_{i=1}^{n}}
\left\langle \varphi,\theta_{i}\right\rangle \right)  e^{\left\langle
\varphi,J\right\rangle }\mathcal{J}\left(  \varphi\right)  d\mathbb{P}%
(\varphi)\right\}  .\nonumber
\end{gather}
The announced formula follows from (\ref{Eq_36B}) by Fubini's theorem.
\end{proof}

Now by using (\ref{Eq_36A}) and Remark \ref{Nota_eqqui_Der}, we have the
following result:

\begin{theorem}
\label{Theorem2}Assume that $\mathcal{P}(\varphi)=\varphi^{4}$. The $n$-point
correlation function of the field $\varphi$ admits the following convergent
power series in the coupling constant:%
\begin{equation}
G^{\left(  n\right)  }\left(  x_{1},\ldots,x_{n}\right)  =\frac{\mathcal{Z}%
_{0}}{\mathcal{Z}}\left\{  G_{0}^{\left(  n\right)  }\left(  x_{1}%
,\ldots,x_{n}\right)  +%
{\displaystyle\sum\limits_{m=1}^{\infty}}
G_{m}^{\left(  n\right)  }\left(  x_{1},\ldots,x_{n}\right)  \right\}  \text{
in }\mathcal{L}_{\mathbb{R}}^{\prime}\left(  \mathbb{Q}_{p}^{N}\right)  ,
\label{Eq_38}%
\end{equation}
where%
\begin{multline*}
G_{m}^{\left(  n\right)  }\left(  x_{1},\ldots,x_{n}\right)  :=\\
\frac{1}{m!}\left(  \frac{-\alpha_{4}}{4}\right)  ^{m}%
{\displaystyle\int\limits_{\left(  \mathbb{Q}_{p}^{N}\right)  ^{m}}}
G_{0}^{\left(  n+4m\right)  }\left(  z_{1},z_{1},z_{1},z_{1},\ldots
,z_{m},z_{m},z_{m},z_{m},x_{1},\ldots,x_{n}\right)
{\textstyle\prod\limits_{i=1}^{m}}
d^{N}z_{i}\in\mathcal{L}_{\mathbb{R}}^{\prime}\left(  \mathbb{Q}_{p}%
^{N}\right)
\end{multline*}

\end{theorem}

It is important to emphasize that formula (\ref{Eq_38}) is an equality between
distributions `with singularities'. The free-field correlation functions
$G_{0}^{\left(  n+4m\right)  }$ in the sum may now Wick-expanded as in
(\ref{Wick-Expansion}) into sums over products of propagators $G_{0}^{\left(
2\right)  }$. Then, like in the classical case, a renormalization procedure is needed.

\section{\label{Section_Wick_rotation}The Wick rotation}

The classical generating functional of $\mathcal{P}(\varphi)$-theory with
Lagrangian density $E_{0}(\varphi)+E_{\text{int}}(\varphi)+E_{\text{source}%
}(\varphi,J)$ in the Minkowski space is%
\[
\mathcal{Z}^{\text{phys}}(J)=\frac{\int D(\varphi)e^{\sqrt{-1}\left\{
E_{0}(\varphi)+E_{\text{int}}(\varphi)+E_{\text{source}}(\varphi,J)\right\}
}}{\int D(\varphi)e^{\sqrt{-1}\left\{  E_{0}(\varphi)+E_{\text{int}}%
(\varphi)\right\}  }}.
\]
A natural $p$-adic analogue of this function is%
\[
\mathcal{Z}_{\mathbb{C}}(J)=\frac{%
{\displaystyle\int\limits_{\mathcal{L}_{\mathbb{R}}\left(  \mathbb{Q}_{p}%
^{N}\right)  }}
e^{\sqrt{-1}\left\{  E_{\text{int}}(\varphi)+E_{\text{source}}(\varphi
,J)\right\}  }d\mathbb{P}(\varphi)}{%
{\displaystyle\int\limits_{\mathcal{L}_{\mathbb{R}}\left(  \mathbb{Q}_{p}%
^{N}\right)  }}
e^{\sqrt{-1}\left\{  E_{0}(\varphi)+E_{\text{int}}(\varphi)\right\}
}d\mathbb{P}(\varphi)}.
\]
Which is a complex-value measure. The key point is that $e^{\sqrt{-1}\left\{
E_{0}(\varphi)+E_{\text{int}}(\varphi)+E_{\text{source}}(\varphi,J)\right\}
}$ is integrable, see \cite[Theorem 1.9]{Hida et al}, and then the techniques
presented here \ can be applied to $\mathcal{Z}_{\mathbb{C}}(J)$ and its
discrete version%
\[
\mathcal{Z}_{\mathbb{C}}^{\left(  l\right)  }(J)=\frac{%
{\displaystyle\int\limits_{\mathcal{L}_{\mathbb{R}}^{l}\left(  \mathbb{Q}%
_{p}^{N}\right)  }}
e^{\sqrt{-1}\left\{  E_{\text{int}}(\varphi)+E_{\text{source}}(\varphi
,J)\right\}  }d\mathbb{P}_{l}(\varphi)}{%
{\displaystyle\int\limits_{\mathcal{L}_{\mathbb{R}}^{l}\left(  \mathbb{Q}%
_{p}^{N}\right)  }}
e^{\sqrt{-1}\left\{  E_{0}(\varphi)+E_{\text{int}}(\varphi)\right\}
}d\mathbb{P}_{l}(\varphi)},\text{ }l\in N\smallsetminus\left\{  0\right\}  .
\]
In particular a version Theorem \ref{Theorem2} is valid for $\mathcal{Z}%
_{\mathbb{C}}(J)$.\ To explain the connection of these constructions with Wick
rotation, we rewrite (\ref{Eq_Char_func_2A}) as follows:%
\begin{equation}%
{\textstyle\int\limits_{\mathcal{L}_{\mathbb{R}}^{\prime}\left(
\mathbb{Q}_{p}^{N}\right)  }}
e^{\sqrt{-1}\lambda\langle W,f\rangle}d\mathbb{P}(W)=e^{-\frac{\left\vert
\lambda\right\vert ^{2}}{2}\mathbb{B}(f,f)}\text{,}\ \ f\in\mathcal{L}%
_{\mathbb{R}}\left(  \mathbb{Q}_{p}^{N}\right)  \text{, for }\lambda
\in\mathbb{C}\text{.}\label{Eq_Char_func_2}%
\end{equation}
This formula holds true in the case $\lambda\in\mathbb{R}$. The integral in
the right-hand side of (\ref{Eq_Char_func_2}) admits an entire analytic
continuation to the complex plane, see \cite[Proposition 2.4]{Hida et al}.
Furthermore, this fact is exactly the Analyticity Axiom (OS0) in the Euclidean
axiomatic quantum field presented in \cite[Chapter 6]{Glimm-Jaffe}.

A field $\varphi:\mathbb{Q}_{p}^{N}\rightarrow\mathbb{R}$ is a function from
the spacetime $\mathbb{Q}_{p}^{N}$ into $\mathbb{R}$ (the target space). We
perform a Wick rotation in the target space:%
\[%
\begin{array}
[c]{lll}%
\mathbb{R} & \rightarrow & \sqrt{-1}\mathbb{R}\\
&  & \\
\varphi & \rightarrow & \sqrt{-1}\varphi.
\end{array}
\]
Then
\[%
{\textstyle\int\limits_{\mathcal{L}_{\mathbb{R}}^{\prime}\left(
\mathbb{Q}_{p}^{N}\right)  }}
e^{\sqrt{-1}\langle T,\sqrt{-1}\varphi\rangle}d\mathbb{P}(T)=%
{\textstyle\int\limits_{\mathcal{L}_{\mathbb{R}}^{\prime}\left(
\mathbb{Q}_{p}^{N}\right)  }}
e^{\sqrt{-1}\langle\sqrt{-1}T,\varphi\rangle}d\mathbb{P}(T)=e^{-\frac{1}%
{2}\mathbb{B}(\varphi,\varphi)}\text{.}%
\]
Changing variables as $W=\sqrt{-1}T$, we get%
\[
e^{-\frac{1}{2}\mathbb{B}(\varphi,\varphi)}=%
{\textstyle\int\limits_{\sqrt{-1}\mathcal{L}_{\mathbb{R}}^{\prime}\left(
\mathbb{Q}_{p}^{N}\right)  }}
e^{\sqrt{-1}\langle W,\varphi\rangle}d\mathbb{P}^{\prime}(W).
\]
Therefore, $\mathbb{P}^{\prime}$ is a probability measure in $\sqrt
{-1}\mathcal{L}_{\mathbb{R}}^{\prime}\left(  \mathbb{Q}_{p}^{N}\right)  $ with
correlation functional $\mathbb{B}(\cdot,\cdot)$, that can be identified with
$\mathbb{P}$.

\end{document}